\documentclass[onecolumn,romanappendices,10pt]{IEEEtran}
\IEEEoverridecommandlockouts
\usepackage{cite}
\usepackage{amsmath,amssymb,amsfonts}

\usepackage{amsthm}
\usepackage{mathrsfs}
\usepackage{bm}
\usepackage{dsfont}
\usepackage[ruled,vlined]{algorithm2e}
\usepackage{mathtools}
\usepackage{graphicx}
\usepackage{textcomp}
\usepackage{multirow}
\usepackage{threeparttable}
\usepackage{enumerate}
\usepackage{mathtools}
\usepackage{hyperref}
\usepackage{xurl}
\usepackage{stfloats}
\usepackage{subfig}

\graphicspath{{figures/}}

\newtheorem{theorem}{Theorem}
\newtheorem{lemma}{Lemma}
\newtheorem{definition}{Definition}
\newtheorem{proposition}{Proposition}
\newtheorem{corollary}{Corollary}
\newtheorem{example}{Example}
\newtheorem{assumption}{Assumption}
\newtheorem{remark}{Remark}

\newcommand{\Acal}{{\mathcal A}}
\newcommand{\Ecal}{{\mathcal E}}
\newcommand{\Ncal}{{\mathcal N}}

\newcommand{\Ucal}{{\mathcal U}}

\newcommand{\Fcal}{{\mathcal F}}

\newcommand{\Kcal}{{\mathcal K}}

\newcommand{\Xcal}{{\mathcal X}}
\newcommand{\Ycal}{{\mathcal Y}}

\newcommand{\Wcal}{{\mathcal W}}
\newcommand{\Scal}{{\mathcal S}}

\newcommand{\Zcal}{{\mathcal Z}}

\newcommand{\EE}{{\mathbb{E}}}
\newcommand{\PP}{{\mathbb{P}}}
\newcommand{\RR}{{\mathbb{R}}}
\newcommand{\NN}{{\mathbb{N}}}

\DeclareMathOperator{\gen}{\overline{gen}}
\DeclareMathOperator{\Uniform}{Unif}
\DeclareMathOperator{\Bernoulli}{Bern}
\DeclareMathOperator{\Binomial}{Binom}

\newcommand\indep{ \perp \!\!\! \perp}

\allowdisplaybreaks

\begin{document}
	
	\title{On Unified and Sharpened CMI Bounds for Generalization Errors}
	\author{\IEEEauthorblockN{Yang Lu, Matthias Frey, Margreta Kuijper, and Jingge Zhu}
		\thanks{The authors are with the Department of Electrical and Electronic Engineering, The University of Melbourne,  Parkville, Victoria, Australia (email: yang.lu11@student.unimelb.edu.au, \{matthias.frey, mkuijper, jingge.zhu\}@unimelb.edu.au).}
		\thanks{This work was presented in part at the 2025 IEEE International Symposium on Information Theory (ISIT), Ann Arbor, Michigan, USA, June 2025.}}
	
	\maketitle
	\begin{abstract} 
		We present a new family of information-theoretic generalization bounds within the framework of conditional mutual information (CMI). Most of our results are established based on the leave-$m$-out (L$m$O) cross-validation error, with $m$ denoting the number of the hold-out supersamples. Under this setting, we propose a unified CMI-based bound, allowing to envelop and reproduce many known CMI-based bounds and also bridge the gap between the MI- and CMI-based bounds when $m$ tends to infinity. The proposed framework not only provides a unified description of the existing bounds but also develops new, sharper bounds. We show the benefits of the proposed bounds through several simple examples, where the existing results are either inapplicable or looser. Moreover, under the premise that the loss function is bounded, we tighten the CMI quantities involved in the proposed bounds by reducing the number of conditional terms, thereby enhancing the proposed framework. We show empirically that the resulting new bounds improve upon the previously known ones.
	\end{abstract}
	
	\section{Introduction}\label{sec:intro}
	A fundamental topic of statistical learning theory concerns how well a hypothesis inferred from some training data can generalize to new unseen data drawn from the same distribution as the training data. That is, how can we ensure that the learned hypothesis has captured knowledge that is reflective of the underlying data distribution rather than the training data? It is therefore crucial to understand the generalization error, which measures the discrepancy between the losses that the learning algorithm incurs on training data and unseen test data. Bounding generalization errors for learning algorithms not only provides guarantees whether the produced hypothesis overfits its training data, but also potentially inspires the design of improved learning algorithms. Several information-theoretic quantities such as the Kullback-Leibler (KL) divergence, mutual information (MI) and conditional mutual information (CMI), have emerged as promising tools to derive novel information-theoretic bounds. These bounds capture the joint effect of the learning algorithm and data distribution, beyond what conventional algorithm-independent bounds (e.g., those based on the methods of VC dimension and Rademacher complexity) can express \cite{shalev2014understanding}.
	
	Information-theoretic generalization bounds were first initiated in the seminal works \cite{russo2016controlling} and \cite{xu2017information}, where the generalization error is bounded in terms of the MI between the input training data and the output hypothesis. However, in some learning scenarios where the true generalization error is vanishing, the very first MI bound may yield a vacuous result \cite{bu2020tightening}. This happens particularly when the input data set and output hypothesis are both continuous random variables and the hypothesis is a deterministic function of the training data, thus leading to an unbounded MI quantity. To mitigate this limitation, the original MI bound was advanced and improved with various techniques, including the individual-sample or random-subset technique \cite{bu2020tightening,harutyunyan2021information,rodriguez2021random}, the disintegration technique \cite{negrea2019information} and the chaining method \cite{asadi2018chaining,zhou2023stochastic}. Amongst them, the authors of \cite{bu2020tightening} decomposed the generalization error and bounded each individual term to derive an individual mutual information (IMI) bound, which is proved to be tighter than the MI bound. Remarkably, another major step to improve the MI-based bounds was pushed forward by a series of CMI-based bounds \cite{steinke2020reasoning,haghifam2020sharpened,zhou2022individually,rammal2022leave,haghifam2022understanding}, where the training set is randomly selected from a larger supersample set and the generalization behavior is characterized by the CMI quantity between the output hypothesis and the training set selection, conditioned on the supersample set. Notably, based on different construction approaches of the supersample set, CMI-based bounds can be classified into the standard CMI bound and the leave-one-out CMI (LOO-CMI) bound. The former adopts the setting introduced initially in \cite{steinke2020reasoning}, where a Bernoulli variable is used to distinguish the training membership from a supersample pair. Furthermore, previous tightening techniques for the original MI bound are also applicable to the standard CMI bound \cite{haghifam2020sharpened,rodriguez2021random,zhou2022individually}. However, the $n$ unused samples in the standard supersample set of size $2n$ make it inefficient to compute the CMI quantity. This issue can be alleviated by the LOO-CMI bound, where the introduction of only one supersample brings benefits in terms of lower computational cost, alongside better interpretation and explicit connections with stability theory \cite{rammal2022leave,haghifam2022understanding}.
	
	Although the CMI-based bounds are always non-vacuous, they still suffer from the challenges arising in cumbersome estimation of MI involving the high-dimensional hypothesis, hindering the bounds from practical application. To alleviate this issue, one line of research employs the functional CMI (f-CMI) bound \cite{harutyunyan2021information}, where the high-dimensional hypothesis is replaced by the low-dimensional prediction with respect to the chosen supersample. As a result, the f-CMI bound is applicable to a wide range of methods, including non-parametric approaches, Bayesian algorithms and ensembling algorithms \cite{harutyunyan2021information}. Subsequently, the f-CMI bound was advanced in several follow-up works \cite{hellstrom2022new,wang2023tighter,dongexactly}, yielding strengthened bounds evaluated on the loss incurred on supersamples, including evaluated CMI (e-CMI) \cite{hellstrom2022new}, loss-difference based MI (ld-MI) \cite{wang2023tighter} and binarized-loss based MI (bl-MI) \cite{dongexactly}. In general, the data-processing inequality implies the following progression chain from loosest to tightest: f-CMI $\to$ e-CMI $\to$ ld-MI $\to$ bl-MI. Notably, the bl-MI bound provides an exact characterization of the generalization behavior and is therefore the tightest bound. It should be noted that the loss-based bounds are tailored for the specific loss function that appears in the bound. In contrast, conventional hypothesis-based bounds are compatible with general loss functions as they only require knowledge of the hypothesis itself. 
	
	Following the above-mentioned research thread and further motivated by understanding a deeper relationship between various generalization bounds, we establish a unified supersample framework that envelops most of the existing generalization bounds and further facilitates new promising bounds with improved tightness and effectiveness. The main contributions are summarized as follows.
	
	\begin{itemize}
		\item Towards the generalization analysis of learning algorithms, we present a new leave-$m$-out (L$m$O) supersample setting (Definition~\ref{def:p-lmo}). In this new setting, $n$ training samples are randomly chosen without replacement from $n+m$ supersamples. It holds that the expected L$m$O cross-validation error is equivalent to the classical expected generalization error (Lemma~\ref{lem:IP-lmo-gen-err}). Based on this observation, the generalization behavior can be characterized by a new CMI-based bound (Theorem~\ref{thm:IPCIMI}). By choosing appropriate parameters, the unified CMI bound is able to reproduce many known MI/CMI-based bounds in \cite{xu2017information,bu2020tightening,zhou2022individually,rammal2022leave,haghifam2022understanding}, providing a unified view and derivation of the existing results. 
		\item The proposed unified setting also allows us to derive new bounds (Corollary~\ref{cor:lofo-CMI}) that admit improved tightness. We show that under two learning scenarios, the proposed new bounds are tighter than the IMI \cite{bu2020tightening} and individually conditional individual mutual information (ICIMI) \cite{zhou2022individually} bounds.
		\item Inspired by \cite{dongexactly}, we extend the proposed unified CMI framework and develop new bounds (Theorem~\ref{thm:SIPCIMI}) that are equipped with improved CMI quantities, in which the mutual information is conditioned on only one supersample. Numerical results show that these singly-conditioned CMI-based bounds can be tighter than their original counterparts.
		\item Beyond the regime of conventional hypothesis-based bounds, we further develop the proposed unified supersample setting to establish tighter bounds based on a different discrepancy measure of the generalization error and easy-to-measure bounds based on the predictions for the supersamples (Theorem~\ref{thm:lb-gen-bound}). We also provide an exactly tight bound (Theorem~\ref{thm:tightest-gen-bound}) and illustrate that the generalization behavior can be precisely characterized in some scenarios.
	\end{itemize}
	
	The remainder of this paper is organized as follows: in Section~\ref{sec:preliminary}, we introduce some necessary notations and present the problem formulation. Section~\ref{sec:review} reviews several relevant information-theoretic generalization bounds which are further evaluated through a specific learning example. Section~\ref{sec:unified-CMI} introduces a new supersample setting to facilitate a unified CMI-based bound. The special cases of this unified bound, including various known and new bounds, are further presented, evaluated and compared via both analytical and numerical results. Next, Section~\ref{sec:SCMI-bounds} employs the same supersample setting and proposes a series of parallel bounds equipped with improved CMI quantities, potentially enhancing the results in Section~\ref{sec:unified-CMI}. In Section~\ref{sec:extensions}, we extend our framework in various directions to derive tighter, easy-to-measure bounds. Finally, Section~\ref{sec:conclusion} concludes this paper. Proofs, additional analysis and calculation details can be found in Appendix~\ref{app:lemmas} to \ref{app:calculations}.
	
	\section{Preliminaries}\label{sec:preliminary}
	
	\subsection{Notations}
	We use capital letters (e.g., $X$) for random variables, lower-case letters (e.g., $x$) for their realizations, and calligraphic letters (e.g., $\Xcal$) for their ranges. Let $P_X$ denote the distribution of a random variable $X$ and $P_{X|Y}$ denote the conditional distribution of $X$ given $Y$. When conditioning on a specific realization $Y=y$, we use $P_{X|Y =y}$ or simply $P_{X|y}$, and accordingly the expectation over $X \sim P_{X|Y =y}$ is denoted as $\EE_{X|Y =y}$ or $\EE_{X|y}$. $X \indep Y$ means $X$ and $Y$ are independent. For brevity, we use $X \indep Y$ to indicate that $X$ and $Y$ are independent random variables. Besides, let $[n]$ be the set $\{1, 2, \ldots, n\}$ and let $\NN$ denote the set of natural numbers $\{1,2,\ldots\}$. Moreover, we define the minimum function $\min(a, b)$ and the maximum function $\max(a, b)$ which return the smaller and larger of the two numbers $a$ and $b$, respectively. Let $H(p) := -p \log p - (1 - p) \log (1 - p)$ denote the binary entropy function. Given a uniform variable $X \in S$ that takes values from $S$ with equal probability, we let $\Uniform(S)$ denote the distribution of $X$. Given a Bernoulli variable $X \in \{0, 1\}$ and $P(X = 1) = p$, we let $\Bernoulli(p)$ denote its distribution. The number of successes in $n$ i.i.d. Bernoulli trials is defined as the Binomial variable. Then, given a Binomial variable $Y = \sum_{i=1}^{n} X_i$, where $X_i \overset{i.i.d.}{\sim} \Bernoulli(p)$, we let $\Binomial(n, p)$ denote the distribution of $Y$. Moreover, the Bernstein polynomials are defined as $B_{n,x}(p) := \begin{psmallmatrix} n \\ x \end{psmallmatrix} p^x (1-p)^{n-x}$, where $\begin{psmallmatrix} n \\ i \end{psmallmatrix}$ is a binomial coefficient. Given probability measures $P$ and $Q$, the total variation between $P$ and $Q$ is defined as $D_{\textnormal{TV}}(P \| Q) := \frac{1}{2} \int | \frac{d P}{d Q} - 1 | d Q$, and the KL divergence from $P$ to $Q$ is defined as $D_{\textnormal{KL}}(P \| Q) := \int \log \frac{d P}{d Q} d P$, where $\frac{d P}{d Q}$ is the Radon-Nikodym derivative of $P$ with respect to $Q$. Let $\otimes$ form product measures. The Shannon mutual information between two random variables $X$ and $Y$ is $I(X; Y) := D\left( P_{XY} \| P_X \otimes P_Y \right)$. Besides, we define the disintegrated mutual information as $I^z(X; Y) := D\left( P_{XY | Z = z} \| P_{X | Z = z} \otimes P_{Y | Z = z} \right)$, which in expectation is the conditional mutual information $I(X; Y | Z) := \EE_Z \left[ I^Z(X; Y)\right]$. Throughout this paper, all the logarithms are to the base $e$, and thus all the information-theoretic quantities are measured in nats.
	
	We define the following.
	\begin{definition}
		For a random variable $X$, the cumulant generating function (CGF) of $X$ is defined as 
		\begin{align*}
			\psi_X (\lambda) := \log \EE \left[ e^{\lambda(X - \EE [X])} \right].
		\end{align*}
	\end{definition}
	
	\begin{definition}\label{def:con-CGF}
		For random variables $X$ and $Y$, the CGF of $X$ conditioned on $Y=y$ is defined as
		\begin{align*}
			\psi_{X|Y=y} (\lambda) := \log \EE_{X|Y=y} \left[ e^{\lambda(X - \EE_{X|Y=y} [X])} \right].
		\end{align*}
	\end{definition}
	
	\subsection{Problem Formulation} \label{subsec:pbl_setup}
	Let $\Zcal$ denote the data domain, and $\Wcal$ the hypothesis domain. Let $Z_{[n]} := \{Z_i\}_{i=1}^n \in \Zcal^n$ be a collection composed of $n$ i.i.d. samples drawn from some unknown distribution $\mu$, where $n \in \NN$. Taking the training set $Z_{[n]}$ as input, a possibly stochastic learning algorithm $\Acal := \Zcal^n \to \Wcal$ selects a hypothesis $W \in \Wcal$ as the output, i.e., $W := \Acal(Z_{[n]})$. This algorithm induces a conditional distribution $P_{W | Z_{[n]}}$ which maps from $\Zcal^n$ to $\Wcal$. We denote the joint distribution of the training set $Z_{[n]}$ and the hypothesis $W$ by $P_{W,Z_{[n]}} := \mu^{\otimes n} \otimes P_{W|Z_{[n]}}$, and the marginal distribution of $W$ by $P_W$. Given a loss function $\ell : \Wcal \times \Zcal \to \RR^+$, the learner would find a hypothesis $w \in \Wcal$ that minimizes the following population loss:
	\begin{equation}
		\begin{aligned}
			L_\mu(w) = \EE_{Z \sim \mu} \left[ \ell(w, Z) \right],
		\end{aligned}
		\label{eq:def-pop-loss}
	\end{equation}
	and the quantity $L_\mu := \EE_W [L_\mu(W)]$ is then defined as the expected population risk. However, it is typically infeasible to compute $L_\mu(w)$ directly since the underlying distribution $\mu$ is usually unknown. One alternate approach is to minimize the empirical loss defined by
	\begin{equation}
		\begin{aligned}
			L_{Z_{[n]}}(w) := \frac{1}{n} \sum_{i=1}^{n} \ell(w, Z_i).
		\end{aligned}
		\label{eq:def-emp-loss}
	\end{equation}
	We also define $\widehat{L}_n := \EE [L_{Z_{[n]}}(W)]$ as the expected empirical risk, where the expectation is taken over the joint distribution $P_{W,Z_{[n]}}$. To evaluate how the hypothesis selected by the learner performs on unseen samples compared to those that are used for training and thus assess the overfitting effect, the expected generalization error is introduced as
	\begin{equation}
		\begin{aligned}
			\gen := L_\mu - \widehat{L}_n.
		\end{aligned}
		\label{eq:def-gen}
	\end{equation}
	In the sequel, our main efforts focus on bounding the expected generalization error $\gen$.
	
	\section{Review of Related Results}\label{sec:review}
	\begin{figure}[t]
		\centering
		\includegraphics[width=0.95\textwidth]{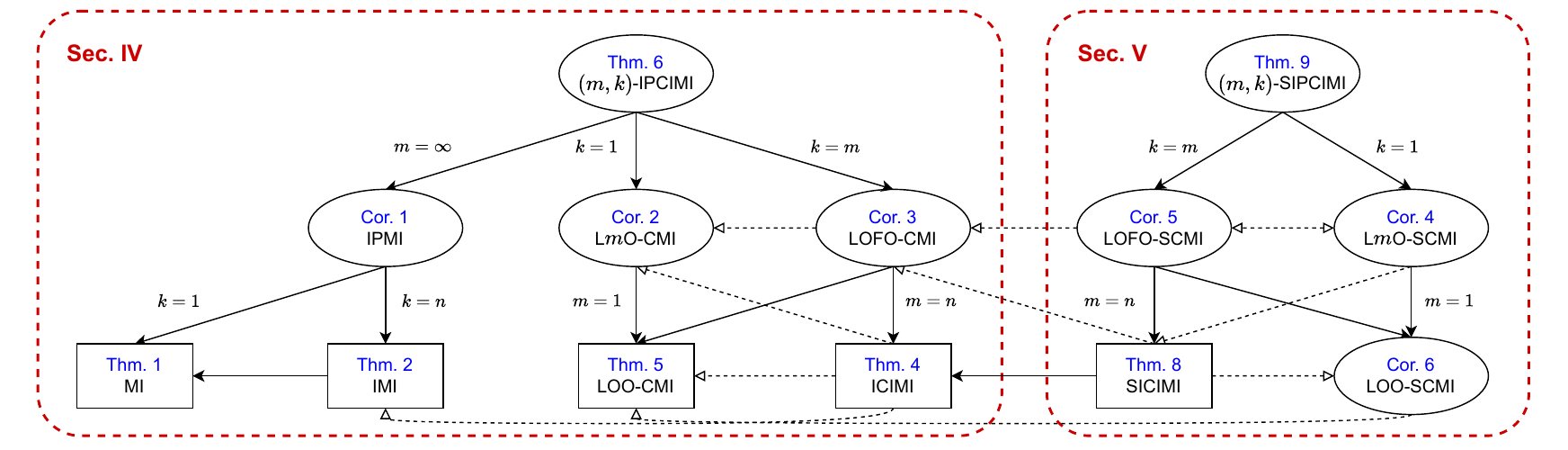}
		\caption{An overview of the proposed framework that involves all bounds of interest. In this figure, the bounds from the prior literature are marked in rectangles, whereas the bounds proposed in this paper are marked in ellipses. For any two bounds A and B such that A points to B, we use 1) a filled arrow with a solid line to indicate that A is tighter than B, where the label over the arrow (if any) represents the parameter setting with which B can be specialized from A; 2) a hollow arrow with a dashed line to indicate that A is tighter than B, as shown by at least one of the Examples~\ref{eg:Bernoulli}, \ref{eg:Gaussian} and \ref{eg:finite_Gaussian} studied in Sections~\ref{sec:unified-CMI} and \ref{sec:SCMI-bounds}. We further note that Theorem~\ref{thm:bld-IPCIMI} is a special version of Theorem~\ref{thm:IPCIMI} with a bounded loss difference assumption and hence omitted from the figure.}
		\label{fig:overview}
	\end{figure}
	
	In this section, we briefly review some related results on information-theoretic generalization bounds. As an overview, Figure~\ref{fig:overview} provides a characterization of the relationships between all the considered bounds, and the thorough discussions are deferred to Sections~\ref{sec:unified-CMI} and \ref{sec:SCMI-bounds}.
	
	\subsection{MI-Based Bounds}
	Originating in \cite{russo2016controlling}, MI is used to upper bound the expected generalization error. This result was further studied and extended in \cite{xu2017information}, leading to the following standard MI-based generalization bound.
	\begin{theorem}[MI Bound \cite{xu2017information}] \label{thm:MI} 
		If $\ell(w, Z)$, where $Z \sim \mu$, is $\sigma$-sub-Gaussian for all $w \in \Wcal$, then
		\begin{equation}
			\begin{aligned}
				\gen \leq \sqrt{\frac{2\sigma^2}{n} I(W; Z_{[n]})}. 
			\end{aligned}
		\end{equation}
		In particular, if $\ell(w, z) \in [0, 1]$ for all $w \in \Wcal$ and $z \in \Zcal$, then
		\begin{equation}
			\begin{aligned}
				\gen \leq \sqrt{\frac{1}{2n} I(W; Z_{[n]})}. 
			\end{aligned}
			\label{eq:MI}
		\end{equation}
	\end{theorem}
	Instead of bounding the whole generalization error \eqref{eq:def-gen} directly, as a proof technique, the authors of \cite{bu2020tightening} chose to bound each individual difference $\EE \left[ L_{\mu}(W) - \ell(W,Z_i) \right]$. The application of such an individual-sample technique resulted in the following individual mutual information (IMI) bound, which tightened the result in Theorem~\ref{thm:MI}. 
	\begin{theorem}[IMI Bound \cite{bu2020tightening}] \label{thm:IMI} 
		Let $\tilde{W}$ be an independent copy of $W$ and $\hat{Z}$ an independent copy of $Z$ such that $P_{\tilde{W},\hat{Z}} = P_W \otimes \mu$. Suppose that $\psi_{\ell(\tilde{W}, \hat{Z})}(\lambda) \leq \Psi_{-}(-\lambda)$ for $\lambda \leq 0$ under distribution $P_{\tilde{W},\hat{Z}}$, then
		\begin{equation}
			\begin{aligned}
				\gen \leq \frac{1}{n} \sum_{i=1}^n \inf_{\lambda > 0} \frac{1}{\lambda} \left( I(W; Z_i) + \Psi_{-}(-\lambda) \right).
			\end{aligned}\label{eq:gen-IMI}
		\end{equation}
		In particular, if $\ell(w, z) \in [0, 1]$ for all $w \in \Wcal$ and $z \in \Zcal$, then
		\begin{equation}
			\begin{aligned}
				\gen \leq \frac{1}{n} \sum_{i=1}^{n} \sqrt{\frac{1}{2} I(W; Z_i)}. 
			\end{aligned}
			\label{eq:IMI}
		\end{equation}
	\end{theorem}
	
	\subsection{CMI-Based Bounds}
	The aforementioned MI-based bounds may yield vacuous results, particularly when the hypothesis is a deterministic function of the training data. To alleviate this issue, a new supersample setting was first introduced in \cite{steinke2020reasoning} to obtain non-vacuous and tightened generalization bounds. We note that all of the supersample settings introduced in this paper serve as purely analytical tools without altering the learning problem stated in Section~\ref{sec:preliminary}.
	\begin{definition}[Standard Supersample Setting \cite{steinke2020reasoning}] \label{def:standard-supersample}
		Define the supersample set $\hat{Z}_{[2n]} = \{\hat{Z}_i\}_{i=1}^{2n} \in \Zcal^{2n}$ as a collection of $2n$ samples drawn i.i.d. from $\mu$. Let $R_{[n]} = \{R_i\}_{i=1}^{n} \in \{0,1\}^n$ be an independent membership sequence, where each $R_i$ is drawn i.i.d. from $\Bernoulli(\frac{1}{2})$ and used to divide each pair $(\hat{Z}_{i}, \hat{Z}_{i+n})$ into a training sample and a test sample\footnote{We note that the so-called ``test'' sample does not refer to the test data used in practice. The test samples are used as an analytical tool to define the population loss.}. Specifically, $n$ samples are selected to formulate the training set $\hat{Z}_{R_{[n]}} := \{\hat{Z}_{i + R_i n}\}_{i=1}^{n}$, while the remaining half of the samples constitute the test set, and thus the output hypothesis is characterized as $W^{\textrm{std}} \sim P_{W|Z_{[n]}}(\cdot | \hat{Z}_{R_{[n]}})$.
	\end{definition}
	It is shown in the proof of \cite[Theorem 5.1]{steinke2020reasoning} that $(W^{\textrm{std}}, \hat{Z}_{R_{[n]}})$ follows the same distribution as $(W,Z_{[n]})$ defined in Section~\ref{subsec:pbl_setup}. In particular, they have the same (expected) generalization error. We therefore write, whenever it is clever from context, $W$ instead of $W^{\textrm{std}}$ and use $\gen$ to denote the expected generalization error of both $W$ and $W^{\textrm{std}}$. Under this setting, a new generalization bound involving the CMI quantity was established.
	\begin{theorem}[Standard CMI Bound \cite{steinke2020reasoning}] \label{thm:CMI}
		Consider the standard supersample setting in Definition~\ref{def:standard-supersample}. If $\ell(w, z) \in [0, 1]$ for all $w \in \Wcal$ and $z \in \Zcal$, then
		\begin{equation}
			\begin{aligned}
				\gen \leq \sqrt{\frac{2}{n} I\left( W; R_{[n]} | \hat{Z}_{[2n]} \right)}.
			\end{aligned}
			\label{eq:CMI}
		\end{equation}
	\end{theorem}
	By nature, the CMI quantity is always bounded, while the MI quantities can be unbounded particularly when $W$ and $Z_i$ are both continuous. With the applications of the individual-sample technique \cite{bu2020tightening} and the disintegration approach \cite{negrea2019information}, the result in \eqref{eq:CMI} can be strengthened to a tighter bound, namely individually conditional individual mutual information (ICIMI) bound, which is shown as below.
	\begin{theorem}[ICIMI Bound \cite{zhou2022individually}] \label{thm:ICIMI}
		Consider the standard supersample setting in Definition~\ref{def:standard-supersample}. Let $(\tilde{W}, \tilde{R}_i)$ be a decoupled pair of $(W, R_i)$ conditioned on $(\hat{Z}_{i}, \hat{Z}_{i+n})$, i.e., $P_{\tilde{W},\tilde{R}_i| \hat{Z}_{i}, \hat{Z}_{i+n}} = P_{W | \hat{Z}_{i}, \hat{Z}_{i+n}} \otimes P_{R_i | \hat{Z}_{i}, \hat{Z}_{i+n}}$, where $P_{R_i | \hat{Z}_{i}, \hat{Z}_{i+n}} = P_{R_i}$ as $R_i \indep (\hat{Z}_{i}, \hat{Z}_{i+n})$. Then,
		\begin{align}
			\gen &\leq \frac{1}{n} \sum_{i=1}^n \EE \left[ \inf_{\lambda > 0} \frac{1}{\lambda} \left( I^{\hat{Z}_i, \hat{Z}_{i+n}}\left( W; R_i \right) + \psi_{\tilde{R}_i(\ell(\tilde{W}, \hat{Z}_i) - \ell(\tilde{W}, \hat{Z}_{i+n})) | \hat{Z}_{i}, \hat{Z}_{i+n}}(\lambda) \right) \right] \label{eq:gen-dis-ICIMI}\\
			&\leq \frac{1}{n} \sum_{i=1}^n \inf_{\lambda > 0} \frac{1}{\lambda} \left( I\left( W; R_i | \hat{Z}_i, \hat{Z}_{i+n} \right) + \EE \left[ \psi_{\tilde{R}_i(\ell(\tilde{W}, \hat{Z}_i) - \ell(\tilde{W}, \hat{Z}_{i+n})) | \hat{Z}_{i}, \hat{Z}_{i+n}}(\lambda) \right] \right). \label{eq:gen-ICIMI}
		\end{align}
		In particular, if $\ell(w, z) \in [0, 1]$ for all $w \in \Wcal$ and $z \in \Zcal$, then
		\begin{align}
			\gen &\leq \frac{1}{n} \sum_{i = 1}^n \sqrt{2 I\left( W; R_i | \hat{Z}_i, \hat{Z}_{i+n} \right)}.
			\label{eq:ICIMI}
		\end{align}
	\end{theorem}
	While the above-mentioned CMI-based bounds are established on the standard supersample setting, it was discovered in concurrent works \cite{haghifam2022understanding,rammal2022leave} that the following leave-one-out (LOO) supersample setting enables a different approach to derive new generalization bounds.
	\begin{definition}[LOO Supersample Setting \cite{haghifam2022understanding,rammal2022leave}] \label{def:loo-supersample}
		Define a supersample set $\hat{Z}_{[n+1]} = \{\hat{Z}_i\}_{i=1}^{n+1} \in \Zcal^{n+1}$ as a collection of $n+1$ samples drawn i.i.d. from $\mu$. Let $U \in [n+1]$ be an independent random variable drawn from $\Uniform([n+1])$. The $U$-th sample $\hat{Z}_U$ is left out as the test sample, while the remaining $n$ samples are used for training, and thus the output hypothesis is characterized as $W^{\textrm{LOO}} \sim P_{W|Z_{[n]}}(\cdot | \hat{Z}_{[n+1] \backslash U})$.
	\end{definition}
	It is also proved in \cite[Lemma~7.2]{rammal2022leave} that generalization error incurred by $W^{\textrm{LOO}}$ and $\hat{Z}_{[n+1] \backslash U}$ is identical to $\gen$ defined in \eqref{eq:def-gen}, and thus we still use $W$ and $\gen$ instead as the convention. An LOO-CMI bound was built based on this setting.
	\begin{theorem}[LOO-CMI Bound \cite{haghifam2022understanding,rammal2022leave}] \label{thm:LOO-CMI}
		Consider the LOO supersample setting in Definition~\ref{def:loo-supersample}. If $\ell(w, z) \in [0, 1]$ for all $w \in \Wcal$ and $z \in \Zcal$, then
		\begin{align}
			\gen &\leq \frac{n+1}{n} \sqrt{\frac{1}{2} I\left( W; U | \hat{Z}_{[n+1]} \right)}.
			\label{eq:LOO-CMI}
		\end{align}
	\end{theorem}
	
	Stemming from the new supersample construction, the LOO-CMI bound is more friendly to computation as compared to other CMI-based bounds. To compute the LOO-CMI quantity, one only needs to average $n+1$ possible values of $U$, while a number of $2^n$ values of $R_{[n]}$ should be considered for the former CMI quantities. However, the computational efficiency of the LOO-CMI bound comes with the sacrifice of its tightness, which we will show through a specific example in the sequel.

	\subsection{Evaluation of the Existing Bounds}
	In this subsection, we evaluate the tightness of the existing bounds through a specific example, unveiling the important role that the supersample construction plays in developing generalization bounds. For all the considered bounds, the calculations of the involved information measure quantities and decay orders are detailed in Appendix~\ref{app:calculations}.
	
	\begin{example}[A Bernoulli Example] \label{eg:Bernoulli}
		Suppose that all data are sampled from a Bernoulli distribution $\Bernoulli(p)$ with unknown $p \in (0, 1)$. Let $W$ be the output of the empirical risk minimization (ERM) algorithm that gives the average of the training samples, i.e.,
		\begin{align}
			W = \frac{1}{n} \sum_{i=1}^{n} Z_i. 
		\end{align}
		Consider a quadratic loss function $\ell(w, z) = (w - z)^2$, then the true generalization error can be exactly calculated as $\gen = \frac{2p (1-p)}{n}$ (c.f. Appendix~\ref{subapp:true-gen-err}), which is of order $O(\frac{1}{n})$.
	\end{example}
	The loss function clearly satisfies the boundedness assumption. The calculations of the information measure quantities are detailed in Appendix~\ref{app:calculations}. We first evaluate the IMI and ICIMI bounds under Example~\ref{eg:Bernoulli}, with the calculation details postponed to Appendices~\ref{subapp:cal_IMI} and \ref{subapp:cal_ICIMI}.
	\begin{proposition} \label{pro:ICIMI}
		Under Example~\ref{eg:Bernoulli}, the IMI bound in \eqref{eq:IMI} and ICIMI bound are both of order $O \left( \frac{1}{\sqrt{n}} \right)$.
	\end{proposition}
	This proposition aligns with the result in the Gaussian location problem studied in \cite{zhou2022individually}. We next evaluate the LOO-CMI bound and arrive at the following proposition. 
	\begin{proposition} \label{pro:LOO-CMI}
		Under Example~\ref{eg:Bernoulli}, the LOO-CMI bound in \eqref{eq:LOO-CMI} converges to a constant, i.e.,
		\begin{equation}
			\begin{aligned}
				\textnormal{RHS of \eqref{eq:LOO-CMI}} = \sqrt{\frac{1}{2} H(p)} + O\left( \frac{1}{\sqrt{n}} \right).
			\end{aligned}
		\end{equation}
	\end{proposition}
	The calculations are detailed in Appendix~\ref{subapp:cal_LOO}. The observations from Propositions~\ref{pro:ICIMI} and \ref{pro:LOO-CMI} demonstrate that the LOO-CMI bound can be order-wise looser than the other considered bounds, including the MI-based bounds. It seems like the adoption of the LOO supersample setting may impair the tightness of the resulting bound. Motivated by this, we introduce in the next sections a new supersample framework, from which all of the aforementioned bounds can be recovered and further improved.
	
	\section{A Unified Framework for CMI-Based Bounds}\label{sec:unified-CMI}
	In this section, we construct the supersample set by a new approach to yield a more general CMI-based bound. We then show that some special cases of the unified CMI-bound can either recover the existing results or provide improved bounds. Finally, we conduct comparisons of all the considered bounds for the case of Example~\ref{eg:Bernoulli}, demonstrating the benefits of the proposed bounds.
	
	\subsection{A Unified CMI-Based Bound}
	Without loss of generality, we assume that the number of training samples is $n$ throughout this paper. Now consider a more general setting where $m$ samples are removed from a supersample set of size $n+m$, where $m \in \NN$. 
	
	\begin{definition}[Partitioned Leave-$m$-Out Supersample Setting] \label{def:p-lmo}
		Define a supersample set $\hat{Z}_{[n+m]} := \{\hat{Z}_i\}_{i=1}^{n+m} \in \Zcal^{n+m}$ as a collection of $n+m$ samples drawn i.i.d. from $\mu$. Let $k \in \Kcal$ be an integer that divides $n$ and $m$, where $\Kcal := \{ k \in \NN: n \mod k = 0, m \mod k =0 \}$. Then, each supersample is mapped to $\hat{Z}^{(i)}_j := \hat{Z}_{\frac{n+m}{k}(i-1) + j}$ in an ascending order, where $i \in [k]$ and $j \in [\frac{n+m}{k}]$. Let $\hat{Z}^{(i)}_{[\frac{n+m}{k}]} := \{ \hat{Z}^{(i)}_j \}_{j=1}^{\frac{n+m}{k}}$ denote the $i$-th block of size $\frac{n+m}{k}$, and $\hat{Z}_{[n+m]} = \{\hat{Z}^{(i)}_{[\frac{n+m}{k}]}\}_{i=1}^{k}$ a partition of $\hat{Z}_{[n+m]}$ composed of $k$ disjoint blocks. For the $i$-th block, we further define $U^{(i)}_{[\frac{n}{k}]} := \{U^{(i)}_{j}\}_{j=1}^{\frac{n}{k}}$ as the corresponding training membership vector and $\bar{U}^{(i)}_{[\frac{m}{k}]} := [\frac{n+m}{k}] \backslash U^{(i)}_{[\frac{n}{k}]}$ the test membership vector, where $U^{(i)}_{[\frac{n}{k}]}$ is a subset of size $\frac{n}{k}$ chosen uniformly from $[\frac{n+m}{k}]$ (i.e., $U^{(i)}_{[\frac{n}{k}]}$ is distributed uniformly among all $\big(\begin{smallmatrix} \frac{n+m}{k} \\ \frac{n}{k} \end{smallmatrix}\big)$ possibilities), such that $\frac{n}{k}$ samples are selected as the training subset $Z^{(i)}_{[\frac{n}{k}]} := \hat{Z}^{(i)}_{U^{(i)}_{[\frac{n}{k}]}} = \{ \hat{Z}^{(i)}_{U^{(i)}_{j}} \}_{j=1}^{\frac{n}{k}}$ and the remaining $\frac{m}{k}$ supersamples $\hat{Z}^{(i)}_{\bar{U}^{(i)}_{[\frac{m}{k}]}}$ constitutes the test subset. The training membership of the partition is then defined by the collection of $U^{(i)}_{[\frac{n}{k}]}$'s, i.e., $U_{[n]} := \{U^{(i)}_{[\frac{n}{k}]}\}_{i=1}^{k}$, and the test membership $\bar{U}_{[m]}$ is defined analogously. Accordingly, the whole training set of size $n$ is constructed as $\{ Z^{(i)}_{[\frac{n}{k}]} \}_{i=1}^{k}$ (or equivalently $\hat{Z}_{U_{[n]}} = \{ \hat{Z}^{(i)}_{U^{(i)}_{[\frac{n}{k}]}} \}_{i=1}^k$) and the remaining $m$ supersamples $\hat{Z}_{\bar{U}_{[m]}}$ constitutes the test set. Thus, the output hypothesis is characterized as $W^{\textnormal{L$m$O}} \sim P_{W|Z_{[n]}} (\cdot | \hat{Z}_{U_{[n]}})$. 
	\end{definition}
	
	\begin{figure}[t]
		\centering
		\includegraphics[width=0.8\textwidth]{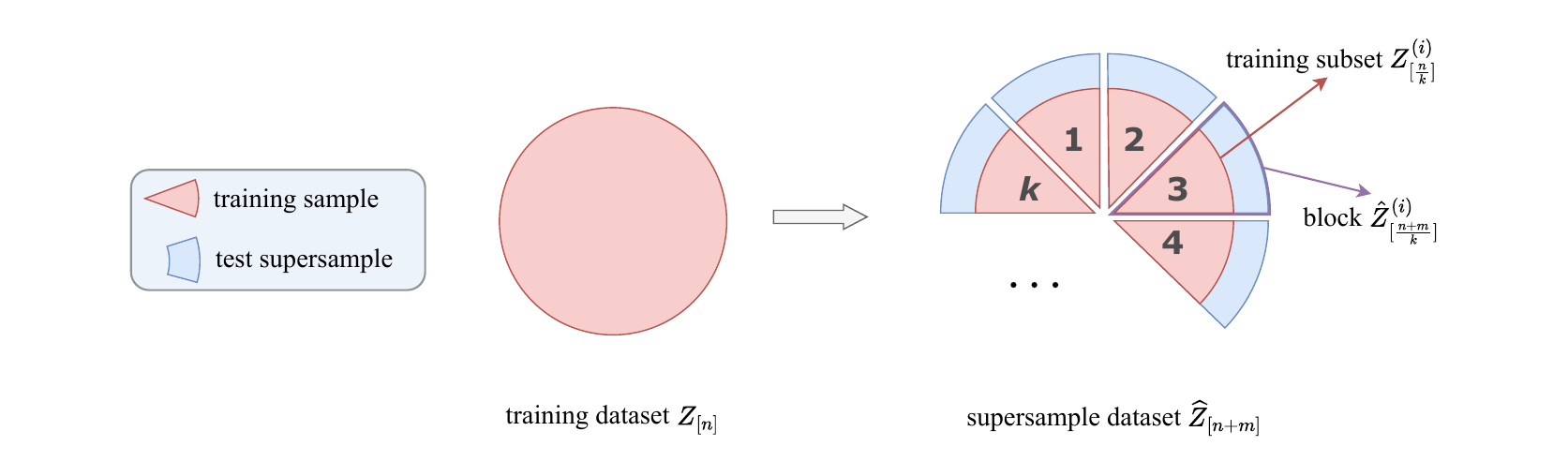}
		\caption{A brief illustration of the proposed partitioned L$m$O supersample setting. In this figure, the area of a shape denotes the number of its data samples, e.g., the training dataset $Z_{[n]}$ (left) has an area of $n$. Composed of $Z_{[n]}$ and some extra test supersamples with a total area of $m$, the new supersample set $\hat{Z}_{[n+m]}$ (right) is equally divided into $k$ blocks. }
		\label{fig:intro_supersample}
	\end{figure}
	\begin{figure}[t]
		\centering
		\includegraphics[width=0.95\textwidth]{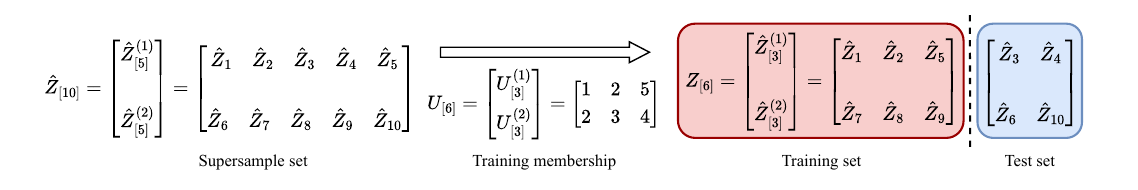}
		\caption{An example of the partitioned sumpersample set, where $n = 6$, $m = 4$ and $k = 2$.}
		\label{fig:general-supersample-set}
	\end{figure}
	
	\begin{remark}
		For the benefit of analytical tractability, the setting in Definition~\ref{def:p-lmo} divides the supersample set $\hat{Z}_{[n+m]}$ into equal-sized blocks. In essence, this setting can be extended by adopting other partitions of $\hat{Z}_{[n+m]}$ with variable-sized blocks. In this case, the blocks remain disjoint by the definition of partition, and thus most of our results can be readily extended.
	\end{remark}
	
	Adopting the convention, we still use $W$ instead of $W^{\textnormal{L$m$O}}$ whenever it is clear from context. For illustrative purposes, the proposed supersample setting is briefly depicted in Fig.~\ref{fig:intro_supersample}. In addition, a specific example of the proposed setting is provided in Fig.~\ref{fig:general-supersample-set}. 
	
	We next argue that a certain error incurred on the introduced supersamples is equivalent to $\gen$ defined in $\eqref{eq:def-gen}$, akin to the cases in other supersample settings. Specifically, we define the following L$m$O cross-validation error to measure the difference between the training and test losses:
	\begin{equation}
		\begin{aligned}
			\Ecal_{\textnormal{L$m$O}} \left( w, \hat{z}_{[n+m]}, u_{[n]} \right) &:= L_{\hat{z}_{\bar{u}_{[m]}}}(w) - L_{\hat{z}_{u_{[n]}}}(w) \\
			&= \sum_{i=1}^{k} \Bigg( \frac{1}{m} \sum_{j \in \bar{u}^{(i)}_{[\frac{m}{k}]}} \ell(w, \hat{z}^{(i)}_j) - \frac{1}{n} \sum_{j \in u^{(i)}_{[\frac{n}{k}]}} \ell(w, \hat{z}^{(i)}_j) \Bigg).
		\end{aligned}\label{eq:lmo-cv-err}
	\end{equation}
	Let define $\Ecal^{(i)}\left( w, \hat{z}^{(i)}_{[\frac{n+m}{k}]}, u^{(i)}_{[\frac{n}{k}]} \right) := L_{\hat{z}_{\bar{u}^{(i)}_{[\frac{m}{k}]}}}(w) - L_{\hat{z}_{u^{(i)}_{[\frac{n}{k}]}}}(w)$ as the cross-validation error with respect to the $i$-th block, then \eqref{eq:lmo-cv-err} can be rewritten as
	\begin{equation}
		\begin{aligned}
			\Ecal_{\textnormal{L$m$O}} \left( w, \hat{z}_{[n+m]}, u_{[n]} \right) = \frac{1}{k} \sum_{i=1}^{k} \Ecal^{(i)}\left( w, \hat{z}^{(i)}_{[\frac{n+m}{k}]}, u^{(i)}_{[\frac{n}{k}]} \right).
		\end{aligned}
	\end{equation}
	Hence, the following lemma shows that the L$m$O cross-validation error, when averaged over the joint distribution of the random variables involved, provides an unbiased estimate of the expected generalization error in \eqref{eq:def-gen}.
	\begin{lemma} \label{lem:IP-lmo-gen-err}
		Consider the partitioned L$m$O supersample setting in Definition~\ref{def:p-lmo}. The expected generalization error defined in \eqref{eq:def-gen} can be rewritten as
		\begin{equation}
			\begin{aligned}
				\gen &= \EE_{W, \hat{Z}_{[n+m]}, U_{[n]}} \left[ \Ecal_{\textnormal{L$m$O}}( W, \hat{Z}_{[n+m]}, U_{[n]} ) \right] \\
				&= \frac{1}{k} \sum_{i=1}^{k} \EE_{W, \hat{Z}^{(i)}_{[\frac{n+m}{k}]}, U^{(i)}_{[\frac{n}{k}]}} \left[ \Ecal^{(i)}\left( W, \hat{Z}^{(i)}_{[\frac{n+m}{k}]}, U^{(i)}_{[\frac{n}{k}]} \right) \right].
			\end{aligned}\label{eq:IP-lmo-gen-err}
		\end{equation}
	\end{lemma}
	
	The proof is given in Appendix~\ref{appsub:lemma1}. Based on this lemma, we are now in a position to derive generalization bounds with new information quantities. The following theorem holds by the linearity of expectation and marginalization, resulting in a new CMI-based bound that is conditioned on an individual block. The complete proof can be found in Appendix~\ref{appsub:theorem6}. We refer to the new bound as the individually partitioned conditional individual mutual information (IPCIMI) bound, which is shown as follows. 
	\begin{theorem}[IPCIMI Bounds] \label{thm:IPCIMI}
		Consider the partitioned L$m$O supersample setting in Definition~\ref{def:p-lmo}. Let $(\tilde{W}, \tilde{U}^{(i)}_{[\frac{n}{k}]})$ be a decoupled pair of $(W, U^{(i)}_{[\frac{n}{k}]})$ conditioned on $\hat{Z}^{(i)}_{[\frac{n+m}{k}]}$, i.e., $P_{\tilde{W},\tilde{U}^{(i)}_{[\frac{n}{k}]} | \hat{Z}^{(i)}_{[\frac{n+m}{k}]}} = P_{W | \hat{Z}^{(i)}_{[\frac{n+m}{k}]}} \otimes P_{U^{(i)}_{[\frac{n}{k}]} | \hat{Z}^{(i)}_{[\frac{n+m}{k}]}}$, where $P_{U^{(i)}_{[\frac{n}{k}]} | \hat{Z}^{(i)}_{[\frac{n+m}{k}]}} = P_{U^{(i)}_{[\frac{n}{k}]}}$ as $U^{(i)}_{[\frac{n}{k}]} \indep \hat{Z}^{(i)}_{[\frac{n+m}{k}]}$. Then,
		\begin{align}
			\gen &\leq \frac{1}{k} \sum_{i=1}^{k} \EE \left[ \inf_{\lambda > 0} \frac{1}{\lambda} \left( I^{\hat{Z}^{(i)}_{[\frac{n+m}{k}]}}\left( W; U^{(i)}_{[\frac{n}{k}]} \right) + \psi_{\mathcal{E}^{(i)}(\tilde{W}, \hat{Z}^{(i)}_{[\frac{n+m}{k}]}, \tilde{U}^{(i)}_{[\frac{n}{k}]}) | \hat{Z}^{(i)}_{[\frac{n+m}{k}]}}(\lambda) \right) \right] \label{eq:gen-dis-IPCIMI}\\
			&\leq \frac{1}{k} \sum_{i=1}^{k} \inf_{\lambda > 0} \frac{1}{\lambda} \left(I\left( W; U^{(i)}_{[\frac{n}{k}]} | \hat{Z}^{(i)}_{[\frac{n+m}{k}]} \right) + \EE \left[ \psi_{\mathcal{E}^{(i)}(\tilde{W}, \hat{Z}^{(i)}_{[\frac{n+m}{k}]}, \tilde{U}^{(i)}_{[\frac{n}{k}]}) | \hat{Z}^{(i)}_{[\frac{n+m}{k}]}}(\lambda) \right] \right). \label{eq:gen-IPCIMI}
		\end{align}
	\end{theorem}
	
	\begin{remark}
		It is noted that although the disintegrated bounds are in general tighter than their integrated counterparts (e.g., \eqref{eq:gen-dis-ICIMI} is tighter than \eqref{eq:gen-ICIMI}, and \eqref{eq:gen-dis-IPCIMI} is tighter than \eqref{eq:gen-IPCIMI}), they may suffer from the issue of computational intractability \cite{haghifam2020sharpened,zhou2022individually}. For efficient evaluations and fair comparisons, we mainly focus on the integrated CMI-based bounds (e.g., \eqref{eq:gen-ICIMI} and \eqref{eq:gen-IPCIMI}) unless otherwise stated.
	\end{remark}
	
	The intuition conveyed by the CMI quantity $I( W; U^{(i)}_{[\frac{n}{k}]} | \hat{Z}^{(i)}_{[\frac{n+m}{k}]} )$ is whether the output hypothesis $W$ can infer the supersamples, from a given block, that contribute to the training, and a negative answer to this question implies a good generalization behavior. Next, we show in the following lemma that the CGF of $\mathcal{E}^{(i)}(\tilde{W}, \hat{z}^{(i)}_{[\frac{n+m}{k}]}, \tilde{U}^{(i)}_{[\frac{n}{k}]})$ can be upper bounded when the loss function satisfies some specific assumption. 
	\begin{lemma}\label{lem:bld-CGF-bound}
		Consider the partitioned L$m$O supersample setting in Definition~\ref{def:p-lmo}. Suppose $\sup_{w \in \Wcal} |\ell(w, z) - \ell(w, z')| \leq \Delta$ for any $z, z' \in \Zcal$, then for all $\hat{z}^{(i)}_{[\frac{n+m}{k}]} \in \Zcal^{\frac{n+m}{k}}$, $i \in [k]$ and $\lambda \in \RR$, we have
		\begin{align}
			\psi_{\mathcal{E}^{(i)}(\tilde{W}, \hat{z}^{(i)}_{[\frac{n+m}{k}]}, \tilde{U}^{(i)}_{[\frac{n}{k}]})}(\lambda) \leq \frac{\lambda^2 \Delta^2 C^k_{n,m} \cdot k(n+m)}{8nm}, \label{eq:CGF-bound}
		\end{align}
		where $C^k_{n,m} = \begin{cases}
			\frac{n+m}{\max(n, m)} & \min(n, m) = k, \\
			\frac{n+m}{n+m-k} \cdot \frac{nm}{nm - k \min(n, m)} & \text{otherwise.}
		\end{cases}$.
	\end{lemma}
	The proof of the above lemma is provided in Appendix~\ref{appsub:lemma2}. Combining this lemma with Theorem~\ref{thm:IPCIMI} directly yields a specific upper bound for bounded loss functions.
	
	\begin{theorem}[IPCIMI Bounds for Bounded Loss Difference] \label{thm:bld-IPCIMI}
		Consider the partitioned L$m$O supersample setting in Definition~\ref{def:p-lmo}. Suppose $\sup_{w \in \Wcal} |\ell(w, z) - \ell(w, z')| \leq \Delta$ for all $z, z' \in \Zcal$, then
		\begin{align}
			\gen
			%		 &\leq \frac{1}{k} \sum_{i=1}^{k} \EE_{\hat{Z}^{(i)}_{[\frac{n+m}{k}]}} \left[ \sqrt{\frac{\Delta^2 C^k_{n,m} \cdot k(n+m)}{2nm} I^{\hat{Z}^{(i)}_{[\frac{n+m}{k}]}}\left( W; U^{(i)}_{[\frac{n}{k}]} \right)} \right] \label{eq:bld-dis-IPCIMI} \\
			&\leq \frac{1}{k} \sum_{i=1}^{k} \sqrt{\frac{\Delta^2 C^k_{n,m} \cdot k(n+m)}{2nm} I\left( W; U^{(i)}_{[\frac{n}{k}]} | \hat{Z}^{(i)}_{[\frac{n+m}{k}]} \right)},
			\label{eq:bld-IPCIMI}
		\end{align}
	\end{theorem}
	
	%\begin{remark}
	%	Under the partitioned supersample setting, we can obtain another type of CMI-based bound by marginalizing out only the irrelevant $U^{(i)}_{[\frac{n}{k}]}$'s. As such, the new bound, referred to as the individually partitioned conditional individual mutual information (IPCMI) bound, has its CMI quantity conditioning on the whole supersample set, which is given by
	%	\begin{align}
		%		\gen &\leq \frac{1}{k} \sum_{i=1}^{k} \EE_{\hat{Z}_{[n+m]}} \left[ \sqrt{\frac{\Delta^2 C^k_{n,m} \cdot k(n+m)}{2nm} I^{\hat{Z}_{[n+m]}} \left( W; U^{(i)}_{[\frac{n}{k}]} \right)} \right] \label{eq:dis-IPCMI} \\
		%		&\leq \frac{1}{k} \sum_{i=1}^{k} \sqrt{\frac{\Delta^2 C^k_{n,m} \cdot k(n+m)}{2nm} I\left( W; U^{(i)}_{[\frac{n}{k}]} | \hat{Z}_{[n+m]} \right)}. 
		%		\label{eq:delta-IPCMI}
		%	\end{align}
	%	In comparison with the IPCIMI bound in Theorem~\ref{thm:IPCIMI}, \eqref{eq:delta-IPCMI} is looser than \eqref{eq:bld-IPCIMI}, which holds with $I( W; U^{(i)}_{[\frac{n}{k}]} | \hat{Z}^{(i)}_{[\frac{n+m}{k}]} ) \leq I( W; U^{(i)}_{[\frac{n}{k}]} | \hat{Z}_{[n+m]} )$ by applying Lemma~\ref{lem:more-info-by-condi} with $A = W$, $B = U^{(i)}_{[\frac{n}{k}]}$, $C = \hat{Z}^{(i)}_{[\frac{n+m}{k}]}$ and $D = \hat{Z}_{[n+m]} \backslash \hat{Z}^{(i)}_{[\frac{n+m}{k}]}$. However, the comparison between \eqref{eq:dis-IPCMI} and \eqref{eq:bld-IPCIMI} is non-trivial, implying that the disintegrated IPCMI bound in \eqref{eq:dis-IPCMI} can be a potentially tighter result.
	%\end{remark}
	
	We note that the assumption in Theorem~\ref{thm:bld-IPCIMI} requires the loss difference to be bounded, which is a weaker condition than the bounded loss assumption widely adopted in the previous bounds, e.g., \eqref{eq:CMI}, \eqref{eq:ICIMI} and \eqref{eq:LOO-CMI}. In particular, suppose that $\ell(\cdot, \cdot) \in [0,1]$, we have $\Delta \leq 1$ and \eqref{eq:bld-IPCIMI} can be specialized to
	\begin{equation}
		\begin{aligned}
			\gen \leq \frac{1}{k} \sum_{i=1}^{k} \sqrt{\frac{C^k_{n,m} \cdot k(n+m)}{2nm} I\left( W; U^{(i)}_{[\frac{n}{k}]} | \hat{Z}^{(i)}_{[\frac{n+m}{k}]} \right)}.
		\end{aligned}\label{eq:IPCIMI}
	\end{equation}
	
	For clarity, we use the prefix $(m ,k)$ ahead of the IPCIMI bound to indicate the joint effect of both free parameters $m$ and $k$, i.e., \eqref{eq:gen-IPCIMI} is referred to as the $(m ,k)$-IPCIMI bound. Thus, the strongest result can be achieved by optimizing over $(m,k)$ such that
	\begin{align}
		\gen \leq \inf_{m \in \NN, k \in \Kcal} \left\{ \textnormal{RHS of \eqref{eq:gen-IPCIMI}} \right\}. \label{eq:optimal-IPCIMI}
	\end{align}
	We note that finding the optimal $(m, k)$ for the above optimizer is intractable. The challenge arises when the bound is calculated for each possible $(m, k)$ through a brute-force search process. However, we will show in the next subsection that some specific $(m, k)$ pair values allow to recover most of the classical MI- and CMI-based bounds, implying that \eqref{eq:optimal-IPCIMI} is at least no looser than the existing results. Moreover, it will be further explored that some other specifications of $(m, k)$ are sufficient to derive new bounds even sharper than the existing bounds.
	
	\subsection{Special Cases}\label{subsec:IPCIMI-case}
	\begin{table*}[t]
		\centering
		\caption{Various Bounds Specialized from the Partitioned Leave-$m$-Out Supersample Setting}
		\begin{threeparttable}
			\begin{tabular}{|c|c|c|c|} \hline
				Setting & Bound & General form & Bounded loss\tnote{$\dag$} \\ \hline\hline
				\multirow{2}{*}{$(m, k)$} & IPCIMI & \eqref{eq:gen-IPCIMI} & \eqref{eq:IPCIMI} \\ \cline{2-4}
				& SIPCIMI & -- & \eqref{eq:SIPCIMI} \\ \hline
				\multirow{2}{*}{$(m, 1)$} & L$m$O-CMI & -- & \eqref{eq:lmo-CMI} \\ \cline{2-4}
				& L$m$O-SCMI & -- & \eqref{eq:lmo-SCMI} \\ \hline
				\multirow{2}{*}{$(1, 1)$} & LOO-CMI \cite{rammal2022leave} & -- & \eqref{eq:LOO-CMI} \\ \cline{2-4}
				& LOO-SCMI & -- & \eqref{eq:loo-SCMI} \\ \hline
				\multirow{2}{*}{$(n, n)$} & ICIMI \cite{zhou2022individually} & \eqref{eq:gen-ICIMI} & \eqref{eq:ICIMI} \\ \cline{2-4}
				& SICIMI \cite{dongexactly} & -- & \eqref{eq:SICIMI} \\ \hline
				$(\infty, k)$ & IPMI & -- & \eqref{eq:IPMI} \\ \hline
				$(\infty, 1)$ & MI \cite{xu2017information} & -- & \eqref{eq:MI} \\ \hline
				$(\infty, n)$ & IMI \cite{bu2020tightening} & \eqref{eq:gen-IMI} & \eqref{eq:IMI} \\ \hline
				\multirow{2}{*}{$(m, m)$} & LOFO-CMI & \eqref{eq:gen-lofo} & \eqref{eq:lofo-CMI} \\ \cline{2-4}
				& LOFO-SCMI & -- & \eqref{eq:lofo-SCMI} \\ \hline
			\end{tabular}
			\begin{tablenotes} \footnotesize
				\item[$\dag$] Bounds under the assumption of $\ell(\cdot, \cdot) \in [0,1]$.
			\end{tablenotes} 
		\end{threeparttable} 
		\label{tab:special_cases}
	\end{table*}
	
	Now, we specify some values of $(m, k)$ and study the relationships between the special cases and the previously proposed bounds. Table~\ref{tab:special_cases} summarizes various special cases of the proposed $(m, k)$-IPCIMI bound \eqref{eq:gen-IPCIMI} when taking different pair values of $(m, k)$.
	
	First, thanks to the construction of the partitioned supersample set, it can be easily verified that the ICIMI bound \eqref{eq:gen-ICIMI} can be immediately recovered by directly plugging $(m, k) = (n, n)$ into \eqref{eq:gen-IPCIMI}, and meanwhile \eqref{eq:ICIMI} can also be recovered from \eqref{eq:IPCIMI} when the loss function is bounded within $[0,1]$. Moreover, still consider bounded loss functions, then the LOO-CMI bound \eqref{eq:LOO-CMI} can be recovered from \eqref{eq:bld-IPCIMI} by setting $(m, k) = (1, 1)$. We also note that, by applying the same parameter setting to Theorem~\ref{thm:IPCIMI}, one can directly produce a general LOO bound beyond the bounded loss assumption. Although such a bound was not reported in \cite{haghifam2022understanding,rammal2022leave}, it can be readily established from Theorem~\ref{thm:IPCIMI} and is therefore omitted. 
	
	We further argue that the proposed IPCIMI bound is also related to MI-based bounds when $m \to \infty$, which bridges the gap between two seemingly distant settings. To see this, the following result first demonstrates that the gap between the CMI in \eqref{eq:IPCIMI} and its counterpart MI diminishes as $m$ goes to infinity.
	\begin{lemma} \label{lem:lim-MI}
		Suppose $k, n$ are constant and $|\Wcal| < \infty$. When $m \to \infty$, for each $i \in [k]$ we have $\lim\limits_{m \to \infty} I\left( W; U^{(i)}_{[\frac{n}{k}]} | \hat{Z}^{(i)}_{[\frac{n+m}{k}]} \right) = I\left( W; Z^{(i)}_{[\frac{n}{k}]} \right)$.
	\end{lemma}
	The proof of this result is given in Appendix~\ref{appsub:lemma3}. This lemma shows that the CMI quantity will approach the MI quantity in the limit as $m \to \infty$, in which case the leading coefficient under the square-root of \eqref{eq:IPCIMI} is of order $\frac{k}{2n} + O(\frac{1}{m})$ (See Appendix~\ref{subsec:proof-IPMI} for complete calculations). Collecting all the facts together, we observe that the IPCIMI bound for bounded loss is asymptotically equivalent to a new MI-based bound, as shown in the corollary below.
	\begin{corollary}[($(\infty, k)$-IPCIMI) IPMI Bound] \label{cor:IPMI}
		Consider the partitioned L$m$O supersample setting in Definition~\ref{def:p-lmo}. Assume $\ell(\cdot, \cdot) \in [0,1]$ and $|\Wcal| < \infty$. When $m \to \infty$, the IPCIMI bound \eqref{eq:IPCIMI} will reduce to an individually partitioned mutual information (IPMI) bound:
		\begin{align}
			\gen \leq \frac{1}{k} \sum_{i=1}^{k} \sqrt{\frac{k}{2n} I\left( W; Z^{(i)}_{[\frac{n}{k}]} \right)}.
			\label{eq:IPMI}
		\end{align}
	\end{corollary}
	As straightforward results of Corollary~\ref{cor:IPMI}, the MI bound \eqref{eq:MI} and the IMI bound \eqref{eq:IMI} can be immediately obtained by setting $k=1$ and $k = n$, respectively. It is verified that the IPMI bound can achieve the same decay order $O(\frac{1}{\sqrt{n}})$ with the ICIMI bound under Example~\ref{eg:Bernoulli} (see Appendix~\ref{app:calculations}), implying that partitioning the supersample set may sharpen the resulting bound.
	
	Apart from providing connections with the existing bounds, the proposed unified CMI framework also produce new and tighter bounds. For instance, simply setting $k=1$ leads to the following L$m$O-CMI bound.
	\begin{corollary}[($(m, 1)$-IPCIMI) L$m$O-CMI Bound] \label{cor:lmo-CMI}
		Consider the partitioned L$m$O supersample setting in Definition~\ref{def:p-lmo}. Suppose that $\sup_{w \in \Wcal} |\ell(w, z) - \ell(w, z')| \leq \Delta$ for all $z, z' \in \Zcal$. When $k=1$, the IPCIMI bound \eqref{eq:bld-IPCIMI} is specialized as
		\begin{align}
			\gen \leq \sqrt{\frac{\Delta^2 C_{n,m} \cdot (n+m)}{2nm} I\left( W; U_{[n]} | \hat{Z}_{[n+m]} \right)},
			\label{eq:lmo-CMI}
		\end{align}
		where $C_{n,m} = \begin{cases}
			\frac{n+m}{\max(n, m)} &  \min(n, m) = 1, \\
			\frac{n+m}{n+m-1} \cdot \frac{nm}{nm - \min(n, m)} & \text{otherwise.}
		\end{cases}$.
	\end{corollary}
	
	Moreover, by simply setting $k = m$, where $m \in [n]$, the supersample set is divided into $m$ blocks and only one sample is left out of each block. This setting resembles the standard leave-one-fold-out (LOFO) cross-validation (corresponding to the $l$-fold cross-validation introduced in \cite{mohri2018foundations}, where $l = \frac{n}{m} + 1$). Therefore, the new bound obtained by applying Theorem~\ref{thm:IPCIMI} with $(m, k) = (m, m)$ is referred to as the LOFO-CMI bound.
	
	\begin{corollary}[($(m, m)$-IPCIMI) LOFO-CMI Bound] \label{cor:lofo-CMI}
		Consider the partitioned L$m$O supersample setting in Definition~\ref{def:p-lmo}. Specially, consider $k = m$ (where $m$ divides $n$) and let $\bar{U}^{(i)}$ denote $\bar{U}^{(i)}_{[\frac{m}{k}]}$ for simplicity. Let $(\tilde{W}, \tilde{\bar{U}}^{(i)})$ be a decoupled pair of $(W, \bar{U}^{(i)})$ conditioned on $\hat{Z}^{(i)}_{[\frac{n}{m}+1]}$, i.e., $P_{\tilde{W}, \tilde{\bar{U}}^{(i)} | \hat{Z}^{(i)}_{[\frac{n}{m}+1]}} = P_{W | \hat{Z}^{(i)}_{[\frac{n}{m}+1]}} \otimes P_{\bar{U}^{(i)} | \hat{Z}^{(i)}_{[\frac{n}{m}+1]}}$, where $P_{\bar{U}^{(i)} | \hat{Z}^{(i)}_{[\frac{n}{m}+1]}} = P_{\bar{U}^{(i)}}$ as $\bar{U}^{(i)} \indep \hat{Z}^{(i)}_{[\frac{n}{m}+1]}$. Then, 
		\begin{align}
			\gen &\leq \frac{1}{m} \sum_{i=1}^{m} \inf_{\lambda > 0} \frac{1}{\lambda} \left(I\left( W; \bar{U}^{(i)} | \hat{Z}^{(i)}_{[\frac{n}{m}+1]} \right) + \EE \left[ \psi_{\mathcal{E}^{(i)}(\tilde{W}, \hat{Z}^{(i)}_{[\frac{n}{m}+1]}, \tilde{\bar{U}}^{(i)}) | \hat{Z}^{(i)}_{[\frac{n}{m}+1]}}(\lambda) \right] \right). \label{eq:gen-lofo}
		\end{align}
		In particular, if $\ell(\cdot, \cdot) \in [0,1]$, we have
		\begin{equation}
			\begin{aligned}
				\gen \leq \frac{n+m}{nm} \sum_{i=1}^{m} \sqrt{ \frac{1}{2} I\left( W; \bar{U}^{(i)} | \hat{Z}^{(i)}_{[\frac{n}{m} + 1]} \right)}.
				\label{eq:lofo-CMI}
			\end{aligned}
		\end{equation}
	\end{corollary}

	\subsection{Numerical Results}
	In this subsection, we evaluate and compare the considered bounds under several learning scenarios. Supportive numerical results illustrate the benefits of the proposed bounds.
	\subsubsection{Bernoulli mean estimation}
	\begin{figure}[t]
		\centering
		%	\subfloat[]{\includegraphics[width=0.33\textwidth]{non_asym_sim_m.eps} \label{fig:non_asym_sim_m}}
		%	\subfloat[]{\includegraphics[width=0.66\textwidth]{non_asym_sim_n.eps} \label{fig:non_asym_sim_n}}
		\includegraphics[width=0.8\textwidth]{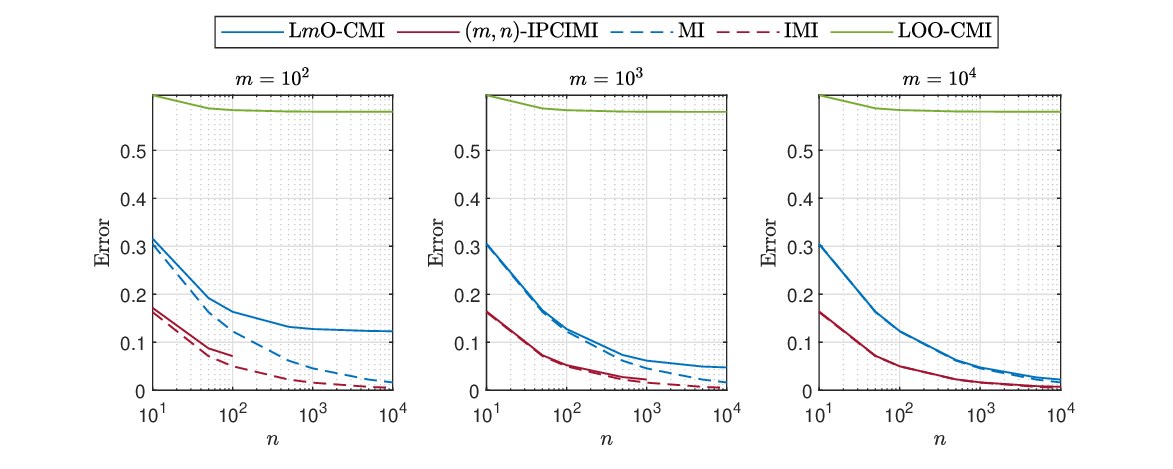}
		\caption{Comparison of various bounds for Example~\ref{eg:Bernoulli}, where we set $p=0.4$ and $m \in \{10^2,10^3,10^4\}$. As shown in the figure, our proposed L$m$O-CMI bound \eqref{eq:lmo-CMI} and $(m, n)$-IPCIMI bound \eqref{eq:m_n_IPCIMI} are significantly tighter than the existing LOO-CMI bound \eqref{eq:LOO-CMI}. When $m$ tends to infinity, the two proposed CMI-based bounds converge to the existing MI bound \eqref{eq:MI} and IMI bound \eqref{eq:IMI}, respectively. }\label{fig:non_asym_sim}
	\end{figure}
	
	We first leverage Example~\ref{eg:Bernoulli} to validate the relationship between the MI- and CMI-based bounds. Note that since the bounded-loss assumption is satisfied for this example, all the considered bounds admit their special forms. To validate Corollary~\ref{cor:IPMI}, we conduct numerical comparisons between the IPCIMI bound and the IPMI bound when $m$ increases. In particular, we choose $k \in \{1, n\}$. For $k=1$, the comparison is between L$m$O-CMI and MI bound. While for $k=n$, the comparison is with respect to a new $(m, n)$-IPCIMI bound and IMI bound, where the $(m, n)$-IPCIMI bound is given by\footnote{This bound can directly be obtained by setting $k = n$ in \eqref{eq:IPCIMI}.}
	\begin{equation}\label{eq:m_n_IPCIMI}
		\begin{aligned}
			\gen \leq \frac{n+m}{nm} \sum_{i=1}^{n} \sqrt{ \frac{1}{2} I\left( W; U^{(i)} | \hat{Z}^{(i)}_{[\frac{m}{n} + 1]} \right)},
		\end{aligned}
	\end{equation}
	where $U^{(i)}$ denotes $U^{(i)}_{[\frac{n}{k}]}$ and $m \geq n$. Therefore, it suffices to show that the L$m$O-CMI bound converges to the MI bound and the $(m, n)$-IPCIMI bound converges to the IMI bound when $m$ goes to infinity, as stated by Corollary~\ref{cor:IPMI}. We calculate all the considered bounds and obtain the numerical results via Monte-Carlo simulations. Note that the calculations are detailed in Appendix~\ref{subapp:cal_MI} to \ref{subapp:cal_mn_IPCIMI}. For the particular Bernoulli example, MI, IMI, LOO-CMI, ICIMI, L$m$O-CMI and LOFO-CMI bounds are specialized as \eqref{eq:Ber_MI}, \eqref{eq:Ber_IMI}, \eqref{eq:Ber_LOO}, \eqref{eq:Ber_ICIMI}, \eqref{eq:Ber_lmo} and \eqref{eq:Ber_LOFO}, respectively. Fig.~\ref{fig:non_asym_sim} plots the performance of the considered bounds as $m$ increases. It is observed that the CMI-based bounds are decreasing with $m$ and converge to their corresponding MI-based bounds when $m$ goes to infinity, supportively validating Corollary~\ref{cor:IPMI}. Besides, we show that the LOO-CMI bound does not provide a vanishing bound for the Bernoulli example (as stated by Proposition~\ref{pro:LOO-CMI}), while the L$m$O-CMI bound is much tighter than the LOO-CMI bound, especially when choosing $m \gg n$. This implies that our L$m$O supersample framework can tighten and improve the existing bounds.
	
	%\begin{figure}[t]
	%	\centering
	%	\includegraphics[width=0.8\textwidth]{lofo_sim.eps}
	%	\caption{Comparison of various bounds, where we set $p = 0.1$ for the Bernoulli example.}
	%	\label{fig:lofo_sim}
	%\end{figure}
	
	\subsubsection{Gaussian mean estimation}
	We now consider the learning problem of estimating the mean of Gaussian random variables.
	\begin{example}[A Gaussian Example] \label{eg:Gaussian}
		Suppose that all data are sampled from a Gaussian distribution $\Ncal(\mu, \sigma^2)$. Let $W$ be the output of the ERM algorithm that gives the average of the training samples. Consider a quadratic loss function $\ell(w, z) = (w - z)^2$, then the true generalization error can be exactly calculated as $\gen = \frac{2\sigma^2}{n}$ \cite{bu2020tightening}, which is of order $O(\frac{1}{n})$.
	\end{example}
	
	\begin{figure}[t]
		\centering
		\includegraphics[width=0.45\textwidth]{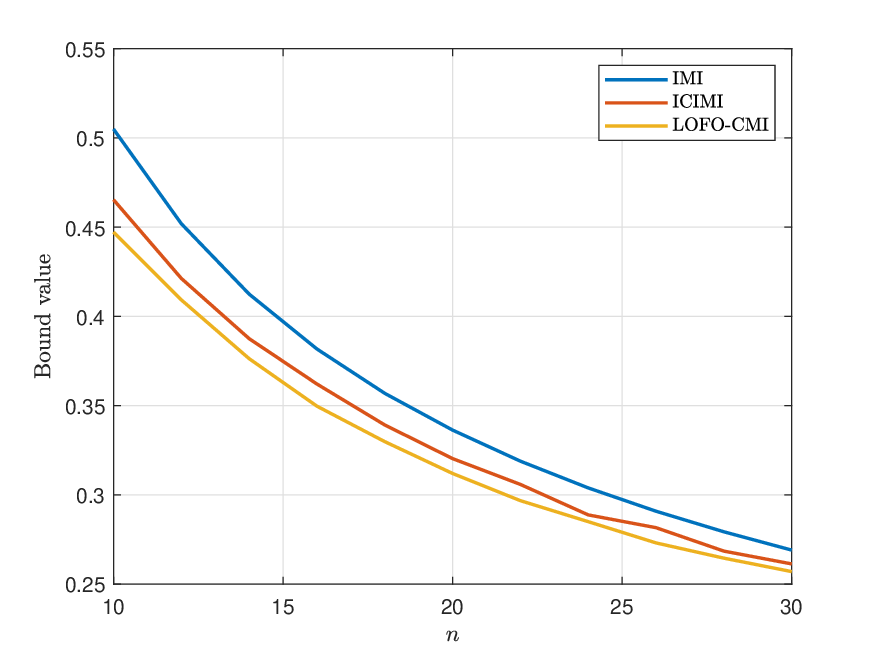}
		\caption{Comparison of various bounds for Example~\ref{eg:Gaussian}, where we set $m = \frac{n}{2}$, $\mu=0$ and $\sigma=1$. It is observed that the proposed LOFO-CMI bound \eqref{eq:gen-lofo} is tighter than the existing IMI bound \eqref{eq:gen-IMI} and ICIMI bound \eqref{eq:gen-ICIMI}, while the LOO-CMI bound \eqref{eq:LOO-CMI} is not applicable in this example.}
		\label{fig:Gaussian}
	\end{figure}
	
	Distinguished from the previous Bernoulli mean estimation in Example~\ref{eg:Bernoulli}, the training samples in Example~\ref{eg:Gaussian} are Gaussian random variables with an infinite alphabet $\Zcal = \RR$. Thus, the assumption of bounded loss difference is no longer satisfied, making bounds in Theorem~\ref{thm:bld-IPCIMI}, LOO-CMI \eqref{eq:LOO-CMI} and LOFO-CMI \eqref{eq:lofo-CMI} bounds inapplicable. Therefore, we turn to the general form of the proposed bounds as they do not require any special assumptions on the loss function. In particular, we compare the proposed LOFO-CMI bound \eqref{eq:gen-lofo} with the known ICIMI bound \eqref{eq:gen-ICIMI} and IMI bound \eqref{eq:gen-IMI}, where the expression for the IMI bound under the Gaussian example is \cite{bu2020tightening}
	\begin{align}
		\gen \leq \sigma^2 \sqrt{\frac{2(n+1)^2}{n^2} \log \frac{n}{n-1}}.
	\end{align}
	Since the other two CMI-based bounds do not have closed-formed expressions, we use Monte-Carlo simulations to approximate the involved CGF terms and the numerical results are reported in Fig.~\ref{fig:Gaussian}. With a carefully selected $m=\frac{n}{2}$, we observe that the proposed LOFO-CMI bound is tighter than the other considered bounds.
	
	\subsubsection{Gaussian mean estimation with finite $\Wcal$}
	The next example, which was explored in \cite{wu2025fast}, is a variant problem of estimating the Gaussian mean.
	\begin{example}[A Gaussian Example with Finite $\Wcal$] \label{eg:finite_Gaussian}
		Consider the Gaussian mean estimation with a finite hypothesis space $\Wcal \in \{-\mu, \mu\}$. Let $w^* \in \Wcal$ be the true hypothesis and all data are sampled from the Gaussian distribution $\Ncal(w^*, \sigma^2)$, where we assume $w^* = \mu$. The ERM algorithm trained with the quadratic loss function $\ell(w, z) = (w - z)^2$ produces the hypothesis $W$ which follows the maximum likelihood decision rule:
		\begin{align}
			W = \begin{cases}
				\mu, & \textnormal{if $\frac{1}{n} \sum_{i=1}^n Z_i \geq 0$,} \\
				-\mu, & \textnormal{otherwise.}
			\end{cases}\label{eq:W_MLD}
		\end{align}
		We alternatively consider a truncated loss function $\ell(w, z) = \min((w-z)^2, 1)$ and let the learning hypothesis $W$ still be as given in \eqref{eq:W_MLD}.
	\end{example}
	\begin{figure}[t]
		\centering
		\includegraphics[width=0.45\textwidth]{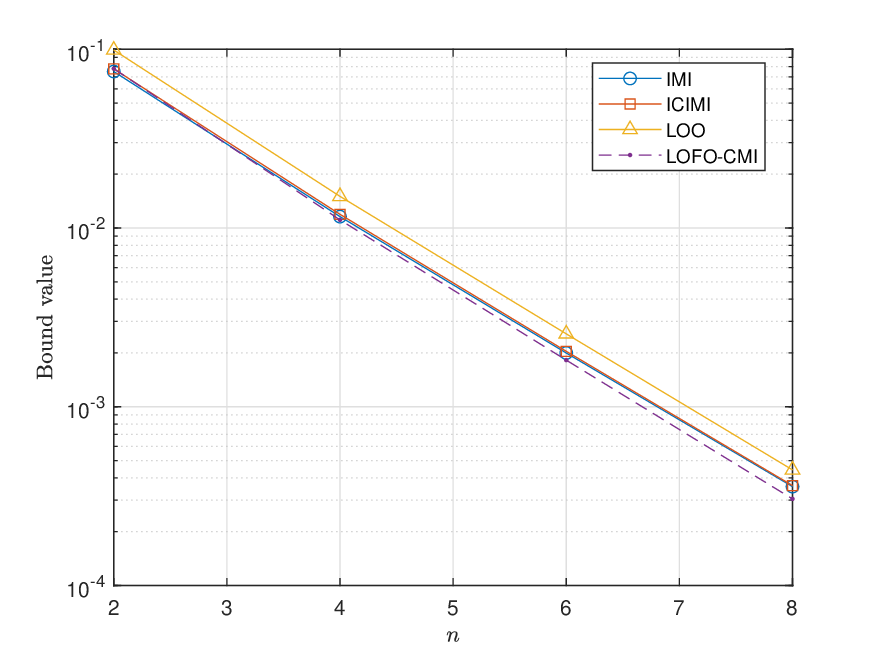}
		\caption{Comparison of various bounds for Example~\ref{eg:finite_Gaussian}, where we set $m = 2$, $\mu=1$ and $\sigma=0.5$. As shown in the figure, the proposed LOFO-CMI bound \eqref{eq:lofo-CMI} is tighter than the existing IMI bound \eqref{eq:IMI}, ICIMI bound \eqref{eq:ICIMI} and LOO-CMI bound \eqref{eq:LOO-CMI}.}
		\label{fig:Gaussian_finite_W}
	\end{figure}
	
	The adopted truncated loss function satisfies the bounded-loss assumption, enabling us to directly apply the simpler form for each considered bound. In Fig.~\ref{fig:Gaussian_finite_W}, we compare the proposed LOFO-CMI bound \eqref{eq:lofo-CMI} with the known IMI bound \eqref{eq:IMI}, ICIMI bound \eqref{eq:ICIMI} and LOO-CMI bound \eqref{eq:LOO-CMI}, under the setting of $m = 2$, $\mu=1$ and $\sigma=0.5$. It is observed that none of the CMI-based bound, expect the proposed LOFO-CMI bound, can outperform the IMI bound. Besides, as argued before, the lower-dimensional $U$ involved in LOFO-CMI bound brings an extra benefit of lower computational complexity, as compared to the other CMI-based bounds.
	
	All the examples explored in this subsection show that the proposed bounds are promising in various scenarios, and the free parameters incorporated in bounds offer chances of flexible trade-off between the tightness and effectiveness.
	
	\section{Bounds Conditioning on Single Supersample}\label{sec:SCMI-bounds}
	In this section, under the partitioned L$m$O supersample setting we establish a series of new CMI-based bounds for bounded loss functions, in parallel with those presented in Section~\ref{sec:unified-CMI}. The new bounds admit CMI quantities in which the mutual information is measured by conditioning on a single supersample instead of a block or the whole supersample set. We show through examples that these single CMI (SCMI) based bounds can be tighter than their counterpart IPCIMI bounds. Additionally, when equipped with an assumption that is widely satisfied by a wide range of learning algorithms, it is safe to eliminate the partition number $k$ in SCMI-based bounds, making the optimization of the remaining free parameter $m$ an easier searching process. Simulation results show that under Example~\ref{eg:Bernoulli} the proposed SCMI-based bound tends to gain better performance with increasing $m$.
	\subsection{Existing Single CMI Bound}
	It was explored in \cite{zhou2022individually, dongexactly} that the CMI-based bounds can be improved by eliminating redundant random variables from the conditional term. In particular, the ICIMI bound in \eqref{eq:ICIMI} was strengthened to its single version, namely the SICIMI bound, where the dependence on the conditional sample pair $(\hat{Z}_i, \hat{Z}_{i+n})$ is relaxed on single sample $\hat{Z}_i$ (or $\hat{Z}_{i+n}$), as shown in the theorem below. 
	
	\begin{theorem}[SICIMI Bound \cite{dongexactly}] \label{thm:SICIMI}
		Consider the standard supersample setting in Definition~\ref{def:standard-supersample}. If $\ell(w, z) \in [0, 1]$ for all $w \in \Wcal$ and $z \in \Zcal$, then
		\begin{align}
			\gen &\leq \frac{1}{n} \sum_{i = 1}^n \EE \left[ \sqrt{2 I^{\hat{Z}_i}\left( W; R_i \right)} \right] \label{eq:dis-SICIMI} \\
			&\leq \frac{1}{n} \sum_{i = 1}^n \sqrt{2 I\left( W; R_i | \hat{Z}_i \right)}.
			\label{eq:SICIMI}
		\end{align}
	\end{theorem}
	As compared to \eqref{eq:ICIMI}, it was proved that $I( W; R_i | \hat{Z}_i ) \leq I( W; R_i | \hat{Z}_i, \hat{Z}_{i+n} )$ due to the fact that conditioning on more indepedent random variables will not decrease the mutual information \cite{dongexactly}. Thus, it follows that the SICIMI bound is tighter than the ICIMI bound, and thus achieves the best performance amongst the CMI-based bounds built on the standard supersample setting (Definition~\ref{def:standard-supersample}). Following this progression, we next extend the SICIMI bound to the proposed general CMI framework, facilitating the single version of the proposed IPCIMI bounds.

	\subsection{A Unified SCMI Bound} \label{subsec:SCMI}
	Prior to our main results, a few definitions are presented as follows. For any $p,q \in [0,1]$, the KL divergence between two Bernoulli random variables with parameters $p$ and $q$ is defined as the binary KL divergence \cite{hellstrom2022new}, i.e.,
	\begin{equation}
		\begin{aligned}
			d_{\textnormal{KL}} (p \| q) := D_{\textnormal{KL}} (\Bernoulli(p) \| \Bernoulli(q)) = p \log \frac{p}{q} + (1 - p) \log \frac{1 - p}{1 - q}.
		\end{aligned}
	\end{equation}
	A relaxed version of the binary KL divergence is further defined as \cite{mcallester2013pac}
	\begin{equation}
		\begin{aligned}
			d_\gamma (p \| q) := \gamma q - \log ( 1 - p + p e^\gamma),
		\end{aligned}\label{eq:def-d_gamma}
	\end{equation}
	which satisfies
	\begin{equation} \label{eq:property_d_gamma}
		\begin{aligned}
			d_{\textnormal{KL}} (p \| q) = \sup_{\gamma \in \RR} \{d_\gamma (p \| q)\}.
		\end{aligned}
	\end{equation}
	Note that both $d_{\textnormal{KL}}$ and $d_{\gamma}$ are jointly convex in their arguments \cite{hellstrom2022new}. Next define a symmetrized and smoothed version of the binary KL divergence, namely the weighted binary Jensen-Shannon (JS) divergence, which is given by
	\begin{align}
		d^\theta_{\textnormal{JS}}(p \| q) := \theta d_{\textnormal{KL}}(p \| \theta p + (1 - \theta) q) + (1 - \theta) d_{\textnormal{KL}}(q \| \theta p + (1 - \theta) q), \label{eq:d_js}
	\end{align}
	where $\theta$ is the weight for $p$. Note that a special case of $\theta = \frac{1}{2}$ can be found in \cite{dongexactly}. Finally, the inverse of the binary JS divergence is defined as
	\begin{equation}
		\begin{aligned}
			d^{-1}_{\textnormal{JS}} (\theta, p, c) := \sup_{q \in [0,1]} \left\{ d^\theta_{\textnormal{JS}}(p \| q) \leq c \right\}.
		\end{aligned}\label{eq:inv-d_js}
	\end{equation}
	
	Built upon the proposed binary JS divergence, an instrumental lemma is introduced to derive a series of information-theoretic generalization bounds later in this paper. 
	
	\begin{lemma}\label{lem:js-bound}
		For a probability space $(\Omega, \Fcal, P)$, let $X : \Omega \to \Xcal$ be a random variable, $Y: \Omega \to \Ycal$ a discrete random variable. Given a measurable set $\Scal \subseteq \Ycal$, let $T := \mathds{1}_{Y \in \Scal}$ be a Bernoulli random variable with $P(T = 1)  = \PP[Y \in \Scal] \in (0, 1)$. Assume that $f(X) \in [0, 1]$ and we let $\theta = P(T = 1)$, then
		\begin{align}
			d^\theta_{\textnormal{JS}} \left( \EE_{X|T=1} \left[ f(X) \right] \left\| \EE_{X|T=0} \left[ f(X) \right] \right. \right) \leq I(X; T) \leq I(X; Y),
		\end{align}
		where the first inequality holds with equality when $f(X) = X$ and $X \in \{0, 1\}$.
	\end{lemma}
	
	The lemma is proved through the convexity of the KL divergence and the Donsker–Varadhan inequality, with details relegated to Appendix~\ref{subapp:js-bound}. A special case of the lemma, where $Y$ follows $\Bernoulli(\frac{1}{2})$, was previously presented and explored in \cite{dongexactly}. Using Lemma~\ref{lem:js-bound}, we present below a family of SCMI-based bounds under the proposed partitioned L$m$O supersample setting (Definition~\ref{def:p-lmo}).

	\begin{theorem}[SIPCIMI Bounds]\label{thm:SIPCIMI}
		Consider the partitioned L$m$O setting in Definition~\ref{def:p-lmo}. Let $\Ucal_j$ be the collection of all the realizations of $U^{(i)}_{[\frac{n}{k}]}$ that includes any given index $j \in [\frac{n+m}{k}]$, i.e., $\Ucal_j = \{u | j \in u, u \subset [\frac{n+m}{k}], |u| = \frac{n}{k} \}$. Let $T^{(i)}_j := \mathds{1}_{U^{(i)}_{[\frac{n}{k}]} \in \Ucal_j}$ be a Bernoulli random variable. If $\ell(w, z) \in [0, 1]$ for all $w \in \Wcal$ and $z \in \Zcal$, then
		\begin{subequations}
			\begin{align}
				\gen \leq& \sum_{i=1}^{k} \sum_{j=1}^{\frac{n+m}{k}} \EE \left[ \sqrt{\frac{1}{2nm} I^{\hat{Z}^{(i)}_j}\left( W; U^{(i)}_{[\frac{n}{k}]} \right)} \right] \qquad (\textnormal{unprocessed disintegrated bound}) \label{eq:dis-SIPCIMI} \\
				\leq& \sum_{i=1}^{k} \sum_{j=1}^{\frac{n+m}{k}} \sqrt{\frac{1}{2nm} I\left( W; U^{(i)}_{[\frac{n}{k}]} | \hat{Z}^{(i)}_j \right)} \qquad (\textnormal{unprocessed integrated bound}) \label{eq:SIPCIMI}
			\end{align}
		\end{subequations}
		\begin{subequations}
			\begin{align}		
				\gen \leq& \sum_{i=1}^{k} \sum_{j=1}^{\frac{n+m}{k}} \EE \left[ \sqrt{\frac{1}{2nm} I^{\hat{Z}^{(i)}_j}\left( W; T^{(i)}_j \right)} \right] \qquad (\textnormal{processed disintegrated bound}) \qquad \label{eq:p-dis-SIPCIMI} \\
				\leq& \sum_{i=1}^{k} \sum_{j=1}^{\frac{n+m}{k}} \sqrt{\frac{1}{2nm} I\left( W; T^{(i)}_j | \hat{Z}^{(i)}_j \right)} \qquad (\textnormal{processed integrated bound}) \label{eq:p-SIPCIMI} 
			\end{align}
		\end{subequations}
	\end{theorem}
	
	The proof of \eqref{eq:p-dis-SIPCIMI} can be found in Appendix~\ref{subapp:SIPCIMI}. We note that bounds \eqref{eq:p-SIPCIMI} and \eqref{eq:SIPCIMI} follow from \eqref{eq:p-dis-SIPCIMI} and \eqref{eq:dis-SIPCIMI}, respectively, by using Jensen's inequality with concavity of the square-root function, i.e., $\eqref{eq:p-dis-SIPCIMI} \leq \eqref{eq:p-SIPCIMI}$ and $\eqref{eq:dis-SIPCIMI} \leq \eqref{eq:SIPCIMI}$. Hence, the disintegrated bounds are tighter than their integrated counterparts. Besides, bounds \eqref{eq:dis-SIPCIMI} and \eqref{eq:SIPCIMI} follow from \eqref{eq:p-dis-SIPCIMI} and \eqref{eq:p-SIPCIMI}, respectively, by data-processing inequality. More specifically, since $T^{(i)}_j$ is a deterministic function of $U^{(i)}_{[\frac{n}{k}]}$, then the data processing inequality indicates that $I^{\hat{z}^{(i)}_j}( W; T^{(i)}_j ) \leq I^{\hat{z}^{(i)}_j}( W; U^{(i)}_{[\frac{n}{k}]} )$ for any given realization $\hat{z}^{(i)}_j$ and $I( W; T^{(i)}_j | \hat{Z}^{(i)}_j ) \leq I( W; U^{(i)}_{[\frac{n}{k}]} | \hat{Z}^{(i)}_j )$. Hence, the processed bounds are always no looser than their unprocessed counterparts, i.e., $\eqref{eq:p-dis-SIPCIMI} \leq \eqref{eq:dis-SIPCIMI}$ and $\eqref{eq:p-SIPCIMI} \leq \eqref{eq:SIPCIMI}$, where the equations hold when the learning algorithm satisfies a simple but universal assumption, as shown below.
	
	\begin{assumption} \label{ass:invariant-alg}
		Let $\tau = \{\tau_i\}_{i=1}^n$ denote a permutation of $[n]$, such that for a given set $S$ of size $n$, $S_{\tau} = \{S_{\tau_i}\}_{i=1}^n$ represent a permutation of $S$. We say that a learning algorithm $\Acal$ is invariant to permutations of the input dataset $Z_{[n]}$, if for every possible $\tau$ and all $z_{[n]} \in \Zcal^n$, it holds that $\Acal(z_{[n]})$ and $\Acal(z_{\tau})$ are identically distributed (or equivalently, $P_{W| Z_{[n]}} (w | z_{[n]}) = P_{W| Z_{[n]}} (w | z_{\tau})$ for all $w \in \Wcal$).
	\end{assumption}
	
	This assumption considers a learning algorithm $\Acal$ which is insensitive to any permutations of the input training samples. We note that this assumption is not overly limiting as it is typically satisfied by a variety of learning algorithms, to name a few, linear regressions, support vector machines (SVM) and formal stochastic gradient descent (SGD) which are commonly used in modern machine learning scenarios. We next claim the following.
	
	\begin{proposition}\label{pro:processed-SIPCIMI}
		Let $T^{(i)}_j$ be defined in Theorem~\ref{thm:SIPCIMI}. For all $i \in [k]$ and $j \in [\frac{n+m}{k}]$, suppose that Assumption~\ref{ass:invariant-alg} is satisfied, then
		\begin{align} \label{eq:processed-SIPCIMI} 
			\forall z \in \Zcal, \quad I^{\hat{Z}^{(i)}_j = z}\left( W; T^{(i)}_j \right) &= I^{\hat{Z}^{(i)}_j = z}\left( W; U^{(i)}_{[\frac{n}{k}]} \right), \\
			I\left( W; T^{(i)}_j | \hat{Z}^{(i)}_j \right) &= I\left( W; U^{(i)}_{[\frac{n}{k}]} | \hat{Z}^{(i)}_j \right).
		\end{align}
	\end{proposition}
	
	The proof of this result is given in Appendix~\ref{subapp:processed-SIPCIMI}. This proposition implies that the processed and unprocessed bounds are identical, i.e., $\eqref{eq:dis-SIPCIMI} = \eqref{eq:p-dis-SIPCIMI}$ and $\eqref{eq:SIPCIMI} = \eqref{eq:p-SIPCIMI}$, when Assumption~\ref{ass:invariant-alg} is satisfied. Unless otherwise stated, we consider the unprocessed, disintegrated SIPCIMI bound in \eqref{eq:dis-SIPCIMI}. Similar to the case of the IPCIMI bound, the strongest result facilitated from the SIPCIMI bound can be obtained by optimizing over $(m, k)$ such that
	\begin{align}
		\gen \leq \inf_{m \in \NN, k \in \Kcal} \left\{ \textnormal{RHS of \eqref{eq:dis-SIPCIMI}} \right\}. \label{eq:optimal-SIPCIMI}
	\end{align}
	Different from the case in \eqref{eq:optimal-IPCIMI}, where the determination of the optimal $(m, k)$ is non-trivial, we will show that \eqref{eq:optimal-SIPCIMI} can be more tractable in certain scenarios. Specifically, the bound \eqref{eq:dis-SIPCIMI} will not depend on $k$ if the learning algorithm satisfies Assumption~\ref{ass:invariant-alg}, thus enabling an efficient search of the optimal $m$. Detailed discussions and numerical results are relegated to the next subsection.
	
	\subsection{Special Cases}\label{subsec:SIPCIMI-cases}
	We start by stating some new bounds specialized from the proposed SIPCIMI bound. By straightforwardly applying specific pair values of $(m,k)$ to Theorem~\ref{thm:SIPCIMI}, we can recover the SICIMI bound in Theorem~\ref{thm:SICIMI} (with $m=n$ and $k=n$) and further obtain the following corollaries.
	
	\begin{corollary}[($(m, 1)$-SIPCIMI) L$m$O-SCMI Bound]\label{cor:lmo-SCMI}
		Consider the partitioned L$m$O setting in Definition~\ref{def:p-lmo} and assume $\ell(w, z) \in [0, 1]$ for all $w \in \Wcal$ and $z \in \Zcal$. When $k = 1$, the $(m, k)$-SIPCIMI bound gives a new L$m$O-SCMI bound as
		\begin{align}
			\gen &\leq \sum_{i=1}^{n+m} \EE \left[ \sqrt{ \frac{1}{2nm} I^{\hat{Z}_i}\left( W; U_{[n]} \right)} \right] \label{eq:dis-lmo-SCMI}\\
			&\leq \sum_{i=1}^{n+m} \sqrt{ \frac{1}{2nm} I\left( W; U_{[n]} | \hat{Z}_i \right)}.
			\label{eq:lmo-SCMI}
		\end{align}
	\end{corollary}
	
	\begin{corollary}[($(m, m)$-SIPCIMI) LOFO-SCMI Bound]
		Consider the partitioned L$m$O setting in Definition~\ref{def:p-lmo} and assume $\ell(w, z) \in [0, 1]$ for all $w \in \Wcal$ and $z \in \Zcal$. When $k = m$ (where $m$ divides $n$), the $(m, k)$-SIPCIMI bound gives a new LOFO-SCMI bound as
		\begin{align}
			\gen &\leq \sum_{i=1}^{m} \sum_{j=1}^{\frac{n}{m}+1} \EE \left[ \sqrt{ \frac{1}{2nm} I^{\hat{Z}^{(i)}_j}\left( W; \bar{U}^{(i)} \right)} \right] \label{eq:dis-lofo-SCMI} \\
			&\leq \sum_{i=1}^{m} \sum_{j=1}^{\frac{n}{m}+1} \sqrt{ \frac{1}{2nm} I\left( W; \bar{U}^{(i)} | \hat{Z}^{(i)}_j \right)}.
			\label{eq:lofo-SCMI}
		\end{align}
	\end{corollary}
	
	\begin{corollary}[($(1, 1)$-SIPCIMI) LOO-SCMI Bound]
		Consider the LOO setting in Definition~\ref{def:loo-supersample} and assume $\ell(w, z) \in [0, 1]$ for all $w \in \Wcal$ and $z \in \Zcal$. When $m = 1$ and $k = 1$, the $(m, k)$-SIPCIMI bound gives a new LOO-SCMI bound as
		\begin{align}
			\gen &\leq \sum_{i=1}^{n+1} \EE \left[ \sqrt{ \frac{1}{2n} I^{\hat{Z}_i}\left( W; U \right)} \right]	\label{eq:dis-loo-SCMI} \\
			&\leq \sum_{i=1}^{n+1} \sqrt{ \frac{1}{2n} I\left( W; U | \hat{Z}_i \right)}.
			\label{eq:loo-SCMI}
		\end{align}
	\end{corollary}
	
	Next, we explore relationships of the above specialized bounds with different parameters. When Assumption~\ref{ass:invariant-alg} is satisfied, we find the following proposition with its proof provided in Appendix~\ref{subapp:invariant-SCMI}.
	
	\begin{proposition} \label{pro:invariant-SCMI}
		Consider any given constant $m$ and suppose $\Acal$ satisfies Assumption~\ref{ass:invariant-alg}. For any given $m$ and $k_1, k_2 \in \Kcal$, for all $i_1 \in [k_1]$, $i_2 \in [k_2]$, $j_1 \in [\frac{n+m}{k_1}]$, $j_2 \in [\frac{n+m}{k_2}]$ and $z \in \Zcal$, we have
		\begin{align}
			I^{\hat{Z}^{(i_1)}_{j_1} = z}\left( W; U^{(i_1)}_{[\frac{n}{k_1}]} \right) = I^{\hat{Z}^{(i_2)}_{j_2} = z}\left( W; U^{(i_2)}_{[\frac{n}{k_2}]} \right). \label{eq:invariant-SCMI-p2}
		\end{align}
	\end{proposition}
	Proposition~\ref{pro:invariant-SCMI} indicates that, if Assumption~\ref{ass:invariant-alg} is satisfied, the free parameter $k$ has no impact on the SIPCIMI bound, i.e., the SIPCIMI \eqref{eq:dis-SIPCIMI}, L$m$O-SCMI \eqref{eq:dis-lmo-SCMI} and LOFO-SCMI \eqref{eq:dis-lofo-SCMI} bounds with the same $m$ are equivalent in this case\footnote{It is straightforward to verify that the integrated versions of the mentioned bounds are equivalent, too.}. Moreover, Proposition~\ref{pro:invariant-SCMI} also enables a simplification of \eqref{eq:optimal-SIPCIMI}. Specifically, since the optimal $(m,k)$ relies only on $m$ if the learning algorithm satisfies Assumption~\ref{ass:invariant-alg}, we can focus on the case of $k=1$, i.e., the L$m$O-SCMI bound \eqref{eq:lmo-SCMI}, and \eqref{eq:optimal-SIPCIMI} is therefore simplified as
	\begin{align}
		\gen \leq \inf_{m \in \NN} \left\{ (n+m) \EE \left[ \sqrt{ \frac{1}{2nm} I^{\hat{Z}_1}\left( W; U_{[n]} \right)} \right] \right\}. \label{eq:optimal-m-SCMI}
	\end{align}
	
	\begin{figure}[ht]
		\centering
		\includegraphics[width=0.8\textwidth]{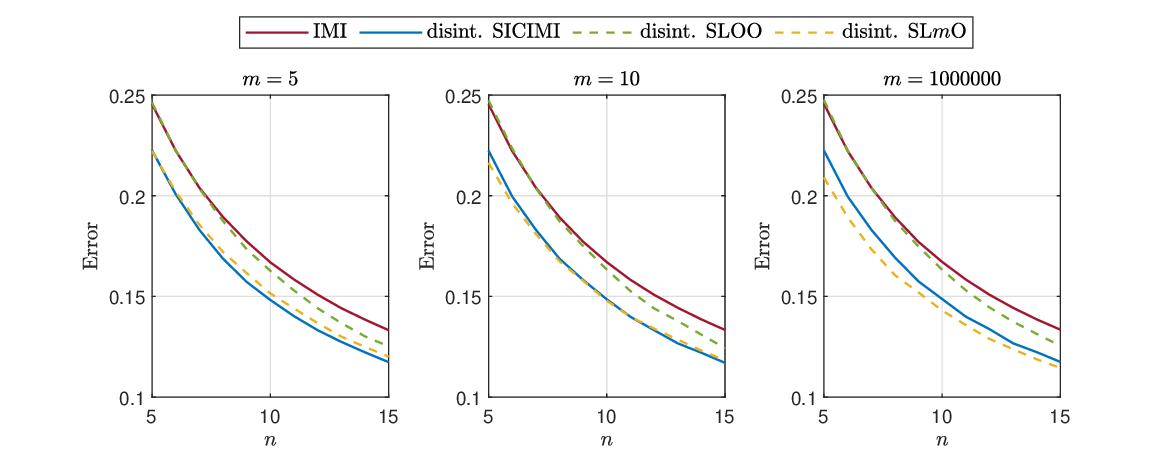}
		\caption{Comparison of various bounds for Example~\ref{eg:Bernoulli}, where we set $p=0.25$. As shown in the figure, our proposed disintegrated L$m$O-SCMI bound \eqref{eq:dis-lmo-SCMI} improves with increasing $m$ and outperforms the existing SICIMI bound \eqref{eq:dis-SICIMI} when $m$ tends to infinity. Again, we note that $m$ can be freely chosen for analysis and is not a parameter of the underlying learning problem. }
		\label{fig:SCMI_Bernoulli}
	\end{figure}
	
	Since Assumption~\ref{ass:invariant-alg} is automatically satisfied under the Bernoulli example (Example~\ref{eg:Bernoulli}), a straightforward application of \eqref{eq:optimal-m-SCMI} is thus feasible to obtain improved results. In this case, we select the proposed disintegrated L$m$O-SCMI bound \eqref{eq:dis-lmo-SCMI} along with the disintegrated LOO-SCMI bound \eqref{eq:dis-loo-SCMI} and compare them with the IMI bound \eqref{eq:IMI} and disintegrated SICIMI bound \eqref{eq:dis-SICIMI}. Fig.~\ref{fig:SCMI_Bernoulli} presents the numerical results obtained through Monte-Carlo simulations. Note that the special expressions of all the considered SCMI quantities can be found in Appendix~\ref{subapp:cal_SICIMI} to \ref{subapp:cal_LOFO_SCMI} (e.g., \eqref{eq:Ber_SICIMI} for SICIMI, \eqref{eq:Ber_LOO_SCMI} for LOO-SCMI and \eqref{eq:Ber_LmO_SCMI} for L$m$O-SCMI and LOFO-CMI\footnote{Although these results are for integrated CMI quantities, their disintegrated counterparts can be calculated analogously and are therefore omitted.}.). We observe that the proposed L$m$O-SCMI bound tends to gain better performance as $m$ increases, especially when $m \gg n$, in which case the L$m$O-SCMI bound obtains the tightest result. Besides, it is shown in Fig.~\ref{fig:SCMI_Bernoulli} that the new LOO-SCMI bound can converge, in contrast to the case in Proposition~\ref{pro:LOO-CMI} where the LOO-CMI bound is non-vanishing. This implies that the SCMI quantity is order-wise tighter as compared to the CMI quantity, demonstrating the benefit of the SCMI framework.

	\section{Extensions}\label{sec:extensions}
	The generalization bounds presented so far are all based on the hypothesis $W$ and loss difference $\gen$. In this section, we extend these results by measuring the generalization error with different approaches, and the resulting new bounds are tighter and more friendly to computation as compared to the conventional ones.
	
	\subsection{Beyond Hypothesis-Based Bounds}
	For the simplicity of notions, let $\Lambda_i := \ell(W, \hat{Z}_i)$ denote the loss that the hypothesis $W$ incurs on the supersample $\hat{Z}_i$. Analogous to the definition of $\hat{Z}^{(i)}_{[\frac{n}{k}]}$, it is clear from the context of the partitioned L$m$O supersample setting that, the shorthand $\Lambda^{(i)}_{[\frac{n+m}{k}]} = \{\Lambda^{(i)}_j\}_{j=1}^{\frac{n+m}{k}}$ denotes the collection of all individual losses in the $i$-th partition, where $\Lambda^{(i)}_j = \ell(W, \hat{Z}^{(i)}_j)$ denotes the loss with respect to $\hat{Z}^{(i)}_j$. With these notions, we have $\widehat{L}_n$ and $L_\mu$ rewritten as
	\begin{align}
		\widehat{L}_n = \frac{1}{k} \sum_{i=1}^k \EE_{\Lambda^{(i)}_{[\frac{n+m}{k}]}, U^{(i)}_{[\frac{n}{k}]}} \left[ \frac{k}{n} \sum_{j \in U^{(i)}_{[\frac{n}{k}]}} \Lambda^{(i)}_j \right], \label{eq:lb-IP-exp-emp-risk}\\
		L_\mu = \frac{1}{k} \sum_{i=1}^k \EE_{\Lambda^{(i)}_{[\frac{n+m}{k}]}, U^{(i)}_{[\frac{n}{k}]}} \left[ \frac{k}{m} \sum_{j \in \bar{U}^{(i)}_{[\frac{n}{k}]}} \Lambda^{(i)}_j \right]. \label{eq:lb-IP-exp-pop-risk}
	\end{align}
	Hence, the derivation of the previously proposed (hypothesis-based) bounds can be adapted by expressing CMI quantities via the incurred losses $\Lambda^{(i)}_j$, without an explicit mention of the hypothesis $W$. Furthermore, these prediction-based bounds can be strengthened by another approach initiated in \cite{hellstrom2022new}. In particular, instead of directly bounding the generalization error gap $\gen$, \cite{hellstrom2022new, hellstrom2024comparing, dongexactly} established bounds for various discrepancy measures between the expected empirical and population risks, thus often yielding tighter characterizations of the expected population risk $L_\mu$ when the expected empirical risk $\widehat{L}_n$ is small. 
	
	Inspired by the above-mentioned works, we choose the weighted binary JS divergence (defined in \eqref{eq:d_js}) as the discrepancy measure between $L_\mu$ and $\widehat{L}_n$ to yield the following new prediction-based bounds. The complete proof is relegated to Appendix~\ref{subapp:lb-gen-bound}.
	
	\begin{theorem}[JS-Based Bounds Evaluated on Predictions]\label{thm:lb-gen-bound}
		Consider the partitioned L$m$O supersample setting in Definition~\ref{def:p-lmo}. Let $\theta = \frac{n}{n+m}$ and assume $\ell(w, z) \in [0, 1]$ for all $w \in \Wcal$ and $z \in \Zcal$, then
		\begin{align}
			d^\theta_{\textnormal{JS}} (\widehat{L}_n \| L_\mu) &\leq \frac{2}{n+m} \sum_{i=1}^{k} I\left( \Lambda^{(i)}_{[\frac{n+m}{k}]}; U^{(i)}_{[\frac{n}{k}]} \right), \label{eq:lb-gen-bound} \\
			d^\theta_{\textnormal{JS}} (\widehat{L}_n \| L_\mu) &\leq \frac{1}{n+m} \sum_{i=1}^{k} \sum_{j=1}^{\frac{n+m}{k}} I \left( \Lambda^{(i)}_j; U^{(i)}_{[\frac{n}{k}]} \right). \label{eq:lb-gen-bound-single}
		\end{align}
	\end{theorem}
	
	Using the inverse function $d^{-1}_{\textnormal{JS}}(\cdot, \cdot, \cdot)$ defined in \eqref{eq:inv-d_js}, $L_\mu$ can be upper bounded by
	\begin{equation}
		\begin{aligned}
			L_\mu \leq d^{-1}_{\textnormal{JS}} \left( \frac{n}{n+m}, \widehat{L}_n, c \right),
		\end{aligned}
	\end{equation}
	where $c$ can be the RHS of \eqref{eq:lb-gen-bound} or \eqref{eq:lb-gen-bound-single}. Notably, the conventional hypothesis-based bounds corresponding to \eqref{eq:lb-gen-bound} and \eqref{eq:lb-gen-bound-single} are the IPCIMI bound \eqref{eq:IPCIMI} and SIPCIMI bound \eqref{eq:SICIMI}, respectively.
	
	\begin{remark}
		We provide below a variant IPCIMI bound, which is obtained by using the same proof techniques as for \eqref{eq:lb-gen-bound}. We still adopt the assumption in Theorem~\ref{thm:lb-gen-bound}, i.e., $\ell(w, z) \in [0, 1]$ for all $w \in \Wcal$ and $z \in \Zcal$, then
		\begin{equation}
			\begin{aligned}
				\gen \leq \frac{1}{k} \sum_{i=1}^{k} \sqrt{ \frac{k(n+m)}{nm} I \left( W; U^{(i)}_{[\frac{n}{k}]} | \tilde{Z}^{(i)}_{[\frac{n+m}{k}]} \right) }.
			\end{aligned}\label{eq:var-IPCIMI}
		\end{equation}
		A proof sketch for this bound is given in Appendix~\ref{subapp:var-IPCIMI}. As compared to \eqref{eq:IPCIMI}, the new variant IPCIMI bound in \eqref{eq:var-IPCIMI} eliminates a leading constant $\sqrt{\frac{1}{2}C_{n,m}^k}$. This implies that the variant IPCIMI bound \eqref{eq:var-IPCIMI} can be slightly tighter under some parameter settings (e.g., $\frac{1}{2}C_{n,m}^k = \frac{4}{3} > 1$ with $m=n$ and $k=\frac{n}{2}$), while the original bound \eqref{eq:IPCIMI} adopts less limiting assumption that allows for broader application scenarios.
	\end{remark}
	The advantages of prediction-based bounds lie in two aspects. First, due to the fact that $\Lambda^{(i)}_j$ is a processed function of $(W, \hat{Z}^{(i)}_j)$, it follows directly from the data-processing inequality that $I(\Lambda^{(i)}_j; U^{(i)}_{[\frac{n}{k}]}) \leq I(W, \hat{Z}^{(i)}_j; U^{(i)}_{[\frac{n}{k}]}) = I(W; U^{(i)}_{[\frac{n}{k}]} | \hat{Z}^{(i)}_j)$, showing that the prediction-based bounds are in general tighter than their hypothesis-based counterparts. Second, since the loss function maps the high-dimensional hypothesis space $\Wcal$ to $\RR^+$, making the estimation of the MI between $\Lambda^{(i)}_j$ and $U^{(i)}_{[\frac{n}{k}]}$ can be an easier process in practical learning scenarios (e.g., deep neural network) \cite{harutyunyan2021information,wang2023tighter}. We also note that these advantages come at the cost of losing the ability to characterize the generalization behavior of the hypothesis itself\textemdash prediction-based bounds apply only to a specific loss function, whereas hypothesis-based bounds require only knowledge of the hypothesis itself and therefore apply to arbitrary loss functions.
	
	\subsection{Exactly Tight Bounds}
	Inspired by \cite{haghifam2022understanding,wang2023tighter,dongexactly}, we also obtain a series of exactly tight bounds under the proposed supersample setting, with the proof provided in Appendix~\ref{subapp:tightest-gen-bound}.
	
	\begin{theorem}[Bound for the 0-1 Loss Function]\label{thm:tightest-gen-bound}
		Consider the partitioned L$m$O supersample setting in Definition~\ref{def:p-lmo}. Let $\theta = \frac{n}{n+m}$, if Assumption~\ref{ass:invariant-alg} is satisfied and $\ell(\cdot, \cdot) \in \{0, 1\}$, then
		\begin{align}
			d^\theta_{\textnormal{JS}} ( \widehat{L}_n \| L_\mu ) = \frac{1}{n+m} \sum_{i=1}^{k} \sum_{j=1}^{\frac{n+m}{k}} I\left( \Lambda^{(i)}_j; U^{(i)}_{[\frac{n}{k}]} \right)
		\end{align}
	\end{theorem}
	
	\begin{figure}[t]
		\centering
		\includegraphics[width=0.5\textwidth]{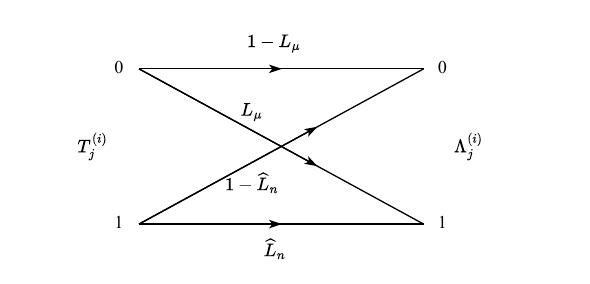}
		\caption{The transition diagram of the BAC channel from $T^{(i)}_j$ to $\Lambda^{(i)}_j$. }
		\label{fig:BAC}
	\end{figure}
	
	Inspired by \cite{wang2023tighter}, we also relate the result in \eqref{eq:tightest_bound} with a special communication setting, where $P_{\Lambda^{(i)}_j | T^{(i)}_j}$ characterizes a memoryless channel with a binary input $T^{(i)}_j$ and a binary output $\Lambda^{(i)}_j$. First, we consider $T^{(i)}_j = 1$ in which case $\hat{Z}^{(i)}_j$ is selected for training, then 
	\begin{equation}
		\begin{aligned}
			\widehat{L} &\overset{\eqref{eq:exp-emp-risk-single}}{=} \frac{1}{n+m} \sum_{i=1}^{k}\sum_{j=1}^{\frac{n+m}{k}} \EE_{\Lambda^{(i)}_j | T^{(i)}_j=1} \left[ \Lambda^{(i)}_j \right] \\
			&\overset{(a)}{=} \frac{1}{n+m} \sum_{i=1}^{k}\sum_{j=1}^{\frac{n+m}{k}} P_{\Lambda^{(i)}_j | T^{(i)}_j} (1|1) \\
			&\overset{(b)}{=} P_{\Lambda^{(i)}_j | T^{(i)}_j} (1|1),
		\end{aligned}
	\end{equation}
	where $(a)$ holds since $\Lambda^{(i)}_j \in \{0,1\}$, $(b)$ uses the fact that $P_{\Lambda^{(i)}_j | T^{(i)}_j}$ is fixed for all valid $i,j$ under Assumption~\ref{ass:invariant-alg}. This equation specifies that $\widehat{L}_n$ is exactly the error rate of the 1 to 1 transition. Likewise, when $T^{(i)}_j = 0$, it follows that $L_\mu$ is the 0 to 1 cross-over error rate. In summary, the transition matrix is given as
	\begin{equation}
		\begin{aligned}
			P_{\Lambda^{(i)}_j | T^{(i)}_j} = \begin{bmatrix} P_{\Lambda^{(i)}_j | T^{(i)}_j}(0|0) & P_{\Lambda^{(i)}_j | T^{(i)}_j}(1|0) \\ P_{\Lambda^{(i)}_j | T^{(i)}_j}(0|1) & P_{\Lambda^{(i)}_j | T^{(i)}_j}(1|1) \end{bmatrix} = \begin{bmatrix} 1-L_\mu & L_\mu \\ 1-\widehat{L}_n & \widehat{L}_n \end{bmatrix},
		\end{aligned}
	\end{equation}
	which specifies a binary asymmetric channel (BAC) model \cite{moser2009error}, as depicted in Fig.~\ref{fig:BAC}. The achievable rate of this BAC channel is $I(\Lambda^{(i)}_j; T^{(i)}_j)$, which is calculated as
	\begin{equation}
		\begin{aligned}
			I(\Lambda^{(i)}_j; T^{(i)}_j) &= H(\Lambda^{(i)}_j) - H(\Lambda^{(i)}_j | T^{(i)}_j) \\
			&= H(\theta \widehat{L}_n + (1-\theta) L_\mu) - \theta H(\widehat{L}_n) - (1-\theta) H(L_\mu) \\
			&= d^\theta_{\textnormal{JS}} ( \widehat{L}_n \| L_\mu ),
		\end{aligned}
	\end{equation}
	where $\theta = \frac{n}{n+m}$. Thus, depending on the input distribution $P_{T^{(i)}_j}$ (which is determined by $m$ and $n$), the communication rate over the BAC channel can be exactly characterized by the proposed binary JS divergence. This implies that we can derive the tightest bound for each supersample setting, extending the results under the standard and LOO supersample settings \cite{wang2023tighter,haghifam2022understanding}. Moreover, the previous exact tight bounds only hold for interpolating algorithms with the 0-1 loss function, while our bound also holds for any learning algorithm that satisfies Assumption~\ref{ass:invariant-alg}. Besides, the incorporated parameter $m$ allows us to flexibly choose a more efficient approach to evaluate the exact population risk, as shown in the following remark.
	\begin{remark}
		From Theorem~\ref{thm:tightest-gen-bound}, we thus have various expressions of generalization characterizations that hold with equality, implying that any two characterizations must be equivalent. For instance, we pick $(m,k)=(1,1)$ and $(m,k)=(n,n)$, corresponding to the LOO-CMI (Definition~\ref{def:loo-supersample}) and standard CMI (Definition~\ref{def:standard-supersample}) settings, respectively. Assume $\widehat{L}_n \leq L_\mu$ (which is typically satisfied by the ERM learning algorithm), then it follows that
		\begin{align}
			d^{-1}_{\textnormal{JS}} \left( \frac{1}{n+1}, \widehat{L}_n, I\left( \Lambda_i; U \right) \right) = d^{-1}_{\textnormal{JS}} \left( \frac{1}{2}, \widehat{L}_n, I\left( \Lambda_i; R_i \right) \right). \label{eq:tightest_bound}
		\end{align}
		This dual perspective offers a flexibility to choose whichever representation is easier to analyze in specific applications.
	\end{remark}

	\section{Conclusions}\label{sec:conclusion}
	This paper proposes a new information-theoretic framework to characterize the generalization behavior of learning algorithms. Particularly, we present a new construction of the supersample set, shedding light on a unified view and derivation of the known MI- and CMI-based bounds. The incorporated free parameters enable reproducing many existing results, facilitating new bounds and a joint optimization over a bunch of bounds. The advantages of the proposed bounds are demonstrated by applying them to various learning scenarios. Additionally, we advance our framework under the premise that the loss function is bounded, establishing potentially sharper bounds. Lastly, our framework is integrated with other methods to yield other advantages in terms of sharpness or effectiveness.
	
	\newpage
	
	{
		\bibliographystyle{IEEEtran}
		\bibliography{IEEEabrv,mybibfile}
	}
	
	\onecolumn
	\appendices
	
	\section{Prerequisite Lemmas} \label{app:lemmas}
	\begin{lemma}[Hoeffding's Lemma {\cite[Lemma~2.2.1]{raginsky2013concentration}}] \label{lem:Hoeffding}
		Let $X$ be a random variable, such that $X \in [a, b]$ almost surely for some finite constants $a \leq b$. Then, for all $\lambda \in \RR$,
		\begin{equation}
			\begin{aligned}
				\EE \left[ e^{\lambda \left( X - \EE [X] \right)} \right] \leq e^{-\frac{\lambda^2 (b-a)^2}{8}}.
			\end{aligned}
		\end{equation}
	\end{lemma}
	
	\begin{lemma}[Conditioned Hoeffding's Lemma\footnote{We note that this lemma can be proved with very similar methods as the standard form of Hoeffding's Lemma (Lemma~\ref{lem:Hoeffding}), as has also been noted in {\cite[Lemma~1]{lecturenote}}. Thus, we omit the proof.}] \label{lem:con-Hoeffding}
		Let $X, Y$ be random variables. Assume there exist a function $\psi$ and a finite constant $c$ such that $\psi(Y) \leq X \leq \psi(Y) + c$ almost surely. Then, for all $\lambda \in \RR$,
		\begin{equation}
			\begin{aligned}
				\EE \left[ \left. e^{\lambda \left( X - \EE [X | Y] \right)} \right| Y \right] \leq e^{-\frac{\lambda^2 c^2}{8}}.
			\end{aligned}
		\end{equation}
	\end{lemma}
	
	\begin{lemma}[McDiarmid's Inequality {\cite[Theorem~2.2.2]{raginsky2013concentration}}] \label{lem:McDiarmid}
		Let $\{X_i\}_{i=1}^n$ be independent (but not necessarily identically distributed) random variables taking values in a measurable space $\Xcal$. Let $f : \Xcal^n \to \RR$ be a measurable function such that for each $i \in [n]$,
		\begin{equation*}
			\sup\limits_{x_1,\ldots,x_n,x'_i \in \Xcal} \left| f(x_1, \ldots, x_{i-1}, x_i, x_{i+1}, \ldots, x_n) - f(x_1, \ldots, x_{i-1}, x'_i, x_{i+1}, \ldots, x_n) \right| \leq d_i.
		\end{equation*}
		Then, for all $\epsilon \geq 0$,
		\begin{equation}
			\PP \left\{ \left| f(X_1, \ldots, X_n) - \EE \left[ f(X_1, \ldots, X_n) \right] \right| \geq \epsilon \right\} \leq 2 \exp\left( - \frac{2 \epsilon^2}{\sum_{i=1}^n d_i^2} \right).
		\end{equation}
	\end{lemma}
	
	\begin{lemma}[Donsker-Varadhan Variational Representation of KL Divergence {\cite[Theorem~4.6]{polyanskiy2025information}}]\label{lem:DV-formula}
		Let $P$ and $Q$ be two probability distributions defined on a common measurable space $\Xcal$ such that $P \ll Q$. Then, for any $f : \Xcal \to \RR$ such that $f(X)$ and $e^{f(X)}$ are $P$- and $Q$-integrable, respectively,
		\begin{align}
			D_{\textnormal{KL}}(P \| Q) = \sup_f \left\{ \EE_{X \sim P} [f(X)] - \log \EE_{X \sim Q} \left[ e^{f(X)} \right] \right\}.
		\end{align}
	\end{lemma}
	
	\begin{lemma}[{\cite[Section~IV-B]{zhou2022individually}}] \label{lem:inv_conj}
		Consider the conditional CGF defined in Definition~\ref{def:con-CGF}. The inverse of its Fenchel conjugate function is defined as $\psi^{*-1}_{X|Y=y}(\eta) = \inf\limits_{\lambda>0} \frac{\eta + \psi_{X|Y=y}(\lambda)}{\lambda}$, which is concave and non-decreasing in $\eta > 0$.
	\end{lemma}
	
	\begin{lemma} \label{lem:more-info-by-condi}
		Let $A, B, C, D$ be random variables, where $B, C, D$ are mutually independent, then
		\begin{equation*}
			\begin{aligned}
				I( A; B | C ) \leq I( A; B | C, D ).
			\end{aligned}
		\end{equation*}
	\end{lemma}
	\begin{proof}
		Since $B, C, D$ are mutually independent, we get $I(B; C) = 0$ and $I(B; C, D) = 0$. Then, the chain rule of mutual information implies that
		\begin{align}
			&I(A, C; B) = I( A; B | C ) + I(B; C) = I( A; B | C ), \\
			&I( A, C, D; B ) = I( A; B | C, D ) + I(B; C, D) = I( A; B | C, D ).
		\end{align}
		It follows that
		\begin{equation*}
			\begin{aligned}
				I( A; B | C, D ) - I( A; B | C ) = I( A, C, D; B ) - I(A, C; B) = I( D; B | A, C ) \geq 0,
			\end{aligned}
		\end{equation*}
		where the inequality is due to the non-negativity of conditional mutual information. 
	\end{proof}
	
	\begin{lemma}[Pinsker’s Inequality {\cite[Theorem~7.10]{polyanskiy2025information}}]\label{lem:Pinsker-inequality}
		Let $P$ and $Q$ be two probability distributions defined on a common measurable space $\Xcal$, then
		\begin{align}
			D_{\textnormal{TV}} (P \| Q) \leq \sqrt{\frac{D_{\textnormal{KL}} (P \| Q)}{2}}.
		\end{align}
	\end{lemma}
	
	\begin{lemma}[{\cite[Lemma~2]{hellstrom2022new}}]\label{lem:binary-relative-entropy-concentration}
		Let $X_1, \ldots, X_n$ be independent random variables which satisfy $\EE[X_i] = \mu_i$ and $X_i \in [0,1]$ for each $i \in [n]$, then for any $\gamma \in \RR$,
		\begin{align}
			\EE \left[ e^{d_\gamma\left(\left. \frac{1}{n} \sum_{i=1}^{n} X_i \right\| \frac{1}{n} \sum_{i=1}^{n} \mu_i \right)} \right] \leq 1,
		\end{align}
		where $d_\gamma(\cdot)$ was defined in \eqref{eq:def-d_gamma}.
	\end{lemma}

	\section{Proofs in Section~\ref{sec:unified-CMI}} \label{app:proof-unified-CMI}
	\subsection{Proof of Lemma~\ref{lem:IP-lmo-gen-err}} \label{appsub:lemma1}
	It holds by linearity of expectation that
	\begin{equation}
		\begin{aligned}
			\EE_{W, \hat{Z}_{[n+m]}, U_{[n]}} \left[ \Ecal_{\textnormal{L$m$O}}( W, \hat{Z}_{[n+m]}, U_{[n]} ) \right] &= \frac{1}{k} \sum_{i=1}^{k} \EE_{W, \hat{Z}^{(i)}_{[\frac{n+m}{k}]}, U^{(i)}_{[\frac{n}{k}]}} \left[ \Ecal^{(i)}\left( W, \hat{Z}^{(i)}_{[\frac{n+m}{k}]}, U^{(i)}_{[\frac{n}{k}]} \right) \right] \\
			&= \frac{1}{k} \sum_{i=1}^{k} \EE_{W, \hat{Z}^{(i)}_{[\frac{n+m}{k}]}, U^{(i)}_{[\frac{n}{k}]}} \left[ L_{\hat{Z}_{\bar{U}^{(i)}_{[\frac{m}{k}]}}}(W) - L_{\hat{Z}_{U^{(i)}_{[\frac{n}{k}]}}}(W) \right].
		\end{aligned}
	\end{equation}
	Let us condition on $U_{[n]} = u_{[n]}$, then
	\begin{equation}
		\begin{aligned}
			\frac{1}{k} \sum_{i=1}^{k} \EE_{W, \hat{Z}^{(i)}_{[\frac{n+m}{k}]} | u^{(i)}_{[\frac{n}{k}]}} \left[ L_{\hat{Z}_{U^{(i)}_{[\frac{n}{k}]}}}(W) \right] &= \frac{1}{k} \sum_{i=1}^{k} \EE_{W, \hat{Z}_{u^{(i)}_{[\frac{n}{k}]}}, \hat{Z}_{\bar{u}^{(i)}_{[\frac{m}{k}]}}} \left[ \frac{k}{n} \sum_{j \in u^{(i)}_{[\frac{n}{k}]}} \ell(W, \hat{Z}^{(i)}_j) \right] \\
			&\overset{(a)}{=} \sum_{i=1}^{k} \EE_{W, Z^{(i)}_{[\frac{n}{k}]}} \left[ \frac{1}{n} \sum_{j = 1}^{\frac{n}{k}} \ell(W, Z^{(i)}_j) \right] \\
			&\overset{(b)}{=} \EE_{W, Z_{[n]}} \left[ \frac{1}{n} \sum_{i = 1}^{n} \ell(W, Z_i) \right] \\
			&= \widehat{L}_n,
		\end{aligned}
	\end{equation}
	and similarly,
	\begin{equation}
		\begin{aligned}
			\frac{1}{k} \sum_{i=1}^{k} \EE_{W, \hat{Z}^{(i)}_{[\frac{n+m}{k}]} | u^{(i)}_{[\frac{n}{k}]}} \left[ L_{\hat{Z}_{\bar{U}^{(i)}_{[\frac{m}{k}]}}}(W) \right] &= \frac{1}{k} \sum_{i=1}^{k} \EE_{W, \hat{Z}_{u^{(i)}_{[\frac{n}{k}]}}, \hat{Z}_{\bar{u}^{(i)}_{[\frac{m}{k}]}}} \left[ \frac{k}{m} \sum_{j \in \bar{u}^{(i)}_{[\frac{m}{k}]}} \ell(W, \hat{Z}^{(i)}_j) \right] \\
			&= \sum_{i=1}^{k} \EE_{W} \EE_{\hat{Z}_{\bar{u}^{(i)}_{[\frac{m}{k}]}}} \left[ \frac{1}{m} \sum_{j \in \bar{u}^{(i)}_{[\frac{m}{k}]}} \ell(W, \hat{Z}^{(i)}_j) \right] \\
			&= \EE_{W} \EE_{Z} \left[ \ell(W, Z) \right] \\
			&= L_\mu,
		\end{aligned}
	\end{equation}
	where $(a)$ holds with $Z^{(i)}_{[\frac{n}{k}]} = \hat{Z}_{U^{(i)}_{[\frac{n}{k}]}}$ and $(b)$ is due to $Z_{[n]} = \{Z^{(i)}_{[\frac{n}{k}]}\}_{i=1}^k$, as given in Definition~\ref{def:p-lmo}. By taking expectation over $U_{[n]}$, we have
	\begin{align}
		&\widehat{L}_n = \frac{1}{k} \sum_{i=1}^{k} \EE_{W, \hat{Z}^{(i)}_{[\frac{n+m}{k}]}, U^{(i)}_{[\frac{n}{k}]}} \left[ L_{\hat{Z}_{U^{(i)}_{[\frac{n}{k}]}}}(W) \right], \label{eq:emp_risk_lmo}\\
		&L_\mu = \frac{1}{k} \sum_{i=1}^{k} \EE_{W, \hat{Z}^{(i)}_{[\frac{n+m}{k}]}, U^{(i)}_{[\frac{n}{k}]}} \left[ L_{\hat{Z}_{\bar{U}^{(i)}_{[\frac{m}{k}]}}}(W) \right], \label{eq:pop_risk_lmo}
	\end{align}
	which completes proof.
	
	\subsection{Proof of Theorem~\ref{thm:IPCIMI}}\label{appsub:theorem6}
	We first state an essential lemma.
	\begin{lemma} \label{lem:symmetry} 
		Consider the partitioned L$m$O supersample setting in Definition~\ref{def:p-lmo}. For all $i \in [k]$, $w \in \Wcal$ and $\hat{z}^{(i)}_{[\frac{n+m}{k}]} \in \Zcal^{\frac{n+m}{k}}$, we have
		\begin{align}
			\EE_{U^{(i)}_{[\frac{n}{k}]}} \left[ \Ecal^{(i)} \left( w, \hat{z}^{(i)}_{[\frac{n+m}{k}]}, U^{(i)}_{[\frac{n}{k}]} \right) \right] = 0.
			\label{eq:symmetry}
		\end{align}
	\end{lemma}
	
	\begin{proof}
		It follows by the definition of $\Ecal^{(i)}$ that
		\begin{equation}
			\begin{aligned}	
				\EE_{U^{(i)}_{[\frac{n}{k}]}} \left[ \Ecal^{(i)} \left( w, \hat{z}^{(i)}_{[\frac{n+m}{k}]}, U^{(i)}_{[\frac{n}{k}]} \right) \right] &= \frac{1}{\begin{psmallmatrix} \frac{n+m}{k} \\ \frac{n}{k} \end{psmallmatrix}} \sum_{u^{(i)}_{[\frac{n}{k}]}} \left( \frac{k}{m} \sum_{i \not\in u^{(i)}_{[\frac{n}{k}]}} \ell(w, \hat{z}_i) - \frac{k}{n} \sum_{i \in u^{(i)}_{[\frac{n}{k}]}} \ell(w, \hat{z}_i) \right) \\
				&\overset{(a)}{=} \frac{1}{\begin{psmallmatrix}	\frac{n+m}{k} \\ \frac{n}{k} \end{psmallmatrix}} \sum_{i=1}^{\frac{n+m}{k}} \left( \frac{k}{m} \sum_{u^{(i)}_{[\frac{n}{k}]}: i \not\in u^{(i)}_{[\frac{n}{k}]}} \ell(w, \hat{z}_i) - \frac{k}{n} \sum_{u^{(i)}_{[\frac{n}{k}]}: i \in u^{(i)}_{[\frac{n}{k}]}} \ell(w, \hat{z}_i) \right) \\
				&= \frac{1}{\begin{psmallmatrix}	\frac{n+m}{k} \\ \frac{n}{k} \end{psmallmatrix}} \sum_{i=1}^{\frac{n+m}{k}} \left( \frac{k}{m} \begin{psmallmatrix}	\frac{n+m}{k}-1 \\ \frac{n}{k} \end{psmallmatrix} \ell(w, \hat{z}_i) - \frac{k}{n} \begin{psmallmatrix}	\frac{n+m}{k}-1 \\ \frac{n}{k}-1\end{psmallmatrix} \ell(w, \hat{z}_i) \right) \\
				&= \sum_{i=1}^{\frac{n+m}{k}} \left( \frac{k}{m} \cdot \frac{m}{n+m} - \frac{k}{n} \cdot \frac{m}{n+m} \right) \ell(w, \hat{z}_i) \\
				&= 0.
			\end{aligned}
		\end{equation}
		where $(a)$ is obtained by swapping the order of summation.
	\end{proof}
	
	Starting from \eqref{eq:IP-lmo-gen-err}, for all $\lambda \in \RR$, we have
	\begin{equation}
		\begin{aligned}
			\EE_{W, \hat{Z}_{[n+m]}, U_{[n]}} \left[ \lambda \Ecal_{\textnormal{L$m$O}}\left( W, \hat{Z}_{[n+m]}, U_{[n]} \right) \right] = \frac{1}{k} \sum_{i=1}^{k} \EE_{W, \hat{Z}^{(i)}_{[\frac{n+m}{k}]}, U^{(i)}_{[\frac{n}{k}]}} \left[ \lambda \Ecal^{(i)} \left( W, \hat{Z}^{(i)}_{[\frac{n+m}{k}]}, U^{(i)}_{[\frac{n}{k}]} \right) \right].
		\end{aligned}
	\end{equation}
	When conditioned on $\hat{Z}^{(i)}_{[\frac{n+m}{k}]}$, by using the Donsker-Varadhan inequality with $X = (W, U^{(i)}_{[\frac{n}{k}]})$, $P = P_{W, U^{(i)}_{[\frac{n}{k}]} | \hat{Z}^{(i)}_{[\frac{n+m}{k}]}}$, $Q = P_{W | \hat{Z}^{(i)}_{[\frac{n+m}{k}]}} \otimes P_{U^{(i)}_{[\frac{n}{k}]} | \hat{Z}^{(i)}_{[\frac{n+m}{k}]}}$ and $f(X) = \lambda \Ecal^{(i)} ( W, \hat{Z}^{(i)}_{[\frac{n+m}{k}]}, U^{(i)}_{[\frac{n}{k}]} )$, we derive an upper bound for each summand of the above expression:
	\begin{equation}
		\begin{aligned}
			&\EE_{W, U^{(i)}_{[\frac{n}{k}]} | \hat{Z}^{(i)}_{[\frac{n+m}{k}]}} \left[ \lambda \Ecal^{(i)} \left( W, \hat{Z}^{(i)}_{[\frac{n+m}{k}]}, U^{(i)}_{[\frac{n}{k}]} \right) \right] \\
			&\leq I^{\hat{Z}^{(i)}_{[\frac{n+m}{k}]}}\left( W; U^{(i)}_{[\frac{n}{k}]} \right) + \log \EE_{W | \hat{Z}^{(i)}_{[\frac{n+m}{k}]}} \EE_{U^{(i)}_{[\frac{n}{k}]} | \hat{Z}^{(i)}_{[\frac{n+m}{k}]}} \left[ \exp \left( \lambda \Ecal^{(i)} \left( W, \hat{Z}^{(i)}_{[\frac{n+m}{k}]}, U^{(i)}_{[\frac{n}{k}]} \right) \right) \right] \\
			&= I^{\hat{Z}^{(i)}_{[\frac{n+m}{k}]}}\left( W; U^{(i)}_{[\frac{n}{k}]} \right) + \log \EE_{\tilde{W}, \tilde{U}^{(i)}_{[\frac{n}{k}]} | \hat{Z}^{(i)}_{[\frac{n+m}{k}]}} \left[ \exp \left( \lambda \Ecal^{(i)} \left( \tilde{W}, \hat{Z}^{(i)}_{[\frac{n+m}{k}]}, \tilde{U}^{(i)}_{[\frac{n}{k}]} \right) \right) \right] \\
			&\overset{(a)}{=} I^{\hat{Z}^{(i)}_{[\frac{n+m}{k}]}}\left( W; U^{(i)}_{[\frac{n}{k}]} \right) + \psi_{\mathcal{E}^{(i)}(\tilde{W}, \hat{Z}^{(i)}_{[\frac{n+m}{k}]}, \tilde{U}^{(i)}_{[\frac{n}{k}]}) | \hat{Z}^{(i)}_{[\frac{n+m}{k}]}}(\lambda),
		\end{aligned}\label{eq:dv-IPCIMI}
	\end{equation}
	where $(a)$ holds by Lemma~\ref{lem:symmetry} and Definition~\ref{def:con-CGF}. It follows for any $\lambda >0$ that
	\begin{equation}
		\begin{aligned}
			\EE_{W, U^{(i)}_{[\frac{n}{k}]} | \hat{Z}^{(i)}_{[\frac{n+m}{k}]}} \left[ \Ecal^{(i)} \left( W, \hat{Z}^{(i)}_{[\frac{n+m}{k}]}, U^{(i)}_{[\frac{n}{k}]} \right) \right] &\leq \inf_{\lambda > 0} \frac{1}{\lambda} \left( I^{\hat{Z}^{(i)}_{[\frac{n+m}{k}]}}\left( W; U^{(i)}_{[\frac{n}{k}]} \right) + \psi_{\mathcal{E}^{(i)}(\tilde{W}, \hat{Z}^{(i)}_{[\frac{n+m}{k}]}, \tilde{U}^{(i)}_{[\frac{n}{k}]}) | \hat{Z}^{(i)}_{[\frac{n+m}{k}]}}(\lambda) \right).
		\end{aligned}
	\end{equation}
	The RHS expression of the above inequality is measurable by Lemma~\ref{lem:inv_conj}. Thus, by taking expectation over $\hat{Z}^{(i)}_{[\frac{n+m}{k}]}$ and summation over $i$, we obtain
	\begin{equation}
		\begin{aligned}
			\EE_{W, \hat{Z}_{[n+m]}, U_{[n]}} \left[ \Ecal_{\textnormal{L$m$O}} \left( W, \hat{Z}_{[n+m]}, U_{[n]} \right) \right] &\leq \frac{1}{k} \sum_{i=1}^{k} \EE \left[ \inf_{\lambda > 0} \frac{1}{\lambda} \left( I^{\hat{Z}^{(i)}_{[\frac{n+m}{k}]}}\left( W; U^{(i)}_{[\frac{n}{k}]} \right) + \psi_{\mathcal{E}^{(i)}(\tilde{W}, \hat{Z}^{(i)}_{[\frac{n+m}{k}]}, \tilde{U}^{(i)}_{[\frac{n}{k}]}) | \hat{Z}^{(i)}_{[\frac{n+m}{k}]}}(\lambda) \right) \right] \\
			&\leq \frac{1}{k} \sum_{i=1}^{k} \inf_{\lambda > 0} \frac{1}{\lambda} \left(I\left( W; U^{(i)}_{[\frac{n}{k}]} | \hat{Z}^{(i)}_{[\frac{n+m}{k}]} \right) + \EE \left[ \psi_{\mathcal{E}^{(i)}(\tilde{W}, \hat{Z}^{(i)}_{[\frac{n+m}{k}]}, \tilde{U}^{(i)}_{[\frac{n}{k}]}) | \hat{Z}^{(i)}_{[\frac{n+m}{k}]}}(\lambda) \right] \right),
		\end{aligned}
	\end{equation}
	where in the last step we change the order of expectation and the infimum function.
	
	\subsection{Proof of Lemma~\ref{lem:bld-CGF-bound}} \label{appsub:lemma2}
	For notational simplicity, we let
	\begin{align}
		\Ecal\left( w, \hat{z}_{[n+m]}, U_{[n]} \right) := \frac{1}{m} \sum_{i \in \bar{U}_{[m]}} \ell(w, \hat{z}_i) - \frac{1}{n} \sum_{i \in U_{[n]}} \ell(w, \hat{z}_i), \label{eq:def-cv-err}
	\end{align}
	where $U_{[n]}$ is a subset of size $n$ chosen uniformly from $[n+m]$ and $\bar{U}_{[m]} = [n+m] \backslash U_{[n]}$. We begin by stating the following lemma, where we consider $k=1$ for simplicity. 
	
	\begin{lemma} \label{lem:ub-cgf}
		Consider the L$m$O supersample setting in Definition~\ref{def:p-lmo} with $k=1$. Suppose $\sup_{w \in \Wcal} |\ell(w, z) - \ell(w, z')| \leq \Delta$ for any $z, z' \in \Zcal$, then for all $w \in \Wcal$, $\hat{z}_{[n+m]} \in \Zcal^{n+m}$ and $\lambda \in \RR$, we have
		\begin{align}
			\log \EE_{U_{[n]}} \left[ \exp \left( \lambda \Ecal\left( w, \hat{z}_{[n+m]}, U_{[n]} \right) \right) \right] \leq \frac{\lambda^2 \Delta^2 C_{n,m} \cdot (n+m)}{8nm},
			\label{eq:ub-cgf}
		\end{align}
		where $C_{n,m} = \begin{cases}
			\frac{n+m}{\max(n, m)} &  \min(n, m) = 1, \\
			\frac{n+m}{n+m-1} \cdot \frac{nm}{nm - \min(n, m)} & \text{otherwise.}
		\end{cases}$.
	\end{lemma}
	
	\begin{proof}
		We start by constructing a Doob's martingale difference sequence as
		\begin{align}
			D_i = \EE \left[ \Ecal\left( w, \hat{z}_{[n+m]}, U_{[n]} \right) | U_1, \ldots, U_i \right] - \EE \left[ \Ecal\left( w, \hat{z}_{[n+m]}, U_{[n]} \right) | U_1, \ldots, U_{i-1} \right], \quad \textnormal{for $i \in [n]$},
			\label{eq:def_Di}
		\end{align}
		where the dependence on $\{U_1, \ldots, U_i\}$ is dropped for notational simplicity. Using the tower property of conditional expectation, one can simply verify that 
		\begin{equation}
			\begin{aligned}
				&\EE \left[ D_i | U_1, \ldots, U_{i-1} \right] \\
				&= \EE \left[ \EE \left[ \Ecal\left( w, \hat{z}_{[n+m]}, U_{[n]} \right) | U_1, \ldots, U_i \right] | U_1, \ldots, U_{i-1} \right] - \EE \left[ \Ecal\left( w, \hat{z}_{[n+m]}, U_{[n]} \right) | U_1, \ldots, U_{i-1} \right] \\
				&= \EE \left[ \Ecal\left( w, \hat{z}_{[n+m]}, U_{[n]} \right) | U_1, \ldots, U_{i-1} \right] - \EE \left[ \Ecal\left( w, \hat{z}_{[n+m]}, U_{[n]} \right) | U_1, \ldots, U_{i-1} \right] \\
				&= 0.
			\end{aligned}
			\label{eq:con_exp_Di}
		\end{equation}
		In the following, we will show the computation of $D_i$, where $\EE \left[ \Ecal\left( w, \hat{z}_{[n+m]}, U_{[n]} \right) | U_1, \ldots, U_i \right]$ is computed first, followed by which $\EE \left[ \Ecal\left( w, \hat{z}_{[n+m]}, U_{[n]} \right) | U_1, \ldots, U_{i-1} \right]$ can be obtained similarly. First, by Bayes' rule, the conditional distribution $P_{U_{i+1}, \ldots, U_n | U_1, \ldots, U_i}$ can be obtained as
		\begin{equation}
			\begin{aligned}
				P_{U_{i+1}, \ldots, U_n | U_1, \ldots, U_i}(u_{i+1}, \ldots, u_n | u_1, \ldots, u_i) &= \frac{P_{U_1, \ldots, U_n}(u_1, \ldots, u_n)}{P_{U_1, \ldots, U_i}(u_1, \ldots, u_i)} \\
				&= \frac{P_{U_1, \ldots, U_n}(u_1, \ldots, u_n)}{\sum_{u'_{i+1}, \ldots, u'_{n}} P_{U_1, \ldots, U_n}(u_1, \ldots, u_i, u'_{i+1}, \ldots, u'_{n})} \\
				&= \frac{\frac{1}{\begin{psmallmatrix} n+m \\ n \end{psmallmatrix}}}{\frac{\begin{psmallmatrix} n+m-i \\ n-i \end{psmallmatrix}}{\begin{psmallmatrix} n+m \\ n \end{psmallmatrix}}} = \frac{1}{\begin{psmallmatrix} n+m-i \\ n-i \end{psmallmatrix}},
			\end{aligned}\label{eq:condi-prob-u}
		\end{equation}
		where $\begin{psmallmatrix} n+m-i \\ n-i \end{psmallmatrix}$ counts the number of choices that we select $n-i$ indices from $n+m-i$ possible positions. In general, \eqref{eq:condi-prob-u} means that when conditioned on $U_1, \ldots, U_i$, the set composed of the remaining indices $\{U_{i+1}, \ldots, U_n\}$ is also uniformly distributed among $\begin{psmallmatrix} n+m-i \\ n-i \end{psmallmatrix}$ possibilities. Let condition on $U_1 = u_1, \ldots, U_i = u_i$, it follows by the definition in \eqref{eq:lmo-cv-err} that
		\begin{equation}
			\begin{aligned}
				&\EE \left[ \Ecal\left( w, \hat{z}_{[n+m]}, U_{[n]} \right) | u_1, \ldots, u_i \right] \\
				&= \EE \left[ \frac{1}{m} \sum_{j \in [n+m] \backslash \{u_1, \ldots, u_i, U_{i+1}, \ldots, U_n\}} \ell(w, \hat{z}_j) - \frac{1}{n} \sum_{j \in \{u_1, \ldots, u_i, U_{i+1}, \ldots, U_n\}} \ell(w, \hat{z}_j) \right] \\
				&= \EE \left[ \frac{1}{m} \sum_{j \in [n+m] \backslash \{u_1, \ldots, u_i, U_{i+1}, \ldots, U_n\}} \ell(w, \hat{z}_j) - \frac{1}{n} \sum_{j \in \{U_{i+1}, \ldots, U_n\}} \ell(w, \hat{z}_j) \right] - \frac{1}{n} \sum_{j \in \{u_1, \ldots, u_i\}} \ell(w, \hat{z}_j).
			\end{aligned}
			\label{eq:condi-exp-cv-err}
		\end{equation}	
		Using \eqref{eq:condi-prob-u}, we next expand the expectations in \eqref{eq:condi-exp-cv-err} to get
		\begin{equation}
			\begin{aligned}
				&\EE \left[ \Ecal\left( w, \hat{z}_{[n+m]}, U_{[n]} \right) | u_1, \ldots, u_i \right] \\
				&= \frac{1}{\begin{psmallmatrix} n+m-i \\ n-i \end{psmallmatrix}} \sum_{u'_{i+1}, \ldots, u'_n} \left( \frac{1}{m} \sum_{j \in [n+m] \backslash \{u_1, \ldots, u_i, u'_{i+1}, \ldots, u'_n\}} \ell(w, \hat{z}_j) - \frac{1}{n} \sum_{j \in \{u'_{i+1}, \ldots, u'_n\}} \ell(w, \hat{z}_j) \right) \\
				&\quad - \frac{1}{n} \sum_{j \in \{u_1, \ldots, u_i\}} \ell(w, \hat{z}_j) \\
				&\overset{(a)}{=} \frac{1}{\begin{psmallmatrix} n+m-i \\ n-i \end{psmallmatrix}} \sum_{j \in [n+m] \backslash \{u_1, \ldots, u_i\}} \left( \sum_{\{u'_{i+1}, \ldots, u'_n\}: j \not\in  \{u'_{i+1}, \ldots, u'_n\}} \frac{1}{m} \ell(w, \hat{z}_j) - \sum_{\{u'_{i+1}, \ldots, u'_n\}: j \in  \{u'_{i+1}, \ldots, u'_n\}} \frac{1}{n} \ell(w, \hat{z}_j) \right) \\
				&\quad - \frac{1}{n} \sum_{j \in \{u_1, \ldots, u_i\}} \ell(w, \hat{z}_j) \\
				&= \sum_{j \in [n+m] \backslash \{u_1, \ldots, u_i\}} \left( \frac{\begin{psmallmatrix} n+m-i-1 \\ n-i \end{psmallmatrix}}{\begin{psmallmatrix} n+m-i \\ n-i \end{psmallmatrix}} \cdot \frac{1}{m} \ell(w, \hat{z}_j) - \frac{\begin{psmallmatrix} n+m-i-1 \\ n-i-1 \end{psmallmatrix}}{\begin{psmallmatrix} n+m-i \\ n-i \end{psmallmatrix}} \cdot \frac{1}{n} \ell(w, \hat{z}_j) \right) - \frac{1}{n} \sum_{j \in \{u_1, \ldots, u_i\}} \ell(w, \hat{z}_j) \\
				&= \sum_{j \in [n+m] \backslash \{u_1, \ldots, u_i\}} \left( \frac{m}{n+m-i} \cdot \frac{1}{m} - \frac{n-i}{n+m-i} \cdot \frac{1}{n} \right) \ell(w, \hat{z}_j) - \frac{1}{n} \sum_{j \in \{u_1, \ldots, u_i\}} \ell(w, \hat{z}_j) \\
				&= - \sum_{j \in \{u_1, \ldots, u_i\}} \frac{1}{n} \ell(w, \hat{z}_j) + \sum_{j \in [n+m] \backslash \{u_1, \ldots, u_i\}} \frac{i}{n(n+m-i)} \ell(w, \hat{z}_j),
			\end{aligned}
			\label{eq:def_Di_1}
		\end{equation}
		where in $(a)$ we exchange the order of summation. Since \eqref{eq:def_Di_1} holds for all $i$, we can simply substitute $i-1$ for $i$ to obtain
		\begin{equation}
			\begin{aligned}
				\EE \left[ \Ecal\left( w, \hat{z}_{[n+m]}, U_{[n]} \right) | u_1 \ldots, u_{i-1} \right]
				=& - \sum_{j \in \{u_1, \ldots, u_{i-1}\}} \frac{1}{n} \ell(w, \hat{z}_j) + \sum_{j \in [n+m] \backslash \{u_1, \ldots, u_{i-1}\}} \frac{i-1}{n(n+m-i+1)} \ell(w, \hat{z}_j).
			\end{aligned}
			\label{eq:def_Di_2}
		\end{equation}	
		Plugging \eqref{eq:def_Di_1} and \eqref{eq:def_Di_2} into \eqref{eq:def_Di} leads to
		\begin{equation}
			\begin{aligned}
				D_i =& - \sum_{j \in \{U_1, \ldots, U_i\}} \frac{1}{n} \ell(w, \hat{z}_j) + \sum_{j \in [n+m] \backslash \{U_1, \ldots, U_i\}} \frac{i}{n(n+m-i)} \ell(w, \hat{z}_j) \\
				& + \sum_{j \in \{U_1, \ldots, U_{i-1}\}} \frac{1}{n} \ell(w, \hat{z}_j) - \sum_{j \in [n+m] \backslash \{U_1, \ldots, U_{i-1}\}} \frac{i-1}{n(n+m-i+1)} \ell(w, \hat{z}_j) \\
				=& \left( - \frac{1}{n} - \frac{i-1}{n(n+m-i+1)} \right) \ell(w, \hat{z}_{U_i}) + \sum_{j \in [n+m] \backslash \{U_1, \ldots, U_i\}} \left( \frac{i}{n(n+m-i)} - \frac{i-1}{n(n+m-i+1)} \right) \ell(w, \hat{z}_j) \\
				=& - \frac{n+m}{n(n+m-i+1)} \ell(w, \hat{z}_{U_i}) + \sum_{j \in [n+m] \backslash \{U_1, \ldots, U_i\}} \frac{n+m}{n(n+m-i)(n+m-i+1)} \ell(w, \hat{z}_j) \\
				=& \left( - \frac{n+m}{n(n+m-i+1)} - \frac{n+m}{n(n+m-i)(n+m-i+1)} \right) \ell(w, \hat{z}_{U_i}) \\
				& + \sum_{j \in [n+m] \backslash \{U_1, \ldots, U_{i-1}\}} \frac{n+m}{n(n+m-i)(n+m-i+1)} \ell(w, \hat{z}_j) \\
				=& - \frac{n+m}{n(n+m-i)} \ell(w, \hat{z}_{U_i}) + \sum_{j \in [n+m] \backslash \{U_1, \ldots, U_{i-1}\}} \frac{n+m}{n(n+m-i)(n+m-i+1)} \ell(w, \hat{z}_j).
			\end{aligned}
			\label{eq:D_i}
		\end{equation}	
		Furthermore, let define $A_i$ and $B_i$ as the lower and upper bounds of $D_i$, respectively, such that for $i \in [n]$,
		\begin{equation}
			\begin{aligned}
				&A_i := \min\limits_{u_i \in [n+m] \backslash \{U_1, \ldots, U_{i-1}\}} \EE \left[ \Ecal\left( w, \hat{z}_{[n+m]}, U_{[n]} \right) | U_1, \ldots, U_{i-1}, U_i = u_i \right] - \EE \left[ \Ecal\left( w, \hat{z}_{[n+m]}, U_{[n]} \right) | U_1, \ldots, U_{i-1} \right], \\
				&B_i := \max\limits_{u_i \in [n+m] \backslash \{U_1, \ldots, U_{i-1}\}} \EE \left[ \Ecal\left( w, \hat{z}_{[n+m]}, U_{[n]} \right) | U_1, \ldots, U_{i-1}, U_i = u_i \right] - \EE \left[ \Ecal\left( w, \hat{z}_{[n+m]}, U_{[n]} \right) | U_1, \ldots, U_{i-1} \right].
			\end{aligned}
		\end{equation}
		Since \eqref{eq:D_i} is monotonically decreasing with $\ell(w, \hat{z}_{U_i})$, let
		\begin{equation*}
			\begin{aligned}
				u = \underset{j \in [n+m]\backslash \{U_1, \ldots, U_{i-1}\}}{\arg\max} \ell(w, \hat{z}_j), \\
				u' = \underset{j \in [n+m]\backslash \{U_1, \ldots, U_{i-1}\}}{\arg\min} \ell(w, \hat{z}_j),
			\end{aligned}
		\end{equation*}
		then we have
		\begin{equation}
			\begin{aligned}
				&A_i = \EE \left[ \Ecal\left( w, \hat{z}_{[n+m]}, U_{[n]} \right) | U_1, \ldots, U_{i-1}, U_i = u \right] - \EE \left[ \Ecal\left( w, \hat{z}_{[n+m]}, U_{[n]} \right) | U_1, \ldots, U_{i-1} \right], \\
				&B_i = \EE \left[ \Ecal\left( w, \hat{z}_{[n+m]}, U_{[n]} \right) | U_1, \ldots, U_{i-1}, U_i = u' \right] - \EE \left[ \Ecal\left( w, \hat{z}_{[n+m]}, U_{[n]} \right) | U_1, \ldots, U_{i-1} \right],
			\end{aligned}
		\end{equation}
		which means that $A_i$ and $B_i$ can be further calculated by plugging $U_i = u$ and $U_i = u'$ into \eqref{eq:D_i}, respectively. Thus, the interval $B_i - A_i$ is bounded by
		\begin{equation}
			\begin{aligned}
				B_i - A_i = \frac{n+m}{n(n+m-i)} \left( \ell(w, \hat{z}_{u}) - \ell(w, \hat{z}_{u'}) \right) \leq \frac{\Delta(n+m)}{n(n+m-i)},
			\end{aligned}\label{eq:interval-delta}
		\end{equation}	
		where the inequality is due to the boundedness of the loss difference. The rest of the argument follows similarly with the proof of Azuma-Hoeffding's inequality {\cite[Theorem~3.2.1]{roch2024modern}}, which is detailed as follows to make this paper self-contained. Specifically, by iteration, we have
		\begin{equation}	
			\begin{aligned}
				&\log \EE \left[ \exp \left( \lambda \Ecal\left( w, \hat{z}_{[n+m]}, U_{[n]} \right) \right) \right] \\
				&\overset{(a)}{=} \log \EE \left[ \exp \left( \lambda \left( \Ecal\left( w, \hat{z}_{[n+m]}, U_{[n]} \right) - \EE \left[ \Ecal\left( w, \hat{z}_{[n+m]}, U_{[n]} \right) \right] \right) \right) \right] \\
				&= \log \EE \left[ \exp \left( \lambda \sum_{i=1}^{n} D_i \right) \right] \\
				&\overset{(b)}{=} \log \EE \left[ \EE \left[ \left. \exp \left( \lambda \sum_{i=1}^{n-1} D_i \right) \exp \left( \lambda D_n \right) \right| U_1, U_2, \ldots, U_{n-1} \right] \right] \\
				&= \log \EE \left[ \exp \left( \lambda \sum_{i=1}^{n-1} D_i \right) \EE \left[ \exp \left( \lambda D_n \right) | U_1, \ldots, U_{n-1} \right] \right] \\
				&\overset{\eqref{eq:con_exp_Di}}{=} \log \EE \left[ \exp \left( \lambda \sum_{i=1}^{n-1} D_i \right) \EE \left[ \left. \exp \left( \lambda \left( D_n - \EE \left[ D_n | U_1, \ldots, U_{n-1} \right] \right) \right) \right| U_1, \ldots, U_{n-1} \right] \right] \\
				&\overset{(c)}{\leq} \log \EE \left[ \exp \left( \lambda \sum_{i=1}^{n-1} D_i \right) \right] + \frac{\lambda^2 \Delta^2 (n+m)^2}{8 n^2} \cdot \frac{1}{m^2} \\
				&\quad \vdots \\
				&\leq \frac{\lambda^2 \Delta^2 (n+m)^2}{8 n^2} \sum_{i=1}^{n} \frac{1}{(n+m-i)^2} \\
				&= \frac{\lambda^2 \Delta^2 (n+m)^2}{8 n^2} \sum_{i=m}^{n+m-1} \frac{1}{i^2}.
			\end{aligned}\label{eq:cgf-ub-n}
		\end{equation}
		where $(a)$ is due to Lemma~\ref{lem:symmetry}, $(b)$ is by the tower property of conditional expectation, $(c)$ is obtained by combining \eqref{eq:interval-delta} and Lemma~\eqref{lem:con-Hoeffding} with $X = D_n$ when conditioned on $Y = \{U_1, \ldots, U_{n-1}\}$, and the omitted procedure is with respect to the repeated steps applied to $D_{n-1}, \ldots, D_1$. The above result is established when we assume the randomness comes from $U_{[n]}$. Alternatively, if we treat $\bar{U}_{[m]}$ as the source random variable, one can verify from \eqref{eq:def-cv-err} that the following symmetry holds:
		\begin{equation}	
			\begin{aligned}
				\Ecal\left( w, \hat{z}_{[n+m]}, U_{[n]} \right) &= - \Ecal\left( w, \hat{z}_{[n+m]}, \bar{U}_{[m]} \right).
			\end{aligned}
		\end{equation}
		Thus, we can simply switch $n$ and $m$ in \eqref{eq:cgf-ub-n} to establish another upper bound as
		\begin{equation}	
			\begin{aligned}
				\log \EE \left[ \exp \left( \lambda \Ecal\left( w, \hat{z}_{[n+m]}, U_{[n]} \right) \right) \right] &= \log \EE \left[ \exp \left( - \lambda \Ecal\left( w, \hat{z}_{[n+m]}, \bar{U}_{[m]} \right) \right) \right] \\
				&\leq \frac{\lambda^2 \Delta^2 (n+m)^2}{8 m^2} \sum_{i=n}^{n+m-1} \frac{1}{i^2}. 
			\end{aligned}\label{eq:cgf-ub-m}
		\end{equation}
		Combining \eqref{eq:cgf-ub-n} and \eqref{eq:cgf-ub-m} gives a tightened upper bound as
		\begin{equation}	
			\begin{aligned}
				\log \EE \left[ \exp \left( \lambda \Ecal\left( w, \hat{z}_{[n+m]}, U_{[n]} \right) \right) \right]
				\leq \frac{\lambda^2 \Delta^2 (n+m)^2}{8} \min \left( \frac{1}{n^2} \sum_{i=m}^{n+m-1} \frac{1}{i^2}, \frac{1}{m^2} \sum_{i=n}^{n+m-1} \frac{1}{i^2} \right).
				\label{eq:cgf-ub-mn}
			\end{aligned}
		\end{equation}
		Notice that when $m = 1$, the RHS of \eqref{eq:cgf-ub-mn} can be directly computed as
		\begin{equation}	
			\begin{aligned}
				\frac{\lambda^2 \Delta^2 (n+m)^2}{8} \min \left( \frac{1}{n^2} \sum_{i=1}^{n} \frac{1}{i^2}, \frac{1}{n^2} \right) = \frac{\lambda^2 \Delta^2 (n+m)^2}{8 n^2},
			\end{aligned}
		\end{equation}
		and similarly, for $n = 1$ we have the RHS of \eqref{eq:cgf-ub-mn} simplified as
		\begin{equation}	
			\begin{aligned}
				\frac{\lambda^2 \Delta^2 (n+m)^2}{8} \min \left( \frac{1}{m^2}, \frac{1}{m^2} \sum_{i=1}^{m} \frac{1}{i^2} \right) = \frac{\lambda^2 \Delta^2 (n+m)^2}{8 m^2}.
			\end{aligned}
		\end{equation}
		Collectively, for the case $\min(n, m) = 1$, we therefore have
		\begin{equation}	
			\begin{aligned}
				\log \EE \left[ \exp \left( \lambda \Ecal\left( w, \hat{z}_{[n+m]}, U_{[n]} \right) \right) \right]
				\leq \frac{\lambda^2 \Delta^2 (n+m)^2}{8 \max(n^2, m^2)}.
			\end{aligned}\label{eq:cgf-ub-1}
		\end{equation}
		Otherwise, \eqref{eq:cgf-ub-mn} is further upper bounded by
		\begin{equation}	
			\begin{aligned}
				&\log \EE \left[ \exp \left( \lambda \Ecal\left( w, \hat{z}_{[n+m]}, U_{[n]} \right) \right) \right] \\
				&\leq \frac{\lambda^2 \Delta^2 (n+m)^2}{8} \min \left( \frac{1}{n^2} \sum_{i=m}^{n+m-1} \frac{1}{i(i-1)}, \frac{1}{m^2} \sum_{i=n}^{n+m-1} \frac{1}{i(i-1)} \right) \\
				&= \frac{\lambda^2 \Delta^2 (n+m)^2}{8} \min \left( \frac{1}{n^2} \sum_{i=m}^{n+m-1} \left( \frac{1}{i-1} - \frac{1}{i} \right), \frac{1}{m^2} \sum_{i=n}^{n+m-1} \left( \frac{1}{i-1} - \frac{1}{i} \right) \right) \\
				&= \frac{\lambda^2 \Delta^2 (n+m)^2}{8} \min \left( \frac{1}{n^2} \left( \frac{1}{m-1} - \frac{1}{n+m-1} \right), \frac{1}{m^2} \left( \frac{1}{n-1} - \frac{1}{n+m-1} \right) \right) \\
				&= \frac{\lambda^2 \Delta^2 (n+m)^2}{8 (n+m-1) (nm - \min(n, m))}.
				\label{eq:cgf-ub-gen}
			\end{aligned}
		\end{equation}
		Combining \eqref{eq:cgf-ub-1} and \eqref{eq:cgf-ub-gen} completes the proof.
	\end{proof}
	
	By simply replacing $n$ and $m$ with $\frac{n}{k}$ and $\frac{m}{k}$, respectively, one can verify that this lemma is also applicable to each individual block $\hat{z}^{(i)}_{[\frac{n+m}{k}]}$. Thus, it follows for all $w \in \Wcal$ and $\hat{z}^{(i)}_{[\frac{n+m}{k}]} \in \Zcal^{\frac{n+m}{k}}$ that
	\begin{align}
		\EE_{U^{(i)}_{[\frac{n}{k}]}} \left[ \exp \left( \lambda \Ecal^{(i)}\left( w, \hat{z}^{(i)}_{[\frac{n+m}{k}]}, U^{(i)}_{[\frac{n}{k}]} \right) \right) \right] \leq \exp \left( \frac{\lambda^2 \Delta^2 C^k_{n,m} \cdot k(n+m)}{8nm} \right).
	\end{align}
	Since the above constant bound holds for all $w \in \Wcal$ (including the worst evaluated $w$), then this upper bound still holds when we take the expectation over $W$, i.e.,
	\begin{align}
		\EE_{\tilde{W}, \tilde{U}^{(i)}_{[\frac{n}{k}]} | \hat{z}^{(i)}_{[\frac{n+m}{k}]}} \left[ \exp \Ecal^{(i)}\left( w, \hat{z}^{(i)}_{[\frac{n+m}{k}]}, U^{(i)}_{[\frac{n}{k}]} \right) \right] &= \EE_{W | \hat{z}^{(i)}_{[\frac{n+m}{k}]}} \EE_{U^{(i)}_{[\frac{n}{k}]}} \left[ \exp \left( \lambda \Ecal^{(i)}\left( w, \hat{z}^{(i)}_{[\frac{n+m}{k}]}, U^{(i)}_{[\frac{n}{k}]} \right) \right) \right] \\
		&\leq \exp \left( \frac{\lambda^2 \Delta^2 C^k_{n,m} \cdot k(n+m)}{8nm} \right).
	\end{align}
	Taking logarithm on both sides leads to the desired result.
	
	\subsection{Proof of Theorem~\ref{thm:bld-IPCIMI}}
	In \eqref{eq:gen-dis-IPCIMI}, we let $\Psi_{\hat{Z}^{(i)}_{[\frac{n+m}{k}]}}$ be the RHS of \eqref{eq:CGF-bound}. After optimizing $\lambda$ and taking average over $\hat{Z}^{(i)}_{[\frac{n+m}{k}]}$, the desired result is proved.
	
	\subsection{Proof of Lemma~\ref{lem:lim-MI}} \label{appsub:lemma3}
	We first show the following relationship between the CMI and MI quantities.
	\begin{lemma} \label{lem:ineq-CMI-MI}
		For each $i \in [k]$, we have
		\begin{equation}
			I\left( W; U^{(i)}_{[\frac{n}{k}]} | \hat{Z}^{(i)}_{[\frac{n+m}{k}]} \right) = I\left( W; Z^{(i)}_{[\frac{n}{k}]} \right) - I\left( W; \hat{Z}^{(i)}_{[\frac{n+m}{k}]} \right)
			\label{eq:ineq-CMI-MI}
		\end{equation}
	\end{lemma}
	
	\begin{proof}
		According to the partitioned supersample setting in Definition~\ref{def:p-lmo}, we have
		\begin{equation}
			\begin{aligned}
				I\left( W; \hat{Z}^{(i)}_{[\frac{n+m}{k}]}, U^{(i)}_{[\frac{n}{k}]} \right) &= I\left( W; \hat{Z}^{(i)}_{U^{(i)}_{[\frac{n}{k}]}}, \hat{Z}^{(i)}_{\bar{U}^{(i)}_{[\frac{m}{k}]}}, U^{(i)}_{[\frac{n}{k}]} \right) \\
				&= I\left( W; \hat{Z}^{(i)}_{U^{(i)}_{[\frac{n}{k}]}} \right) + I\left( W; \hat{Z}^{(i)}_{\bar{U}^{(i)}_{[\frac{m}{k}]}}, U^{(i)}_{[\frac{n}{k}]} | \hat{Z}^{(i)}_{U^{(i)}_{[\frac{n}{k}]}} \right) \\
				& \overset{(a)}{=} I\left( W; \hat{Z}^{(i)}_{U^{(i)}_{[\frac{n}{k}]}} \right) \\
				& \overset{(b)}{=} I\left( W; Z^{(i)}_{[\frac{n}{k}]} \right),
			\end{aligned}
		\end{equation}
		where $(a)$ is due to the fact that $W$ is independent of other random variables given $\hat{Z}^{(i)}_{U^{(i)}_{[\frac{n}{k}]}}$, and $(b)$ is by Definition~\ref{def:p-lmo} which reveals that $Z^{(i)}_{[\frac{n}{k}]}$ and $\hat{Z}^{(i)}_{U^{(i)}_{[\frac{n}{k}]}}$ lead to the same training set. Next, by the chain rule of mutual information,
		\begin{equation}
			\begin{aligned}
				I\left( W; U^{(i)}_{[\frac{n}{k}]} | \hat{Z}^{(i)}_{[\frac{n+m}{k}]} \right) &= I\left( W; \hat{Z}^{(i)}_{[\frac{n+m}{k}]}, U^{(i)}_{[\frac{n}{k}]} \right) - I\left( W; \hat{Z}^{(i)}_{[\frac{n+m}{k}]} \right) \\
				&= I\left( W; Z^{(i)}_{[\frac{n}{k}]} \right) - I\left( W; \hat{Z}^{(i)}_{[\frac{n+m}{k}]} \right),
			\end{aligned}
		\end{equation}
		which concludes the proof.
	\end{proof}
	
	Thus, it suffices to prove that the gap between the CMI and MI quantities diminishes as $m$, i.e., $\lim\limits_{m \to \infty} I\left( W; \hat{Z}^{(i)}_{[\frac{n+m}{k}]} \right) = 0$. To see this, we next prove $\lim\limits_{m \to \infty} I( W; \hat{Z}_{[n+m]} ) = 0$. Using the tower property of conditional expectation and the fact that $U_{[n]} \indep \hat{Z}_{[n+m]}$, we get
	\begin{equation}
		\begin{aligned}
			P_{W | \hat{Z}_{[n+m]}}(w | \hat{z}_{[n+m]}) &= \EE_{U_{[n]}} \left[ P_{W | \hat{Z}_{[n+m]}, U_{[n]}}(w | \hat{z}_{[n+m]}, U_{[n]}) \right] \\
			&= \frac{1}{\begin{psmallmatrix} \frac{n+m}{k} \\ \frac{n}{k} \end{psmallmatrix}^k} \sum_{u_{[n]}} P_{W | \hat{Z}_{[n+m]}, U_{[n]}}(w | \hat{z}_{[n+m]}, u_{[n]}) \\
			&= \frac{1}{\begin{psmallmatrix} \frac{n+m}{k} \\ \frac{n}{k} \end{psmallmatrix}^k} \sum_{u_{[n]}} P_{W | Z_{[n]}}(w | \hat{z}_{u_{[n]}}).
			\label{eq:tower-con-exp}
		\end{aligned}
	\end{equation}
	Let $g_w: \Zcal^{n+m} \to [0, 1]$ be a function defined as $g_w(\hat{Z}_{[n+m]}) := P_{W = w | \hat{Z}_{[n+m]}}$. Given a supersample set $\hat{z}_{[n+m]} \in \Zcal^{n+m}$, we use $\hat{z}'_{[n+m]}$ to denote $\hat{z}_{[n+m]}$ in which $\hat{z}^{(i)}_j$ is replaced by an independent copy $\hat{z}'^{(i)}_j$. Then, for each summand in \eqref{eq:tower-con-exp}, it is shown that $P_{W | Z_{[n]}}(w | \hat{z}_{u_{[n]}})$ and $P_{W | Z_{[n]}}(w | \hat{z}'_{u_{[n]}})$ only differ when the index $j$ is included by $u^{(i)}_{[\frac{n}{k}]}$. Thus, the absolute difference between $g_w(\hat{z}_{[n+m]})$ and $g_w({\hat{z}'_{[n+m]}})$ can be bounded as:
	\begin{equation}
		\begin{aligned}
			\left| g_w(\hat{z}_{[n+m]}) - g_w(\hat{z}'_{[n+m]}) \right| &= \left| \frac{1}{\begin{psmallmatrix} \frac{n+m}{k} \\ \frac{n}{k} \end{psmallmatrix}^k} \sum_{u^{(1)}_{[\frac{n}{k}]}} \cdots \sum_{u^{(i)}_{[\frac{n}{k}]}:\ j \in u^{(i)}_{[\frac{n}{k}]}} \cdots \sum_{u^{(k)}_{[\frac{n}{k}]}} P_{W | \hat{Z}_{U_{[n]}}}(w | \hat{z}_{u_{[n]}}) - P_{W | \hat{Z}_{U_{[n]}}}(w | \hat{z}'_{u_{[n]}}) \right| \\
			&\overset{(a)}{\leq} \frac{1}{\begin{psmallmatrix} \frac{n+m}{k} \\ \frac{n}{k} \end{psmallmatrix}^k} \sum_{u^{(1)}_{[\frac{n}{k}]}} \cdots \sum_{u^{(i)}_{[\frac{n}{k}]}:\ j \in u^{(i)}_{[\frac{n}{k}]}} \cdots \sum_{u^{(k)}_{[\frac{n}{k}]}} \left| P_{W | \hat{Z}_{U_{[n]}}}(w | \hat{z}_{u_{[n]}}) - P_{W | \hat{Z}_{U_{[n]}}}(w | \hat{z}'_{u_{[n]}}) \right| \\
			&\overset{(b)}{\leq} \frac{1}{\begin{psmallmatrix} \frac{n+m}{k} \\ \frac{n}{k} \end{psmallmatrix}^k} \cdot \begin{psmallmatrix} \frac{n+m}{k} \\ \frac{n}{k} \end{psmallmatrix}^{k-1} \cdot \begin{psmallmatrix} \frac{n+m}{k}-1 \\ \frac{n}{k}-1 \end{psmallmatrix} \\
			&= \frac{n}{n+m},
		\end{aligned}
	\end{equation}
	where $(a)$ is obtained by applying Jensen's inequality on the convex function $|\cdot|$, and $(b)$ is due to the fact that probability measures are within $[0, 1]$. With the above bounded difference property, applying McDiarmid’s inequality yields
	\begin{equation}
		\begin{aligned}
			\PP \left\{ \left| g_w(\hat{Z}_{[n+m]}) - \EE \left[ g_w(\hat{Z}_{[n+m]}) \right] \right| \geq \epsilon \right\} \leq 2 \exp\left( -\frac{2(n+m) \epsilon^2}{n^2} \right). 
		\end{aligned}
	\end{equation}
	Notice that $\EE \left[ g_w(\hat{Z}_{[n+m]}) \right] = P_{W = w}$, the above inequality reveals a convergence in probability as $\lim\limits_{m \to \infty} P_{W = w | \hat{Z}_{[n+m]}} = P_{W = w}$, that is, $W$ tends to be independent of $\hat{Z}_{[n+m]}$ in the limit as $m \to \infty$. Thus, by definition of mutual information, we obtain
	\begin{equation}
		\begin{aligned}
			\lim_{m \to \infty} I\left( W; \hat{Z}_{[n+m]} \right) &= \lim_{m \to \infty} \EE \left[ D\left( P_{W|\hat{Z}_{[n+m]}} \| P_W \right) \right] \\
			&= \lim_{m \to \infty} \EE \left[ \sum_{w \in \Wcal} P_{W = w |\hat{Z}_{[n+m]}} \log \frac{P_{W = w |\hat{Z}_{[n+m]}}}{P_{W = w}} \right] \\
			&= 0,
		\end{aligned}
	\end{equation}
	where the last step holds because of {\cite[Theorem~2.3.4]{durrett2019probability}}, as $\log \frac{P_{W = w |\hat{Z}_{[n+m]}}}{P_{W = w}}$ is bounded with $|\Wcal| < \infty$. Next, by the chain rule of mutual information, we have
	\begin{equation}
		\begin{aligned}
			I\left( W; \hat{Z}_{[n+m]} \right) = I\left( W; \hat{Z}^{(i)}_{[\frac{n+m}{k}]} \right) + I\left( W; \hat{Z}_{[n+m]} \backslash \hat{Z}^{(i)}_{[\frac{n+m}{k}]} | \hat{Z}^{(i)}_{[\frac{n+m}{k}]} \right),
		\end{aligned}
	\end{equation}
	which further indicates the following bounds due to the non-negativity of mutual information:
	\begin{equation}
		\begin{aligned}
			0 \leq I\left( W; \hat{Z}^{(i)}_{[\frac{n+m}{k}]} \right) \leq I\left( W; \hat{Z}_{[n+m]} \right).
		\end{aligned}
	\end{equation}
	Taking the limit of $m \to \infty$ on both sides of the second inequality, the desired result immediately follows by invoking the squeeze theorem.
	
	\subsection{Proof of Corollary~\ref{cor:IPMI}} \label{subsec:proof-IPMI}
	We first simplify the leading coefficient under the square-root of \eqref{eq:IPCIMI}. For the case of $\min(n, m) = k$, we only need to consider $k = n$ as $m \to \infty$. In this case, we have $C^k_{n,m} = \frac{n+m}{m}$ and the coefficient is
	\begin{equation}
		\begin{aligned}
			\frac{C_{n, m}^k \cdot k(n+m)}{2nm} = \frac{1}{2}\left( 1+\frac{n}{m} \right)^2 = \frac{k}{2n} + O\left( \frac{1}{m} \right).
		\end{aligned}
	\end{equation}
	For the other cases of $k \not= n$, the coefficient can be simplified as
	\begin{equation}
		\begin{aligned}
			\frac{C_{n, m}^k \cdot k(n+m)}{2nm} &= \frac{n+m}{n+m-k} \cdot \frac{nm}{nm - k n} \cdot \frac{k(n+m)}{2nm} \\
			&= \left( 1 + \frac{k}{n+m-k} \right) \left( 1 + \frac{k}{m - k} \right) \frac{k(n+m)}{2nm} \\
			&= \frac{k}{2n} \left( 1 + \frac{k}{n+m-k} \right) \left( 1 + \frac{k}{m - k} \right) \left( 1 + \frac{n}{m} \right) \\
			&= \frac{k}{2n} + O\left( \frac{1}{m} \right).
		\end{aligned}
	\end{equation}
	Combining the result in Lemma~\ref{lem:lim-MI}, one can verify that
	\begin{equation}
		\begin{aligned}
			\frac{1}{k} \sum_{i=1}^{k} \lim_{m \to \infty} \sqrt{\frac{C_{n, m}^k \cdot k(n+m)}{2nm} I\left( W; U^{(i)}_{[\frac{n}{k}]} | \hat{Z}^{(i)}_{[\frac{n+m}{k}]} \right)} &= \frac{1}{k} \sum_{i=1}^{k} \sqrt{\lim_{m \to \infty} \frac{C_{n, m}^k \cdot k(n+m)}{2nm} \lim_{m \to \infty} I\left( W; U^{(i)}_{[\frac{n}{k}]} | \hat{Z}^{(i)}_{[\frac{n+m}{k}]} \right)} \\
			&= \frac{1}{k} \sum_{i=1}^{k} \sqrt{\frac{k}{2n} I\left( W; Z^{(i)}_{[\frac{n}{k}]} \right)}.
		\end{aligned}\label{eq:lim-IPCIMI}
	\end{equation}
	Specially, when $k = 1$ and $m \to \infty$, \eqref{eq:lim-IPCIMI} is simplified as
	\begin{equation}
		\begin{aligned}
			\sqrt{\frac{1}{2n} I\left( W; Z_{[n]} \right)},
		\end{aligned}
	\end{equation}
	which matches the MI bound in \eqref{eq:MI}. In addition, for $k = n$ and $m \to \infty$, the following result of \eqref{eq:lim-IPCIMI} matches the IMI bound in \eqref{eq:IMI}:
	\begin{equation}
		\begin{aligned}
			\frac{1}{n} \sum_{i=1}^{n} \sqrt{\frac{1}{2} I\left( W; Z_i \right)}.
		\end{aligned}
	\end{equation}

	\section{Proofs in Section~\ref{sec:SCMI-bounds}}\label{app:SCMI-bounds}
	\setcounter{subsection}{0}
	\subsection{Proof of Lemma~\ref{lem:js-bound}}\label{subapp:js-bound}
	Starting with the definition of JS divergence, we have
	\begin{equation}
		\begin{aligned}
			&d^\theta_{\textnormal{JS}} \left( \EE_{X|T=1} \left[ f(X) \right] \left\| \EE_{X|T=0} [f(X)] \right. \right) \\
			&= P(T=1) d_{\textnormal{KL}}(\EE_{X|T = 1} [f(X)] \| \EE_{X} [f(X)]) + P(T=0) d_{\textnormal{KL}}(\EE_{X|T=0} [f(X)] \| \EE_{X} [f(X)]) \\
			&\overset{\eqref{eq:property_d_gamma}}{=} \EE_{T} \left[ \sup_{\gamma} d_{\gamma}(\EE_{X|T} [f(X)] \| \EE_{X} [f(X)]) \right] \\
			&\leq \EE_{T} \left[ \sup_{\gamma} \EE_{X | T} \left[ d_{\gamma}(f(X) \| \EE_{X} [f(X)]) \right] \right],
		\end{aligned}\label{eq:js-bound-X-T}
	\end{equation}
	where the inequality is due to the joint convexity of $d_{\gamma}(\cdot \| \cdot)$. By the Donsker-Varadhan formula in Lemma~\ref{lem:DV-formula}, it holds that
	\begin{equation}
		\begin{aligned}
			\EE_{X|T=t} \left[ d_{\gamma}(f(X) \| \EE_{X} [f(X)]) \right] &\leq D_{\textnormal{KL}} (P_{X|T=t} \| P_X) + \log \EE_{X} \left[ e^{d_{\gamma}(f(X) \| \EE_{X} [f(X)])} \right] \\
			&\leq D_{\textnormal{KL}} (P_{X|T=t} \| P_X),
		\end{aligned}\label{eq:d_gamma_bound}
	\end{equation}
	where the last inequality is due to Lemma~\ref{lem:binary-relative-entropy-concentration}. Since the above inequality holds for any $\gamma$, taking the supremum over $\gamma$ preserves the inequality, i.e.,
	\begin{equation}
		\begin{aligned}
			\sup_{\gamma} \EE_{X|T=t} \left[ d_{\gamma}(f(X) \| \EE_{X} [f(X)]) \right] \leq D_{\textnormal{KL}} (P_{X|T=t} \| P_X),
		\end{aligned}
	\end{equation}
	plugging which into \eqref{eq:js-bound-X-T}, we thus get
	\begin{equation}
		\begin{aligned}
			d^\theta_{\textnormal{JS}} \left( \EE_{X|T=1} \left[ f(X) \right] \left\| \EE_{X|T=0} [f(X)] \right. \right) \leq \EE_{T} \left[ D_{\textnormal{KL}} (P_{X|T} \| P_X) \right] = I(X; T).
		\end{aligned}
	\end{equation}
	Since $T$ is a deterministic function of $Y$, we have $I(X; T) \leq I(X; Y)$ by the data-processing inequality, which completes the proof.

	Next, we prove that the upper bound $I(X; T)$ will hold with equality if $f(X) = X$ and $X \in \{0,1\}$. Notice that when $X \in \{0,1\}$, we have $\EE_X [X] = P(X=1)$ and $\EE_{X|T=t} [X] = P(X=1|T=t)$. By the definition of the binary KL divergence, we then have
	\begin{equation}
		\begin{aligned}
			D_{\textnormal{KL}}(P_{X|T=t} \| P_X) &= P(X=0|T=t) \log \frac{P(X=0|T=t)}{P(X=0)} + P(X=1|T=t) \log \frac{P(X=1|T=t)}{P(X=1)} \\
			&= d_{\textnormal{KL}}(\EE_{X|T=t} [X] \| \EE_X [X]).
		\end{aligned}
	\end{equation}
	Thus, the equality holds that
	\begin{equation}
		\begin{aligned}
			d^\theta_{\textnormal{JS}} \left( \EE_{X|T=1} [X] \left\| \EE_{X|T=0} [X] \right. \right) &= P(T=1) d_{\textnormal{KL}}(\EE_{X|T = 1} [X] \| \EE_{X} [X]) + P(T=0) d_{\textnormal{KL}}(\EE_{X|T=0} [X] \| \EE_{X} [X]) \\
			&= \EE_{T} \left[ D_{\textnormal{KL}}(P_{X|T} \| P_X) \right] = I(X; T).
		\end{aligned}
	\end{equation}
	This completes the proof.
	
	\subsection{Proof of Theorem~\ref{thm:SIPCIMI}} \label{subapp:SIPCIMI}
	The proof begins with the following lemma.
	\begin{lemma} \label{lem:js-lb}
		Let $p, q \in [0, 1]$, then
		\begin{equation}
			\begin{aligned}
				d^\theta_{\textnormal{JS}}(p \| q) \geq 2 \theta (1 - \theta) (p - q)^2.
			\end{aligned}
		\end{equation}
	\end{lemma}
	
	\begin{proof}
		By Pinsker's inequality in Lemma~\ref{lem:Pinsker-inequality}, it can be obtained that
		\begin{equation}
			\begin{aligned}
				d^\theta_{\textnormal{JS}}(p \| q) &= \theta d_{\textnormal{KL}}(p \| \theta p + (1 - \theta) q) + (1 - \theta) d_{\textnormal{KL}}(q \| \theta p + (1 - \theta) q) \\
				&\geq 2\theta d^2_{\textnormal{TV}}(p \| \theta p + (1 - \theta) q) + 2(1 - \theta) d^2_{\textnormal{TV}}(q \| \theta p + (1 - \theta) q) \\
				&= 2\theta (p - \theta p - (1 - \theta) q)^2 + 2(1 - \theta) (q - \theta p - (1 - \theta) q)^2 \\
				&= 2 \theta (1 - \theta) (p - q)^2.
			\end{aligned}
		\end{equation}
		where $d_{\textnormal{TV}}(p \| q) := D_{\textnormal{TV}}(\Bernoulli(p) \| \Bernoulli(q))$.
	\end{proof}
	
	Let $W$, $\hat{Z}_{[n+m]}$ and $U_{[n]}$ be as given in Definition~\ref{def:p-lmo}. One can verify that for each $i \in [k]$ and $j \in [\frac{n+m}{k}]$, 
	\begin{equation}
		\begin{aligned}
			P(T^{(i)}_j = 1) = \PP\left[ U^{(i)}_{[\frac{n}{k}]} \in \Ucal_j \right] = \frac{\begin{psmallmatrix} \frac{n+m}{k} - 1 \\ \frac{n}{k} - 1 \end{psmallmatrix}}{\begin{psmallmatrix} \frac{n+m}{k} \\ \frac{n}{k} \end{psmallmatrix}} = \frac{n}{n+m}, \quad P(T^{(i)}_j = 0) = \frac{m}{n+m}.
		\end{aligned}\label{eq:P_Tij}
	\end{equation}
	Then, the expected empirical loss $\widehat{L}_n$ can be rewritten as
		\begin{equation}
		\begin{aligned}
			\widehat{L}_n &\overset{\eqref{eq:emp_risk_lmo}}{=} \frac{1}{k} \sum_{i=1}^{k} \EE_{W, \hat{Z}^{(i)}_{[\frac{n+m}{k}]}, U^{(i)}_{[\frac{n}{k}]}} \left[ \frac{k}{n} \sum_{j \in U^{(i)}_{[\frac{n}{k}]}} \ell(W, \hat{Z}^{(i)}_j) \right] \\
			&= \frac{1}{n} \sum_{i=1}^{k} \sum_{u} \sum_{j \in u} P(U^{(i)}_{[\frac{n}{k}]} = u) \EE_{W, \hat{Z}^{(i)}_j | U^{(i)}_{[\frac{n}{k}]} = u} \left[ \ell(W, \hat{Z}^{(i)}_j) \right] \\
			&\overset{(a)}{=} \frac{1}{n} \sum_{i=1}^{k} \sum_{j=1}^{\frac{n+m}{k}} \sum_{u: u \in \Ucal_j} P(U^{(i)}_{[\frac{n}{k}]} = u) \EE_{W, \hat{Z}^{(i)}_j | U^{(i)}_{[\frac{n}{k}]} = u} \left[ \ell(W, \hat{Z}^{(i)}_j) \right] \\
			&= \frac{1}{n} \sum_{i=1}^{k} \sum_{j=1}^{\frac{n+m}{k}} \EE \left[ \ell(W, \hat{Z}^{(i)}_j) \mathds{1}_{U^{(i)}_{[\frac{n}{k}]} \in \Ucal_j} \right] \\
			&= \frac{1}{n} \sum_{i=1}^{k} \sum_{j=1}^{\frac{n+m}{k}} P(T^{(i)}_j = 1) \EE_{W, \hat{Z}^{(i)}_j | T^{(i)}_j = 1} \left[ \ell(W, \hat{Z}^{(i)}_j) \right] \\
			&\overset{\eqref{eq:P_Tij}}{=} \frac{1}{n+m} \sum_{i=1}^{k} \sum_{j=1}^{\frac{n+m}{k}} \EE_{W, \hat{Z}^{(i)}_j | T^{(i)}_j = 1} \left[ \ell(W, \hat{Z}^{(i)}_j) \right],
		\end{aligned}\label{eq:exp-emp-risk-single}
	\end{equation}
	where $(a)$ is obtained by swapping the order of summation. Similarly, the expected population loss $L_\mu$ is rewritten as
	\begin{equation}
		\begin{aligned}
			L_\mu = \frac{1}{n+m} \sum_{i=1}^{k} \sum_{j=1}^{\frac{n+m}{k}} \EE_{W, \hat{Z}^{(i)}_j | T^{(i)}_j = 0} \left[ \ell(W, \hat{Z}^{(i)}_j) \right].
		\end{aligned}\label{eq:exp-pop-risk-single}
	\end{equation}
	Combining the above results and applying Lemma~\ref{lem:js-lb} with $\theta = P(T^{(i)}_j = 1) = \frac{n}{n+m}$, $p = \EE_{W | \hat{Z}^{(i)}_j, T^{(i)}_j = 1} \left[ \ell(W, \hat{Z}^{(i)}_j) \right]$ and $q = \EE_{W | \hat{Z}^{(i)}_j, T^{(i)}_j = 0} \left[ \ell(W, \hat{Z}^{(i)}_j) \right]$, it holds that
	\begin{equation}
		\begin{aligned}
			\gen =& \frac{1}{n+m} \sum_{i=1}^{k} \sum_{j=1}^{\frac{n+m}{k}} \left( \EE_{W, \hat{Z}^{(i)}_j | T^{(i)}_j = 0} \left[ \ell(W, \hat{Z}^{(i)}_j) \right] - \EE_{W, \hat{Z}^{(i)}_j | T^{(i)}_j = 1} \left[ \ell(W, \hat{Z}^{(i)}_j) \right] \right) \\
			=& \frac{1}{n+m} \sum_{i=1}^{k} \sum_{j=1}^{\frac{n+m}{k}} \EE_{\hat{Z}^{(i)}_j} \left[ \EE_{W | \hat{Z}^{(i)}_j, T^{(i)}_j = 0} \left[ \ell(W, \hat{Z}^{(i)}_j) \right] - \EE_{W | \hat{Z}^{(i)}_j, T^{(i)}_j = 1} \left[ \ell(W, \hat{Z}^{(i)}_j) \right] \right] \\
			\leq& \frac{1}{n+m} \sum_{i=1}^{k} \sum_{j=1}^{\frac{n+m}{k}} \EE_{\hat{Z}^{(i)}_j} \left[ \sqrt{\frac{1}{2 \theta (1-\theta)} d^\theta_{\textnormal{JS}} \left( \EE_{W | \hat{Z}^{(i)}_j, T^{(i)}_j = 1} \left[ \ell(W, \hat{Z}^{(i)}_j) \right] \left\| \EE_{W | \hat{Z}^{(i)}_j, T^{(i)}_j = 0} \left[ \ell(W, \hat{Z}^{(i)}_j) \right] \right. \right) } \right].
		\end{aligned}
	\end{equation}
	When conditioned on $\hat{Z}^{(i)}_j$, using Lemma~\ref{lem:js-bound} with $X = W$, $Y = U^{(i)}_{[\frac{n}{k}]}$, $T = T^{(i)}_j$, $f(X) = \ell(W, \hat{Z}^{(i)}_j)$ and $\theta = \frac{n}{n+m}$, we can get
	\begin{equation}
		\begin{aligned}
			\gen \leq& \sum_{i=1}^{k} \sum_{j=1}^{\frac{n+m}{k}} \EE_{\hat{Z}^{(i)}_j} \left[ \sqrt{\frac{1}{2nm} I^{\hat{Z}^{(i)}_j}\left( W; T^{(i)}_j \right)} \right] \\
			\leq& \sum_{i=1}^{k} \sum_{j=1}^{\frac{n+m}{k}} \sqrt{\frac{1}{2nm} I\left( W; T^{(i)}_j | \hat{Z}^{(i)}_j \right)},
		\end{aligned}
	\end{equation}
	which completes the proof of \eqref{eq:p-SIPCIMI} and \eqref{eq:p-dis-SIPCIMI}.
	
	\subsection{A Lemma Based on Assumption~\ref{ass:invariant-alg}} \label{subapp:invariant_alg}
	\begin{lemma} \label{lem:inv_alg}
		Suppose that $\Acal$ satisfies Assumption~\ref{ass:invariant-alg}, then for all $i, j \in [n]$, $P_{W, Z_i} = P_{W, Z_j}$, where both $P_{W, Z_i}$ and $P_{W, Z_j}$ are marginals of the joint distribution $P_{W, Z_{[n]}}$.
	\end{lemma}
	\begin{proof}
		Let Assumption~\ref{ass:invariant-alg} hold. Without loss of generality, we assume $1\leq i<j\leq n$ and consider a fixed permutation $\tau = \{1,\ldots,i-1,j,i+1,\ldots,j-1,i,j+1,\ldots,n\}$, such that $Z_{\tau}$ denotes $Z_{[n]}$ in which $Z_i$ and $Z_j$ are swapped. Due to the fact that $Z_i$'s are i.i.d., we have
		\begin{equation}
			\begin{aligned}
				P_{W, Z_{[n]}} (w, z_{[n]}) &= P_{W| Z_{[n]}} (w | z_{[n]}) P_{Z_{[n]}}(z_{[n]}) \\
				&= P_{W| Z_{[n]}} (w | z_{[n]}) P_{Z_1}(z_1) \ldots P_{Z_i}(z_i) \ldots P_{Z_j}(z_j) \ldots P_{Z_n}(z_n) \\
				&\overset{(a)}{=} P_{W| Z_{[n]}} (w | z_{\tau}) P_{Z_1}(z_1) \ldots P_{Z_j}(z_j) \ldots P_{Z_i}(z_i) \ldots P_{Z_n}(z_n) \\
				&= P_{W| Z_{[n]}} (w | z_{\tau}) P_{Z_{[n]}}(z_{\tau}) \\
				&= P_{W, Z_{[n]}} (w, z_{\tau}),
			\end{aligned}
		\end{equation}
		where $(a)$ is obtained by using Assumption~\ref{ass:invariant-alg}. By marginalizing out all $Z_{[n]}$ except for $Z_i$ and $Z_j$ on both sides, we then get
		\begin{equation}
			\begin{aligned}
				P_{W, Z_i, Z_j} (w, z_i, z_j) &= \int_{\Zcal^{n-2}} P_{W, Z_{[n]}} (w, z_{[n]}) d z_1 \ldots d z_{i-1} d z_{i+1} \ldots d z_{j-1} d z_{j+1} \ldots d z_n \\
				&= \int_{\Zcal^{n-2}} P_{W, Z_{[n]}} (w, z_{\tau}) d z_1 \ldots d z_{i-1} d z_{i+1} \ldots d z_{j-1} d z_{j+1} \ldots d z_n \\
				&= P_{W, Z_i, Z_j} (w, z_j, z_i).
			\end{aligned}
		\end{equation}
		Since this equation holds for any $z_j \in \Zcal$, then integrating over $z_j$ preserves the equivalence, i.e.,
		\begin{equation}
			\begin{aligned}
				P_{W, Z_i} (w, z_i) &= \int_{\Zcal} P_{W, Z_i, Z_j} (w, z_i, z_j) d z_j \\
				&= \int_{\Zcal} P_{W, Z_i, Z_j} (w, z_j, z_i) d z_j \\
				&= P_{W, Z_j} (w, z_i).
			\end{aligned} \label{eq:invariant_i_j}
		\end{equation}
	\end{proof}
	Intuitively, the above lemma shows that each input training sample contributes equally to the output hypothesis, i.e., $P_{W|Z_i} = P_{W|Z_j}$ for all $i,j \in [n]$. Thus, when considering the partitioned L$m$O supersample set (Definition~\ref{def:p-lmo}) with any given fixed $i_1,i_2 \in [k]$ and $j_1,j_2 \in [\frac{n+m}{k}]$, Lemma~\ref{lem:inv_alg} results in
	\begin{align}
		\forall z \in \Zcal, \forall u \in \Ucal_{j_1}, \forall u' \in \Ucal_{j_2}, \quad P_{W | Z^{(i_1)}_{j_1}, U^{(i_1)}_{[\frac{n}{k}]}}(\cdot | z, u) = P_{W | Z^{(i_2)}_{j_2}, U^{(i_2)}_{[\frac{n}{k}]}}(\cdot | z, u'), \label{eq:invariant_u_1}
	\end{align}
	which further yields 
	\begin{align}
		\forall z \in \Zcal, \forall u \in \Ucal_{j_1}, \forall u' \in \Ucal_{j_2}, \quad P_{W, Z^{(i_1)}_{j_1} | U^{(i_1)}_{[\frac{n}{k}]}}(\cdot, z | u) = P_{W, Z^{(i_2)}_{j_2} | U^{(i_2)}_{[\frac{n}{k}]}}(\cdot, z | u'), \label{eq:invariant_u_2}
	\end{align}
	since $\hat{Z}^{(i)}_j$'s are i.i.d. distributed and they are independent of $U^{(i)}_{[\frac{n}{k}]}$.
	%Collectively, they lead to the following property. For all $j_1,j_2 \in [\frac{n+m}{k}]$, we have
	%\begin{equation}
	%	P_{W, Z^{(i)}_{j_1}, U^{(i)}_{[\frac{n}{k}]}} = P_{W, Z^{(i)}_{j_2}, U^{(i)}_{[\frac{n}{k}]}}.
	%\end{equation}
	%Additionally, we also have
	%\begin{enumerate}
	%	\item For all $j_1, j_2 \in [\frac{n+m}{k}]$,
	%	\begin{align}
		%		P_{W, \hat{Z}^{(i)}_{j_1}, U^{(i)}_{[\frac{n}{k}]}} = P_{W, \hat{Z}^{(i)}_{j_2}, U^{(i)}_{[\frac{n}{k}]}}.
		%	\end{align}
	%	\item For all $i_1, i_2 \in [k]$,
	%	\begin{align}
		%		P_{W, \hat{Z}^{(i_1)}_{[\frac{n+m}{k}]}, U^{(i_1)}_{[\frac{n}{k}]}} = P_{W, \hat{Z}^{(i_2)}_{[\frac{n+m}{k}]}, U^{(i_2)}_{[\frac{n}{k}]}}.
		%	\end{align}
	%\end{enumerate}
	
	\subsection{Proof of Proposition~\ref{pro:processed-SIPCIMI}} \label{subapp:processed-SIPCIMI}
	Consider any fixed $i \in [k]$ and $j \in [\frac{n+m}{k}]$ and let us condition on $\hat{Z}^{(i)}_j$. Due to the Markov chain
	\begin{equation}\label{eq:W-U-T}
		W - U^{(i)}_{[\frac{n}{k}]} - T^{(i)}_j,
	\end{equation}
	we can apply the data-processing inequality to get 
	\begin{align}
		\forall z \in \Zcal, \quad I^{\hat{Z}^{(i)}_j = z}\left( W; T^{(i)}_j \right) &\leq I^{\hat{Z}^{(i)}_j = z}\left( W; U^{(i)}_{[\frac{n}{k}]} \right), \\
		I\left( W; T^{(i)}_j | \hat{Z}^{(i)}_j \right) &\leq I\left( W; U^{(i)}_{[\frac{n}{k}]} | \hat{Z}^{(i)}_j \right).
	\end{align}
	Next, we prove the above inequalities hold with equality under Assumption~\ref{ass:invariant-alg}. In light of the data-processing inequality, it suffices to show that also the reverse Markov chain holds:
	\begin{equation}
		W - T^{(i)}_j - U^{(i)}_{[\frac{n}{k}]}.
	\end{equation}
	More formally, we are going to argue
	\begin{equation}\label{eq:decp_WU-T}
		P_{W, U^{(i)}_{[\frac{n}{k}]} | \hat{Z}^{(i)}_j, T^{(i)}_j} = P_{W | \hat{Z}^{(i)}_j, T^{(i)}_j} P_{U^{(i)}_{[\frac{n}{k}]} | \hat{Z}^{(i)}_j, T^{(i)}_j}.
	\end{equation} 
	
	Let Assumption~\ref{ass:invariant-alg} hold, then
	\begin{equation}
		\begin{aligned}
			\forall z \in \Zcal, \forall u \in \Ucal_j, \quad P_{W | \hat{Z}^{(i)}_j, T^{(i)}_j}(\cdot | z, 1) &= \frac{P_{W, T^{(i)}_j | \hat{Z}^{(i)}_j}(\cdot, 1 | z)}{P(T^{(i)}_j = 1)} \\
			&\overset{(a)}{=} \frac{m+n}{n} \sum_{u': u' \in \Ucal_j} P_{W, U^{(i)}_{[\frac{n}{k}]} | \hat{Z}^{(i)}_j}(\cdot, u' | z) \\
			&\overset{(b)}{=} \frac{m+n}{n} \begin{psmallmatrix} \frac{n+m}{k}-1 \\ \frac{n}{k}-1 \end{psmallmatrix} P_{W, U^{(i)}_{[\frac{n}{k}]} | \hat{Z}^{(i)}_j}(\cdot, u | z) \\
			&= \begin{psmallmatrix} \frac{n+m}{k} \\ \frac{n}{k} \end{psmallmatrix} P_{W, U^{(i)}_{[\frac{n}{k}]} | \hat{Z}^{(i)}_j}(\cdot, u | z) \\
			&= P_{W | \hat{Z}^{(i)}_j, U^{(i)}_{[\frac{n}{k}]}}(\cdot | z, u),
		\end{aligned}\label{eq:eqv_P_T1_U}
	\end{equation}
	where $(a)$ is due to the definition of $T^{(i)}_j$ and \eqref{eq:P_Tij}, and $(b)$ is obtained by using \eqref{eq:invariant_u_1} (which holds since Assumption~\ref{ass:invariant-alg} is satisfied) with $|\Ucal_j| = \begin{psmallmatrix} \frac{n+m}{k}-1 \\ \frac{n}{k}-1 \end{psmallmatrix}$. Then, it follows that
	\begin{equation}
		\begin{aligned}
			P_{W | \hat{Z}^{(i)}_j, T^{(i)}_j}(\cdot | z, 1) P_{U^{(i)}_{[\frac{n}{k}]} | \hat{Z}^{(i)}_j, T^{(i)}_j} (u | z, 1) &\overset{\eqref{eq:eqv_P_T1_U}}{=} P_{W | \hat{Z}^{(i)}_j, U^{(i)}_{[\frac{n}{k}]}} (\cdot | z, u) P_{U^{(i)}_{[\frac{n}{k}]} | \hat{Z}^{(i)}_j, T^{(i)}_j} (u | z, 1) \\
			&\overset{(a)}{=} P_{W | \hat{Z}^{(i)}_j, U^{(i)}_{[\frac{n}{k}]}, T^{(i)}_j} (\cdot | z, u, 1) P_{U^{(i)}_{[\frac{n}{k}]} | \hat{Z}^{(i)}_j, T^{(i)}_j} (u | z, 1) \\
			&= P_{W, U^{(i)}_{[\frac{n}{k}]} | \hat{Z}^{(i)}_j, T^{(i)}_j} (\cdot, u | z, 1),
		\end{aligned}\label{eq:decp_T1}
	\end{equation}
	where $(a)$ holds due to the Markov chain \eqref{eq:W-U-T}. Besides, we obtain
	\begin{equation}
		\begin{aligned}
			\forall z \in \Zcal, \forall \bar{u} \not\in \Ucal_j, \quad P_{W | \hat{Z}^{(i)}_j, T^{(i)}_j}(\cdot | z, 0) &= \frac{P_{W, T^{(i)}_j | \hat{Z}^{(i)}_j}(\cdot, 0 | z)}{P(T^{(i)}_j = 0)} \\
			&\overset{(a)}{=} \frac{m+n}{m} \sum_{\bar{u}': \bar{u}' \not\in \Ucal_j} P_{W, U^{(i)}_{[\frac{n}{k}]} | \hat{Z}^{(i)}_j}(\cdot, \bar{u}' | z) \\
			&\overset{(b)}{=} \frac{m+n}{m} \begin{psmallmatrix} \frac{n+m}{k}-1 \\ \frac{m}{k}-1 \end{psmallmatrix} P_{W, U^{(i)}_{[\frac{n}{k}]} | \hat{Z}^{(i)}_j}(\cdot, u | z) \\
			&= \begin{psmallmatrix} \frac{n+m}{k} \\ \frac{n}{k} \end{psmallmatrix} P_{W, U^{(i)}_{[\frac{n}{k}]} | \hat{Z}^{(i)}_j}(\cdot, \bar{u} | z) \\
			&= P_{W | \hat{Z}^{(i)}_j, U^{(i)}_{[\frac{n}{k}]}}(\cdot | z, \bar{u}),
		\end{aligned}\label{eq:eqv_P_T0_U}
	\end{equation}
	where $(a)$ is by \eqref{eq:P_Tij}, and $(b)$ is due to the fact that $W \indep \hat{Z}^{(i)}_j$ if $\hat{Z}^{(i)}_j$ is not a training sample, i.e., $P_{W | \hat{Z}^{(i)}_j, U^{(i)}_{[\frac{n}{k}]}}(\cdot | z, \bar{u}_1) = P_{W}(\cdot) = P_{W | \hat{Z}^{(i)}_j, U^{(i)}_{[\frac{n}{k}]}}(\cdot | z, \bar{u}_2)$ for all $\bar{u}_1, \bar{u}_2 \not\in \Ucal_j$ (which does not rely on Assumption~\ref{ass:invariant-alg}). Thus, we get the following analogous to \eqref{eq:decp_T1}
	\begin{equation}
		\begin{aligned}
			\forall z \in \Zcal, \forall \bar{u} \not\in \Ucal_j, \quad P_{W | \hat{Z}^{(i)}_j, T^{(i)}_j}(\cdot | z, 0) P_{U^{(i)}_{[\frac{n}{k}]} | \hat{Z}^{(i)}_j, T^{(i)}_j} (\bar{u} | z, 0)
			&= P_{W, U^{(i)}_{[\frac{n}{k}]} | \hat{Z}^{(i)}_j, T^{(i)}_j} (\cdot, \bar{u} | z, 0).
		\end{aligned}\label{eq:decp_T0}
	\end{equation}
	Next, it holds by the definition of $T^{(i)}_j$ that 
	\begin{equation}
		\begin{aligned}
			\forall z \in \Zcal, \forall u \in \Ucal_j, \bar{u} \in \Ucal_j, \quad P_{U^{(i)}_{[\frac{n}{k}]} | \hat{Z}^{(i)}_j, T^{(i)}_j} (u | z, 0) = 0 = P_{U^{(i)}_{[\frac{n}{k}]} | \hat{Z}^{(i)}_j, T^{(i)}_j} (\bar{u} | z, 1).
		\end{aligned}
	\end{equation}
	For the same reason,
	\begin{equation}
		\begin{aligned}
			\forall z \in \Zcal, \forall u \in \Ucal_j, \bar{u} \in \Ucal_j, \quad P_{W, U^{(i)}_{[\frac{n}{k}]} | \hat{Z}^{(i)}_j, T^{(i)}_j} (\cdot, u | z, 0) = 0 = P_{W, U^{(i)}_{[\frac{n}{k}]} | \hat{Z}^{(i)}_j, T^{(i)}_j} (\cdot, \bar{u} | z, 1).
		\end{aligned}
	\end{equation}
	Combining all the facts together, we find that \eqref{eq:decp_WU-T} holds for all realizations of $U^{(i)}_{[\frac{n}{k}]}$ and $T^{(i)}_j$. This completes the proof.
	
	\subsection{Proof of Proposition~\ref{pro:invariant-SCMI}} \label{subapp:invariant-SCMI}
	Since Proposition~\ref{pro:processed-SIPCIMI} states that the SCMI quantity (RHS of \eqref{eq:processed-SIPCIMI}) is equivalent to its processed counterpart (LHS of \eqref{eq:processed-SIPCIMI}) when Assumption~\ref{ass:invariant-alg} is satisfied, it suffices to prove $I^{\hat{Z}^{(i_1)}_{j_1} = z}\left( W; T^{(i_1)}_{j_1} \right) = I^{\hat{Z}^{(i_2)}_{j_2} = z}\left( W; T^{(i_2)}_{j_2} \right)$. 
	
	First, it holds from \eqref{eq:eqv_P_T0_U} and \eqref{eq:eqv_P_T1_U} that
	\begin{align}
		&\forall z \in \Zcal, \forall \bar{u} \not\in \Ucal_j, \quad P_{W | \hat{Z}^{(i)}_j, T^{(i)}_j}(\cdot | z, 0) = P_{W | \hat{Z}^{(i)}_j, U^{(i)}_{[\frac{n}{k}]}}(\cdot | z, \bar{u}) = P_{W}(\cdot), \label{eq:P_W_Z_T0} \\
		&\forall z \in \Zcal, \forall u \in \Ucal_j, \forall j' \in [n], \quad P_{W | \hat{Z}^{(i)}_j, T^{(i)}_j}(\cdot | z, 1) = P_{W | \hat{Z}^{(i)}_j, U^{(i)}_{[\frac{n}{k}]}}(\cdot | z, u) = P_{W | Z_{j'}}(\cdot | z),  \label{eq:P_W_Z_T1}
	\end{align}
	where \eqref{eq:P_W_Z_T1} holds with Assumption~\ref{ass:invariant-alg}. \eqref{eq:P_W_Z_T0} and \eqref{eq:P_W_Z_T1} jointly imply that $P_{W | \hat{Z}^{(i)}_j, T^{(i)}_j}$ does not depend on $i,j,k$. Next, by definition of mutual information,
	\begin{equation}
		\begin{aligned}
			I^{\hat{Z}^{(i)}_j = z}\left( W; T^{(i)}_j \right) =& \sum_{t \in \{0,1\}} \int_{\Wcal} P_{T^{(i)}_j}(t) \log \left( \frac{d P_{W | \hat{Z}^{(i)}_j, T^{(i)}_j}(\cdot | z, t)}{d P_{W | \hat{Z}^{(i)}_j}(\cdot | z)}(w) \right) d P_{W | \hat{Z}^{(i)}_j, T^{(i)}_j}(d w | z, t)\\
			=& -\sum_{t \in \{0,1\}} \int_{\Wcal} P_{T^{(i)}_j}(t) \log \left( \frac{\sum_{t' \in \{0,1\}}P_{T^{(i)}_j}(t') d P_{W | \hat{Z}^{(i)}_j, T^{(i)}_j}(\cdot | z, t')}{d P_{W | \hat{Z}^{(i)}_j, T^{(i)}_j}(\cdot | z, t)}(w) \right) d P_{W | \hat{Z}^{(i)}_j, T^{(i)}_j}(d w | z, t) \\
			=& - \int_{\Wcal} P_{T^{(i)}_j}(0) \log \left( P_{T^{(i)}_j}(0) + P_{T^{(i)}_j}(1) \frac{d P_{W | \hat{Z}^{(i)}_j, T^{(i)}_j}(\cdot | z, 1)}{d P_{W | \hat{Z}^{(i)}_j, T^{(i)}_j}(\cdot | z, 0)}(w) \right) d P_{W | \hat{Z}^{(i)}_j, T^{(i)}_j}(d w | z, 0) \\
			&- \int_{\Wcal} P_{T^{(i)}_j}(1) \log \left( P_{T^{(i)}_j}(1) + P_{T^{(i)}_j}(0) \frac{d P_{W | \hat{Z}^{(i)}_j, T^{(i)}_j}(\cdot | z, 0)}{d P_{W | \hat{Z}^{(i)}_j, T^{(i)}_j}(\cdot | z, 1)}(w) \right) d P_{W | \hat{Z}^{(i)}_j, T^{(i)}_j}(d w | z, 1).
		\end{aligned}\label{eq:I^Z_W_T}
	\end{equation}
	When Assumption~\ref{ass:invariant-alg} is satisfied, \eqref{eq:P_Tij}, \eqref{eq:P_W_Z_T0} and \eqref{eq:P_W_Z_T1} show that both $P_{T^{i}_j}$ and $P_{W | \hat{Z}^{(i)}_j, T^{(i)}_j}$ do not rely on $i,j,k$. Hence, it follows that \eqref{eq:I^Z_W_T} is independent of $i,j,k$. This concludes the proof.
	
	\section{Proofs in Section~\ref{sec:extensions}} \label{app:proof-extensions}
	\setcounter{subsection}{0}	
	\subsection{Proof of Theorem~\ref{thm:lb-gen-bound}} \label{subapp:lb-gen-bound}
	Through the following lemma, we show that given fixed $w$ and $\hat{z}^{(i)}_{[\frac{n+m}{k}]}$, the loss averaged over the training or test subset, when taken expectation over $P_{U^{(i)}_{[\frac{n}{k}]}}$, is the loss averaged over the whole supersample block $\hat{z}^{(i)}_{[\frac{n+m}{k}]}$.
	\begin{lemma}\label{lem:average-emp-loss}
		Consider the partitioned L$m$O supersample setting in Definition~\ref{def:p-lmo}. For all $i \in [k]$, $w \in \Wcal$ and $\hat{z}^{(i)}_{[\frac{n+m}{k}]} \in \Zcal^{\frac{n+m}{k}}$, we have
		\begin{equation}
			\begin{aligned}
				\EE_{U^{(i)}_{[\frac{n}{k}]}} \left[ \frac{k}{n} \sum_{j \in U^{(i)}_{[\frac{n}{k}]}} \Lambda^{(i)}_j \right] = \EE_{U^{(i)}_{[\frac{n}{k}]}} \left[ \frac{k}{m} \sum_{j \not\in U^{(i)}_{[\frac{n}{k}]}} \Lambda^{(i)}_j \right] = \frac{k}{n+m} \sum_{j=1}^{\frac{n+m}{k}} \Lambda^{(i)}_j.
			\end{aligned}
		\end{equation}
	\end{lemma}
	\begin{proof}
		By definition, we have
		\begin{equation}
			\begin{aligned}
				\EE_{U^{(i)}_{[\frac{n}{k}]}} \left[ \frac{k}{n} \sum_{j \in U^{(i)}_{[\frac{n}{k}]}} \Lambda^{(i)}_j \right] &= \frac{k}{n} \sum_{u} P_{U^{(i)}_{[\frac{n}{k}]}}(u) \sum_{j \in u} \Lambda^{(i)}_j \\
				&= \frac{k}{n} \sum_{j=1}^{\frac{n+m}{k}} \Lambda^{(i)}_j \sum_{u : u \in \Ucal_j} P_{U^{(i)}_{[\frac{n}{k}]}}(u) \\
				&= \frac{k}{n}  \sum_{j=1}^{\frac{n+m}{k}} \Lambda^{(i)}_j P(T^{(i)}_j = 1) \\
				&\overset{\eqref{eq:P_Tij}}{=} \frac{k}{n+m} \sum_{j=1}^{\frac{n+m}{k}} \Lambda^{(i)}_j.
			\end{aligned}
		\end{equation}
		Likewise, it holds that
		\begin{equation}
			\begin{aligned}
				\EE_{U^{(i)}_{[\frac{n}{k}]}} \left[ \frac{k}{m} \sum_{j \not\in U^{(i)}_{[\frac{n}{k}]}} \Lambda^{(i)}_j \right] &= \frac{k}{m} \sum_{u} P_{U^{(i)}_{[\frac{n}{k}]}}(u) \sum_{j \not\in u} \Lambda^{(i)}_j \\
				&= \frac{k}{m} \sum_{j=1}^{\frac{n+m}{k}} \Lambda^{(i)}_j \sum_{u : u \not\in \Ucal_j} P_{U^{(i)}_{[\frac{n}{k}]}}(u) \\
				&= \frac{k}{m}  \sum_{j=1}^{\frac{n+m}{k}} \Lambda^{(i)}_j P(T^{(i)}_j = 0) \\
				&\overset{\eqref{eq:P_Tij}}{=} \frac{k}{n+m} \sum_{j=1}^{\frac{n+m}{k}} \Lambda^{(i)}_j.
			\end{aligned}
		\end{equation}
		The proof is complete.
	\end{proof}
	
	Now, we are ready to prove \eqref{eq:lb-gen-bound} and \eqref{eq:lb-gen-bound-single}, respectively.
	\subsubsection{Proof of \eqref{eq:lb-gen-bound}}
	First, by the property of $d_{\textnormal{KL}}$, we get
	\begin{equation}
		\begin{aligned}
			&d_{\textnormal{KL}} \left(\widehat{L}_n \left\| \frac{n}{n+m} \widehat{L}_n + \frac{m}{n+m} L_\mu \right.\right) \\
			&\overset{(a)}{=} \sup_{\gamma} d_{\gamma} \left(\widehat{L}_n \left\| \frac{n}{n+m} \widehat{L}_n + \frac{m}{n+m} L_\mu \right.\right) \\
			&\overset{\eqref{eq:lb-IP-exp-emp-risk}, \eqref{eq:lb-IP-exp-pop-risk}}{=} \sup_{\gamma} d_{\gamma} \left(\frac{1}{k} \sum_{i=1}^k \EE_{\Lambda^{(i)}_{[\frac{n+m}{k}]}, U^{(i)}_{[\frac{n}{k}]}} \left[ \frac{k}{n} \sum_{j \in U^{(i)}_{[\frac{n}{k}]}} \Lambda^{(i)}_j \right] \left\| \frac{1}{k} \sum_{i=1}^k \EE_{\Lambda^{(i)}_{[\frac{n+m}{k}]}, U^{(i)}_{[\frac{n}{k}]}} \left[ \frac{k}{n+m} \sum_{j=1}^{\frac{n+m}{k}} \Lambda^{(i)}_j \right] \right.\right) \\
			&\overset{(b)}{\leq} \sup_{\gamma} \frac{1}{k} \sum_{i=1}^k \EE_{\Lambda^{(i)}_{[\frac{n+m}{k}]}, U^{(i)}_{[\frac{n}{k}]}} \left[ d_{\gamma} \left( \frac{k}{n} \sum_{j \in U^{(i)}_{[\frac{n}{k}]}} \Lambda^{(i)}_j \left\| \frac{k}{n+m} \sum_{j=1}^{\frac{n+m}{k}} \Lambda^{(i)}_j \right. \right) \right],
		\end{aligned}
	\end{equation}
	where $(a)$ is by the property of $d_{\textnormal{KL}}(\cdot \| \cdot)$ in \eqref{eq:property_d_gamma} and $(b)$ is by the joint convexity of $d_{\textnormal{KL}}(\cdot \| \cdot)$. For any $\gamma \in \RR$ and $i \in [k]$, it follows by the Donsker-Varadhan formula that
	\begin{equation}
		\begin{aligned}
			&\EE_{\Lambda^{(i)}_{[\frac{n+m}{k}]}, U^{(i)}_{[\frac{n}{k}]}} \left[ \frac{n}{k} d_{\gamma} \left( \frac{k}{n} \sum_{j \in U^{(i)}_{[\frac{n}{k}]}} \Lambda^{(i)}_j \left\| \frac{k}{n+m} \sum_{j=1}^{\frac{n+m}{k}} \Lambda^{(i)}_j \right. \right) \right] \\
			&\leq D_{\textnormal{KL}}\left( P_{\Lambda^{(i)}_{[\frac{n+m}{k}]},  U^{(i)}_{[\frac{n}{k}]}} \left\| P_{\Lambda^{(i)}_{[\frac{n+m}{k}]}} P_{U^{(i)}_{[\frac{n}{k}]}} \right.\right) + \log \EE_{\Lambda^{(i)}_{[\frac{n+m}{k}]}} \EE_{U^{(i)}_{[\frac{n}{k}]}} \left[ e^{\frac{n}{k} d_{\gamma} \left( \frac{k}{n} \sum_{j \in U^{(i)}_{[\frac{n}{k}]}} \Lambda^{(i)}_j \left\| \frac{k}{n+m} \sum_{j=1}^{\frac{n+m}{k}} \Lambda^{(i)}_j \right. \right)} \right] \\
			&= I\left( \Lambda^{(i)}_{[\frac{n+m}{k}]} ;  U^{(i)}_{[\frac{n}{k}]} \right) + \log \EE_{\Lambda^{(i)}_{[\frac{n+m}{k}]}} \EE_{U^{(i)}_{[\frac{n}{k}]}} \left[ e^{\frac{n}{k} d_{\gamma} \left( \frac{k}{n} \sum_{j \in U^{(i)}_{[\frac{n}{k}]}} \Lambda^{(i)}_j \left\| \frac{k}{n+m} \sum_{j=1}^{\frac{n+m}{k}} \Lambda^{(i)}_j \right. \right)} \right].
		\end{aligned}\label{eq:gen-d_gamma-bound}
	\end{equation}
	In the above expression, the second argument of $d_\gamma$ is the mean of the first argument under $P_{U^{(i)}_{[\frac{n}{k}]}}$, as shown in Lemma~\ref{lem:average-emp-loss}. Therefore, an application of Lemma~\ref{lem:binary-relative-entropy-concentration} directly yields
	\begin{equation}
		\begin{aligned}
			\EE_{\Lambda^{(i)}_{[\frac{n+m}{k}]}, U^{(i)}_{[\frac{n}{k}]}} \left[ d_{\gamma} \left( \frac{k}{n} \sum_{j \in U^{(i)}_{[\frac{n}{k}]}} \Lambda^{(i)}_j \left\| \frac{k}{n+m} \sum_{j=1}^{\frac{n+m}{k}} \Lambda^{(i)}_j \right. \right) \right] \leq \frac{k}{n} I\left( \Lambda^{(i)}_{[\frac{n+m}{k}]} ;  U^{(i)}_{[\frac{n}{k}]} \right).
		\end{aligned}
	\end{equation}
	Since this inequality holds for any $\gamma$, it also holds for the supremum over $\gamma$. By collecting everything together, we get
	\begin{equation}
		\begin{aligned}
			d_{\textnormal{KL}} \left(\widehat{L}_n \left\| \frac{n}{n+m} \widehat{L}_n + \frac{m}{n+m} L_\mu \right.\right) \leq \frac{1}{n} \sum_{i=1}^k I\left( \Lambda^{(i)}_{[\frac{n+m}{k}]} ;  U^{(i)}_{[\frac{n}{k}]} \right).
		\end{aligned}\label{eq:d_KL_emp}
	\end{equation}
	Likewise, repeating the above steps for $d_{\textnormal{KL}} \left(L_\mu \left\| \frac{n}{n+m} \widehat{L}_n + \frac{m}{n+m} L_\mu \right.\right)$ yields
	\begin{equation}
		\begin{aligned}
			d_{\textnormal{KL}} \left(L_\mu \left\| \frac{n}{n+m} \widehat{L}_n + \frac{m}{n+m} L_\mu \right.\right) \leq \frac{1}{m} \sum_{i=1}^k I\left( \Lambda^{(i)}_{[\frac{n+m}{k}]} ;  U^{(i)}_{[\frac{n}{k}]} \right).
		\end{aligned}\label{eq:d_KL_pop}
	\end{equation}
	Finally, combining the above two inequalities gives an upper bound for the weighted binary JSD between $\widehat{L}_n$ and $L_\mu$ (with $\theta = \frac{n}{n+m}$), i.e., 
	\begin{equation}
		\begin{aligned}
			d^\theta_{\textnormal{JS}} \left(\widehat{L}_n \left\| L_\mu \right.\right) &= \frac{n}{n+m} d_{\textnormal{KL}} \left(\widehat{L}_n \left\| \frac{n}{n+m} \widehat{L}_n + \frac{m}{n+m} L_\mu \right.\right) + \frac{m}{n+m} d_{\textnormal{KL}} \left(L_\mu \left\| \frac{n}{n+m} \widehat{L}_n + \frac{m}{n+m} L_\mu \right.\right) \\
			&\leq \frac{2}{n+m} \sum_{i=1}^k I\left( \Lambda^{(i)}_{[\frac{n+m}{k}]} ;  U^{(i)}_{[\frac{n}{k}]} \right).
		\end{aligned}\label{eq:d_JS_IPCIMI}
	\end{equation}
	
	\subsubsection{Proof of \eqref{eq:lb-gen-bound-single}}
	First, following the same way as it did for \eqref{eq:exp-emp-risk-single} and \eqref{eq:exp-pop-risk-single}, \eqref{eq:lb-IP-exp-emp-risk} and \eqref{eq:lb-IP-exp-pop-risk} are further simplified as
	\begin{align}
		\widehat{L}_n =& \frac{1}{n+m} \sum_{i=1}^k \sum_{j=1}^{\frac{n+m}{k}} \EE_{\Lambda^{(i)}_j | T^{(i)}_j = 1} \left[ \Lambda^{(i)}_j \right] \\
		L_\mu =& \frac{1}{n+m} \sum_{i=1}^k \sum_{j=1}^{\frac{n+m}{k}} \EE_{\Lambda^{(i)}_j | T^{(i)}_j = 0} \left[ \Lambda^{(i)}_j \right].
	\end{align}
	Applying Jensen's inequality on the convex function $d_{\textnormal{KL}}(\cdot \| \cdot)$, it follows that
	\begin{equation}
		\begin{aligned}
			d_{\textnormal{KL}} \left(\widehat{L}_n \left\| \frac{n}{n+m} \widehat{L}_n + \frac{m}{n+m} L_\mu \right.\right) =& d_{\textnormal{KL}} \left(\frac{1}{n+m} \sum_{i=1}^k \sum_{j=1}^{\frac{n+m}{k}} \EE_{\Lambda^{(i)}_j | T^{(i)}_j = 1} \left[ \Lambda^{(i)}_j \right] \left\| \frac{1}{n+m} \sum_{i=1}^k \sum_{j=1}^{\frac{n+m}{k}} \EE_{\Lambda^{(i)}_j} \left[ \Lambda^{(i)}_j \right] \right.\right) \\
			\leq& \frac{1}{n+m} \sum_{i=1}^k \sum_{j=1}^{\frac{n+m}{k}} d_{\textnormal{KL}} \left(\EE_{\Lambda^{(i)}_j | T^{(i)}_j = 1} \left[ \Lambda^{(i)}_j \right] \left\| \EE_{\Lambda^{(i)}_j} \left[ \Lambda^{(i)}_j \right] \right.\right).
		\end{aligned}
	\end{equation}
	On the other hand, we can analogously obtain that
	\begin{equation}
		\begin{aligned}
			d_{\textnormal{KL}} \left(L_\mu \left\| \frac{n}{n+m} \widehat{L}_n + \frac{m}{n+m} L_\mu \right.\right) \leq \frac{1}{n+m} \sum_{i=1}^k \sum_{j=1}^{\frac{n+m}{k}} d_{\textnormal{KL}} \left(\EE_{\Lambda^{(i)}_j | T^{(i)}_j = 0} \left[ \Lambda^{(i)}_j \right] \left\| \EE_{\Lambda^{(i)}_j} \left[ \Lambda^{(i)}_j \right] \right.\right).
		\end{aligned}
	\end{equation}
	Combining the above two upper bounds yields
	\begin{subequations}
		\begin{align}
			&d^\theta_{\textnormal{JS}} (\widehat{L}_n \| L_\mu) \nonumber \\
			&= \frac{n}{n+m} d_{\textnormal{KL}} \left(\widehat{L}_n \left\| \frac{n}{n+m} \widehat{L}_n + \frac{m}{n+m} L_\mu \right.\right) + \frac{m}{n+m} d_{\textnormal{KL}} \left(L_\mu \left\| \frac{n}{n+m} \widehat{L}_n + \frac{m}{n+m} L_\mu \right.\right) \nonumber \\
			&\leq \frac{1}{n+m} \sum_{i=1}^k \sum_{j=1}^{\frac{n+m}{k}} \left( \frac{n}{n+m} d_{\textnormal{KL}} \left(\EE_{\Lambda^{(i)}_j | T^{(i)}_j = 1} \left[ \Lambda^{(i)}_j \right] \left\| \EE_{\Lambda^{(i)}_j} \left[ \Lambda^{(i)}_j \right] \right.\right) +  \frac{m}{n+m} d_{\textnormal{KL}} \left(\EE_{\Lambda^{(i)}_j | T^{(i)}_j = 0} \left[ \Lambda^{(i)}_j \right] \left\| \EE_{\Lambda^{(i)}_j} \left[ \Lambda^{(i)}_j \right] \right.\right) \right) \label{eq:ineq-1} \\
			&\overset{(a)}{=} \frac{1}{n+m} \sum_{i=1}^k \sum_{j=1}^{\frac{n+m}{k}} d_{\textnormal{JS}}^\theta \left(\EE_{\Lambda^{(i)}_j | T^{(i)}_j = 1} \left[ \Lambda^{(i)}_j \right] \left\| \EE_{\Lambda^{(i)}_j | T^{(i)}_j = 0} \left[ \Lambda^{(i)}_j \right] \right.\right) \nonumber \\
			&\leq \frac{1}{n+m} \sum_{i=1}^k \sum_{j=1}^{\frac{n+m}{k}} I \left( \Lambda^{(i)}_j; T^{(i)}_j \right) \label{eq:ineq-2} \\
			&\leq \frac{1}{n+m} \sum_{i=1}^k \sum_{j=1}^{\frac{n+m}{k}} I \left( \Lambda^{(i)}_j; U^{(i)}_{[\frac{n}{k}]} \right), \label{eq:ineq-3}
		\end{align}
	\end{subequations}
	where $(a)$ is by the definition of $d_{\textnormal{JS}}^\theta$ in \eqref{eq:d_js} with $\theta = \frac{n}{n+m}$, \eqref{eq:ineq-2} is an application of Lemma~\ref{lem:js-bound} with $f(X) = X = \Lambda^{(i)}_j$, $Y = T^{(i)}_j$ and \eqref{eq:P_Tij}, \eqref{eq:ineq-3} holds by the data-processing inequality as $T^{(i)}_j$ is a function of $U^{(i)}_{[\frac{n}{k}]}$. The proof is complete.
	
	\subsection{A Proof Sketch for \eqref{eq:var-IPCIMI}} \label{subapp:var-IPCIMI}
	By \eqref{eq:lb-IP-exp-emp-risk} and \eqref{eq:lb-IP-exp-pop-risk}, we rewrite $\gen$ as
	\begin{equation}
		\begin{aligned}
			\gen &= \frac{1}{k} \sum_{i=1}^k \left( \EE_{\Lambda^{(i)}_{[\frac{n+m}{k}]}, U^{(i)}_{[\frac{n}{k}]}} \left[ \frac{k}{n} \sum_{j \not\in U^{(i)}_{[\frac{n}{k}]}} \Lambda^{(i)}_j \right] - \EE_{\Lambda^{(i)}_{[\frac{n+m}{k}]}, U^{(i)}_{[\frac{n}{k}]}} \left[ \frac{k}{n} \sum_{j \in U^{(i)}_{[\frac{n}{k}]}} \Lambda^{(i)}_j \right] \right) \\
			&\leq \frac{1}{k} \sum_{i=1}^k \sqrt{\frac{(n+m)^2}{2nm} d^\theta_{\textnormal{JS}} \left( \EE_{\Lambda^{(i)}_{[\frac{n+m}{k}]}, U^{(i)}_{[\frac{n}{k}]}} \left[ \frac{k}{n} \sum_{j \in U^{(i)}_{[\frac{n}{k}]}} \Lambda^{(i)}_j \right] \left\| \EE_{\Lambda^{(i)}_{[\frac{n+m}{k}]}, U^{(i)}_{[\frac{n}{k}]}} \left[ \frac{k}{n} \sum_{j \not\in U^{(i)}_{[\frac{n}{k}]}} \Lambda^{(i)}_j \right] \right. \right)},
		\end{aligned}
	\end{equation}
	where the inequality is an application of Lemma~\ref{lem:js-lb} with $\theta = \frac{n}{n+m}$. Following the similar steps as done for \eqref{eq:d_JS_IPCIMI}, we get
	\begin{equation}
		\begin{aligned}
			d^\theta_{\textnormal{JS}} \left( \EE_{\Lambda^{(i)}_{[\frac{n+m}{k}]}, U^{(i)}_{[\frac{n}{k}]}} \left[ \frac{k}{n} \sum_{j \in U^{(i)}_{[\frac{n}{k}]}} \Lambda^{(i)}_j \right] \left\| \EE_{\Lambda^{(i)}_{[\frac{n+m}{k}]}, U^{(i)}_{[\frac{n}{k}]}} \left[ \frac{k}{n} \sum_{j \not\in U^{(i)}_{[\frac{n}{k}]}} \Lambda^{(i)}_j \right] \right. \right) &\leq \frac{2k}{n+m} I\left( \Lambda^{(i)}_{[\frac{n+m}{k}]} ;  U^{(i)}_{[\frac{n}{k}]} \right) \\
			&\leq \frac{2k}{n+m} I\left( W ;  U^{(i)}_{[\frac{n}{k}]} | \hat{Z}^{(i)}_{[\frac{n+m}{k}]} \right),
		\end{aligned}
	\end{equation}
	where the last inequality is by the data-processing inequality. Combining everything together leads to the desired result.
	
	\subsection{Proof of Theorem~\ref{thm:tightest-gen-bound}}\label{subapp:tightest-gen-bound}
	It suffices to show that all of \eqref{eq:ineq-1}, \eqref{eq:ineq-2} and \eqref{eq:ineq-3} hold with equality. Specifically, we prove below that \eqref{eq:ineq-1} and \eqref{eq:ineq-3} take equality under Assumption~\ref{ass:invariant-alg}, and \eqref{eq:ineq-2} takes equality in the case of the binary loss function.
	
	\begin{enumerate}
		\item[\eqref{eq:ineq-1}] When Assumption~\ref{ass:invariant-alg} is satisfied, for all fixed $i_1,i_2 \in [k]$ and $j_1,j_2 \in [\frac{n+m}{k}]$, we have $P_{\Lambda^{(i_1)}_{j_1} | T^{(i_1)}_{j_1}} = P_{\Lambda^{(i_2)}_{j_2} | T^{(i_2)}_{j_2}}$ by combining \eqref{eq:invariant_u_2} and \eqref{eq:eqv_P_T1_U}. Thus, it follows that $\EE_{\Lambda^{(i_1)}_{j_1} | T^{(i_1)}_{j_1} = t} \left[ \Lambda^{(i_2)}_{j_1} \right] = \EE_{\Lambda^{(i_2)}_{j_2} | T^{(i_2)}_{j_2} = t} \left[ \Lambda^{(i_2)}_{j_2} \right]$ and $\EE_{\Lambda^{(i_1)}_{j_1}} \left[ \Lambda^{(i_2)}_{j_1} \right] = \EE_{\Lambda^{(i_2)}_{j_2}} \left[ \Lambda^{(i_2)}_{j_2} \right]$, with which the Jensen's inequality holds with equality.
		\item[\eqref{eq:ineq-2}] As proved in Appendix~\ref{subapp:js-bound}, \eqref{eq:ineq-2} holds equality with $\Lambda^{(i)}_j \in \{0, 1\}$ for all valid $i,j$. 
		\item[\eqref{eq:ineq-3}] When Assumption~\ref{ass:invariant-alg} is satisfied, we have $I ( W; T^{(i)}_j | \hat{Z}^{(i)}_j ) = I ( W; U^{(i)}_{[\frac{n}{k}]} | \hat{Z}^{(i)}_j )$, as shown in Proposition~\ref{pro:processed-SIPCIMI}. Following the proof as given in Appendix~\ref{subapp:processed-SIPCIMI}, it can be analogously proved that $I ( \Lambda^{(i)}_j; T^{(i)}_j ) = I ( \Lambda^{(i)}_j; U^{(i)}_{[\frac{n}{k}]} )$.
	\end{enumerate}
	\newpage
	
	\section{Calculation Details for Example~\ref{eg:Bernoulli}}\label{app:calculations}
	\setcounter{subsection}{0}
	In this section, we provide the calculations details when evaluating the considered bounds. To make it clear from the context, the realizations of capital random variables are represented with the corresponding lower-case letters. Before presenting the calculation details, we introduce some useful lemmas as below.
	\begin{lemma} \label{lem:bin_coe_pro}
		Let $X$ and $Y$ be two independent binomial random variables following $X \sim \Binomial(n, p)$ and $Y \sim \Binomial(m, p)$, respectively. If $\bar{X} = n - X$, then $\bar{X} \sim \Binomial(n, 1 - p)$. If $Z = X + Y$, then $Z \sim \Binomial(n + m, p)$.
	\end{lemma}
	
	\begin{proof}
		The proof is clear form the definition of Binomial distribution.
	\end{proof}
	
	\begin{lemma} \label{lem:exp_bin_logx}
		If $X$ is a binomial random variable following $X \sim \Binomial(n, p)$ and $a \geq 0$, then
		\begin{equation}
			\begin{aligned}
				\EE \left[ \log (X + a) \right] = \log (np + a) - \frac{np (1 - p)}{2 (np + a)^2} + O\left( \frac{1}{n^2} \right).
			\end{aligned}
		\end{equation}
	\end{lemma}
	
	\begin{proof}
		By Taylor's theorem, we expand the function $f(x) = \log (x + a)$ at point $x_0$ such that
		\begin{equation}
			\begin{aligned}
				f(x) = f(x_0) + \sum_{i=1}^{3} \frac{f^{(i)}(x)}{i !} (x - x_0)^i + \frac{f^{(i)}(\xi)}{4 !} (x - x_0)^4,
			\end{aligned}
		\end{equation}
		where $\xi$ is between $x_0$ and $x$. Let $x_0 = np$, it can be further calculated that
		\begin{equation}
			\begin{aligned}
				\log (x + a) =& \log (np + a) + \frac{x - np}{np + a} - \frac{(x - np)^2}{2 (np + a)^2} + \frac{(x - np)^3}{3 (np + a)^3} - \frac{(x - np)^4}{4 (\xi + a)^4}.
			\end{aligned}
			\label{eq:taylor_logx}
		\end{equation}
		Taking expectation on the upper bound of \eqref{eq:taylor_logx} yields
		\begin{equation}
			\begin{aligned}
				\EE \left[ \log (X + a) \right] &= \EE \left[ \log (np + a) + \frac{X - np}{np + a} - \frac{(X - np)^2}{2 (np + a)^2} + \frac{(X - np)^3}{3 (np + a)^3} \right] - \EE \left[ \frac{(X - np)^4}{4 (\xi + a)^4} \right] \\
				&= \log (np + a) + \frac{\EE \left[ X - np \right]}{np + a} - \frac{\EE \left[ (X - np)^2 \right]}{2 (np + a)^2} + \frac{\EE \left[ (X - np)^3 \right]}{3 (np + a)^3} - \EE \left[ \frac{(X - np)^4}{4 (\xi + a)^4} \right] \\
				&= \log (np + a) + 0 - \frac{np (1 - p)}{2 (np + a)^2} + \frac{np (1 - p) (1 - 2p)}{3 (np + a)^3} - \EE \left[ \frac{(X - np)^4}{4 (\xi + a)^4} \right] \\
				&= \log (np + a) - \frac{np (1 - p)}{2 (np + a)^2} + O\left( \frac{1}{n^2} \right) - \EE \left[ \frac{(X - np)^4}{4 (\xi + a)^4} \right].
			\end{aligned}
		\end{equation}
		It suffices to show that the expectation of the remainder term is of order $O(\frac{1}{n^2})$. Using Hoeffding's inequality, we can derive an upper bound for the tail probability that $X \leq \frac{1}{2} np$, which is shown as follows:
		\begin{equation}
			\begin{aligned}
				\Pr\left[ X \leq \frac{1}{2} np \right] &= \Pr\left[ X - \EE \left[ X \right] \leq -\frac{1}{2} np \right] \\
				&= \Pr\left[\sum_{i=1}^{n} X_i - \sum_{i=1}^{n} \EE \left[ X_i \right] \leq - \frac{1}{2} np \right] \\
				&\leq e^{-\frac{1}{2} p^2 n},
			\end{aligned}
		\end{equation}
		where $X_i \sim \Bernoulli(p)$ is a Bernoulli random variable for any $i \in [1, n]$. Then, by the tower property of expectation, we rewrite the expectation of the remainder term as
		\begin{equation}
			\begin{aligned}
				\EE \left[ \frac{(X - np)^4}{4 (\xi + a)^4} \right] &= \Pr\left[ X > \frac{1}{2} np \right] \EE \left[ \left. \frac{(X - np)^4}{4 (\xi + a)^4} \right| X > \frac{1}{2}np \right] + \Pr\left[ X \leq \frac{1}{2} np \right] \EE \left[ \left. \frac{(X - np)^4}{4 (\xi + a)^4} \right| X \leq \frac{1}{2}np \right].
			\end{aligned}\label{eq:exp_remainder}
		\end{equation}
		Since $X \in [n]$ and $\xi$ is within $np$ and $X$, we have 
		\begin{equation}
			\begin{aligned}
				\frac{(X - np)^4}{4 (\xi + a)^4} \leq \frac{\max((np-1)^4, (n-np)^4)}{4(1+a)^4} = O(n^4).
			\end{aligned}
		\end{equation}
		Due to fact that exponential decay dominates any polynomial growth, we see that the contribution of $X \leq \frac{1}{2} np$ to the expectation is negligible, i.e.,
		\begin{equation}
			\begin{aligned}
				\Pr\left[ X \leq \frac{1}{2} np \right] \EE \left[ \left. \frac{(X - np)^4}{4 (\xi + a)^4} \right| X \leq \frac{1}{2}np \right] \leq e^{-\frac{1}{2} p^2 n} \frac{\max((np-1)^4, (n-np)^4)}{4(1+a)^4} = o(\frac{1}{n^2}).
			\end{aligned}\label{eq:exp_X_out_range}
		\end{equation}
		In addition, similar to \eqref{eq:exp_remainder}, we have
		\begin{equation}
			\begin{aligned}
				\EE \left[ (X - np)^4 \right] &= \Pr\left[ X > \frac{1}{2} np \right] \EE \left[ \left. (X - np)^4 \right| X > \frac{1}{2}np \right] + \Pr\left[ X \leq \frac{1}{2} np \right] \EE \left[ \left. (X - np)^4 \right| X \leq \frac{1}{2}np \right] \\
				&\geq \Pr\left[ X > \frac{1}{2} np \right] \EE \left[ \left. (X - np)^4 \right| X > \frac{1}{2}np \right] \\
				&\geq \left( 1 - e^{-\frac{1}{2} p^2 n} \right) \EE \left[ \left. (X - np)^4 \right| X > \frac{1}{2}np \right].
			\end{aligned}\label{eq:exp_X_in_range}
		\end{equation}
		Now, by using the results in \eqref{eq:exp_X_out_range} and \eqref{eq:exp_X_in_range}, we can bound \eqref{eq:exp_remainder} as
		\begin{equation}
			\begin{aligned}
				0 \leq \EE \left[ \frac{(X - np)^4}{4 (\xi + a)^4} \right] &\overset{\eqref{eq:exp_X_out_range}}{\leq} \EE \left[ \left. \frac{(X - np)^4}{4 (\xi + a)^4} \right| X > \frac{1}{2}np \right] + o(\frac{1}{n^2}) \\
				&\leq \EE \left[ \left. \frac{(X - np)^4}{4 (\frac{1}{2}np + a)^4} \right| X > \frac{1}{2}np \right] + o(\frac{1}{n^2}) \\
				&\overset{\eqref{eq:exp_X_in_range}}{\leq} \frac{1}{1 - e^{-\frac{1}{2} p^2 n}} \frac{1}{4 (\frac{1}{2}np + a)^4} \EE \left[ (X - np)^4 \right] + o(\frac{1}{n^2}) \\
				&= \frac{1}{1 - e^{-\frac{1}{2} p^2 n}} \frac{np(1-p)(1+ 3(n-2)p(1-p))}{4 (\frac{1}{2}np + a)^4} + o(\frac{1}{n^2}) \\
				&= \Theta(\frac{1}{n^2}),
			\end{aligned}
		\end{equation}
		which implies \eqref{eq:exp_remainder} is of order $O(\frac{1}{n^2})$. The proof is complete.
	\end{proof}

	\begin{lemma} \label{lem:exp_bin_log-x}
		If $X$ is a binomial random variable following $X \sim \Binomial(n, p)$ and $b \geq n$. Let $\bar{X} = n - X$, then
		\begin{equation}
			\begin{aligned}
				\EE \left[ \log (b - \bar{X}) \right] = \log (b - n(1-p)) - \frac{np (1 - p)}{2 (b - n(1-p))^2} + O\left( \frac{1}{n^2} \right).
			\end{aligned}
		\end{equation}
	\end{lemma}
	
	\begin{proof}
		The results directly follows by applying Lemma~\ref{lem:exp_bin_logx} with $a = b-n$.
	\end{proof}
	
	\begin{lemma} \label{lem:exp_bin_xlogx}
		If $X$ is a binomial random variable following $X \sim \Binomial(n, p)$, then
		\begin{equation}
			\begin{aligned}
				\EE \left[ X \log X \right] = np \log np + \frac{1 - p}{2} + \frac{(1-p)(5-4p)}{6 np} + O\left( \frac{1}{n^2} \right).
			\end{aligned}
		\end{equation}
	\end{lemma}
	
	\begin{proof}
		Using the Taylor series in \eqref{eq:taylor_logx} with $a=0$, for all $x>0$, we have
		\begin{equation}
			\begin{aligned}
				x \log x = x \left( \log np + \frac{x - np}{np} - \frac{(x - np)^2}{2 (np)^2} + \frac{(x - np)^3}{3 (np)^3} - \frac{(x - np)^4}{4 \xi^4} \right),
			\end{aligned}
		\end{equation}
		where $\xi$ is between $np$ and $x$. Taking expectations on both sides yields
		\begin{equation}
			\begin{aligned}
				\EE \left[ X \log X \right] &= np \log np + \frac{\EE \left[ X(X - np) \right]}{np} - \frac{\EE \left[ X(X - np)^2 \right]}{2 (np)^2} + \frac{\EE \left[ X(X - np)^3 \right]}{3 (np)^3} - \EE \left[ \frac{X(X - np)^4}{4 \xi^4} \right].
			\end{aligned}
			\label{eq:ub_xlogx}
		\end{equation}
		Notice $X(X-np)^i = (X-np)^{i+1} + np(X-np)^i$ for any positive integer $i$, then it can be accordingly calculated that
		\begin{equation}
			\begin{aligned}
				&\EE \left[ X(X-np) \right] = \EE \left[ (X-np)^2 + np(X-np) \right] = np(1-p), \\
				&\EE \left[ X(X-np)^2 \right] = \EE \left[ (X-np)^3 + np(X-np)^2 \right] = np(1-p)(1-2p) + (np)^2(1-p), \\
				&\EE \left[ X(X-np)^3 \right] = \EE \left[ (X-np)^4 + np(X-np)^3 \right] = np(1-p)(1+ 3(n-2)p(1-p)) + (np)^2(1-p)(1-2p).
			\end{aligned}
		\end{equation}
		We plug all of the above expectations into \eqref{eq:ub_xlogx}, yielding
		\begin{equation}
			\begin{aligned}
				\EE \left[ X \log X \right] \leq& np \log np + \frac{np(1-p)}{np} - \frac{np(1-p)(1-2p) + (np)^2(1-p)}{2 (np)^2} \\
				& + \frac{np(1-p)(1+ 3(n-2)p(1-p)) + (np)^2(1-p)(1-2p)}{3 (np)^3} - \EE \left[ \frac{X(X - np)^4}{4 \xi^4} \right] \\
				=& np \log np + \frac{1 - p}{2} + \frac{(1-p)(5-4p)}{6 np} + O\left( \frac{1}{n^2} \right) - \EE \left[ \frac{X(X - np)^4}{4 \xi^4} \right].
			\end{aligned}\label{eq:exp_xlogx}
		\end{equation}
		Following a similar argument as proved in Lemma~\ref{lem:exp_bin_logx}, we are going to show that the expectation term scales as $O(\frac{1}{n^2})$. Specifically, we have
		\begin{equation}
			\begin{aligned}
				0 \leq \EE \left[ \frac{X (X - np)^4}{4 \xi^4} \right] &= \Pr\left[ X > \frac{1}{2} np \right] \EE \left[ \left. \frac{X (X - np)^4}{4 \xi^4} \right| X > \frac{1}{2}np \right] + \Pr\left[ X \leq \frac{1}{2} np \right] \EE \left[ \left. \frac{X (X - np)^4}{4 \xi^4} \right| X \leq \frac{1}{2}np \right] \\
				&\leq \EE \left[ \left. \frac{X (X - np)^4}{4 \xi^4} \right| X > \frac{1}{2}np \right] + o(\frac{1}{n^2}) \\
				&\leq \frac{4}{(np)^4} \EE \left[ \left. X (X - np)^4 \right| X > \frac{1}{2}np \right] + o(\frac{1}{n^2}) \\
				&\leq \frac{1}{1 - e^{-\frac{1}{2} p^2 n}} \frac{4}{(np)^4} \EE \left[ X (X - np)^4 \right] + o(\frac{1}{n^2}) \\
				&= \frac{1}{1 - e^{-\frac{1}{2} p^2 n}} \frac{4}{(np)^4} \EE \left[ X (X - np)^4 \right] + o(\frac{1}{n^2}) \\
				&= \Theta(\frac{1}{n^2}),
			\end{aligned}
		\end{equation}
		where
		\begin{equation}
			\begin{aligned}
				\EE \left[ X (X - np)^4 \right] &= \EE \left[ (X - np)^5 \right] + np \EE \left[ (X - np)^4 \right] \\
				&= np(1-p)(1-2p)(1 + 2p(1-p)(5n-6)) + np(1-p)(1+ 3(n-2)p(1-p)) \\
				&= \Theta(n^2).
			\end{aligned}
		\end{equation}
		Combining this result with \eqref{eq:exp_xlogx} completes the proof.
	\end{proof}
	
	\begin{lemma} \label{lem:log_frac}
		If $a, b, c, d$ are all constants, where $a > 0, b \geq 0, c > 0, d \geq 0$ and $ad \not= bc$, then $\log \frac{an + b}{cn + d}$ scales as $\log \frac{a}{c} + O\left( \frac{1}{n} \right)$.
	\end{lemma}
	
	\begin{proof}
		\begin{equation}
			\begin{aligned}
				\log \frac{an + b}{cn + d} &= \log \frac{a}{c} + \log \left( 1 + \frac{\frac{b}{a} - \frac{d}{c}}{n + \frac{d}{c}} \right) \\
				&\overset{(a)}{=} \log \frac{a}{c} + \frac{\frac{b}{a} - \frac{d}{c}}{n + \frac{d}{c}} - \frac{(\frac{b}{a} - \frac{d}{c})^2}{2(n + \frac{d}{c})^2} + o\left( \frac{1}{(n + \frac{d}{c})^3} \right) \\
				&= \log \frac{a}{c} + O\left( \frac{1}{n} \right)
			\end{aligned}
		\end{equation}
		where $(a)$ holds by applying $\log (1+x) = x -\frac{x^2}{2} + O(x^3) $ with $x = \frac{\frac{b}{a} - \frac{d}{c}}{n + \frac{d}{c}}$.
	\end{proof}
	
	\begin{lemma}[Bounds for a binomial coefficient {\cite[Lemma~7]{macwilliams1977theory}}] \label{lem:bin_coe_bound}
		If $1 \leq k \leq n - 1$, then
		\begin{equation}
			\begin{aligned}
				\sqrt{\frac{n}{8 k (n - k)}} e^{n H\left( \frac{k}{n} \right)} \leq \begin{psmallmatrix} n \\ k \end{psmallmatrix} \leq \sqrt{\frac{n}{2 \pi k (n - k)}} e^{n H\left( \frac{k}{n} \right)}.
			\end{aligned}
		\end{equation}
	\end{lemma}
	
	\begin{lemma} \label{lem:bin_coe_approx}
		If $X$ is a binomial random variable following $X \sim \Binomial(n, p)$, then
		\begin{equation}
			\begin{aligned}
				\EE \left[ \log \begin{psmallmatrix} n \\ X \end{psmallmatrix} \right] = n H(p) - \frac{1}{2} \log n + O \left( 1 \right).
			\end{aligned}
		\end{equation}
	\end{lemma}
	\begin{proof}
		First, by Lemma~\ref{lem:bin_coe_bound}, the lower bound of the expected binomial coefficient can be calculated as
		\begin{equation}
			\begin{aligned}
				\EE \left[ \log \begin{psmallmatrix} n \\ X \end{psmallmatrix} \right] \geq& \EE \left[ \frac{1}{2} \left( \log n - \log X - \log \bar{X} - \log 8 \right) + n H(\frac{X}{n}) \right] \\
				=& \frac{1}{2} \log n - 3 \log 2 - \EE \left[ \frac{1}{2} \left( \log X + \log \bar{X} \right) + X \log \frac{X}{n} + \bar{X} \log \frac{\bar{X}}{n} \right] \\
				=& \frac{1}{2} \log n - 3 \log 2 - \EE \left[ \frac{1}{2} \left( \log X + \log \bar{X} \right) + X \log X + \bar{X} \log \bar{X} - X \log n - \bar{X} \log n \right] \\
				=& \frac{1}{2} \log n - 3 \log 2 - \left( \frac{1}{2} \left( \log np - \frac{1-p}{2 np} + \log n(1 - p) - \frac{p}{2 n(1 - p)} + O \left( \frac{1}{n^2} \right) \right) \right. \\
				& + np \log np + \frac{1 - p}{2} + \frac{(1-p)(5-4p)}{6 np} + O \left( \frac{1}{n^2} \right) \\
				& + n(1 - p) \log n(1 - p) + \frac{p}{2} + \frac{p(1+4p)}{6 n (1 - p)} + O \left( \frac{1}{n^2} \right) \\
				& - np \log n - n(1 - p) \log n \bigg) \\
				=& - \frac{1}{2} \log np (1 - p) - 3 \log 2 - \frac{1}{2} - np \log p - n(1 - p) \log (1 - p) + O \left( \frac{1}{n} \right) \\
				=& n H(p) - \frac{1}{2} \log n - \frac{1}{2} \log p (1 - p) - 3 \log 2 - \frac{1}{2} + O \left( \frac{1}{n} \right).
			\end{aligned}
		\end{equation}
		Likewise, the upper bound can be obtained as
		\begin{equation}
			\begin{aligned}
				\EE \left[ \log \begin{psmallmatrix} n \\ X \end{psmallmatrix} \right] \leq n H(p) - \frac{1}{2} \log n - \frac{1}{2} \log p (1 - p) - \log 2\pi - \frac{1}{2} + O \left( \frac{1}{n} \right).
			\end{aligned}
		\end{equation}
		Combining the lower and upper bounds completes the proof.	
	\end{proof}
	
	\begin{table*}[t]
		\centering
		\caption{Expressions of $W$ under Examples\ref{eg:Bernoulli} and \ref{eg:Gaussian}}
		\begin{threeparttable}
			\begin{tabular}{|c|c|c|} \hline
				\multicolumn{2}{|c|}{Setting} & \multirow{2}{*}{Bound} \\ \cline{1-2}
				$(m, k)$ & Expression of $W$ & \\ \hline\hline
				$(m, k)$ & $\frac{1}{n} \sum_{i=1}^{k} \sum_{j=1}^{\frac{n}{k}} \hat{Z}^{(i)}_{U^{(i)}_j}$ & IPCIMI \\ \hline
				$(m, 1)$ & $\frac{1}{n} \sum_{i=1}^{n} \hat{Z}_{U_i}$ & L$m$O-CMI \\ \hline
				$(1, 1)$ & $\frac{1}{n} \left( \sum_{i=1}^{n+1} \hat{Z}_i - \hat{Z}_U \right)$ & LOO-CMI \cite{rammal2022leave} \\ \hline
				$(n, n)$ & $\frac{1}{n} \sum_{i=1}^{n} \hat{Z}_{i + R_i n}$ & ICIMI \cite{zhou2022individually} \\ \hline
				$(\infty, k)$ & $\frac{1}{n} \sum_{i=1}^{n} Z_i$ & IPMI \\ \hline
				$(\infty, 1)$ & $\frac{1}{n} \sum_{i=1}^{n} Z_i$ & MI \cite{xu2017information} \\ \hline
				$(\infty, n)$ & $\frac{1}{n} \sum_{i=1}^{n} Z_i$ & IMI \cite{bu2020tightening} \\ \hline
				$(m, m)$ & $\frac{1}{n} \sum_{i=1}^{m} \left( \sum_{j=1}^{\frac{n}{m} + 1} \hat{Z}^{(i)}_j - \hat{Z}^{(i)}_{U^{(i)}} \right)$ & LOFO-CMI \\ \hline
			\end{tabular} 
		\end{threeparttable} 
		\label{tab:expression_W}
	\end{table*}
	
	In the following, we present the calculation details of each bound under Example~\ref{eg:Bernoulli}. For ease of reference, the expressions of the hypothesis $W$ for different bounds are summarized in Table~\ref{tab:expression_W}.
	
	\subsection{True Generalization Error} \label{subapp:true-gen-err}
	Under the Bernoulli example in Example~\ref{eg:Bernoulli}, the ERM solution is $W = \frac{1}{n} \sum_{i=1}^{n} Z_i$ and the expected generalization error is given by
	\begin{equation}
		\begin{aligned}
			\gen = \EE \left[ \left( W - Z \right)^2 - \frac{1}{n} \sum_{i=1}^{n} \left( W - Z_i \right)^2 \right],
		\end{aligned}
	\end{equation}
	where $Z \sim \Bernoulli(p)$ is a Bernoulli random variable independent of $Z_i$'s. Collecting all the facts that $W \indep Z$, $\EE[W] = \EE[Z] = p$ and $\EE [Z_i^2] = \EE [Z^2] = p$, we further calculate the above expression as
	\begin{equation}
		\begin{aligned}
			\gen &= \EE \left[ W^2 \right] - 2 \EE \left[ W \right] \EE \left[ Z \right] + \EE \left[ Z^2 \right] - \frac{1}{n} \sum_{i=1}^{n} \EE \left[ W^2 \right] + \frac{2}{n} \sum_{i=1}^{n} \EE \left[ W Z_i \right] - \frac{1}{n} \sum_{i=1}^{n} \EE \left[ Z_i^2 \right] \\
			&= -2p^2 + \frac{2}{n} \sum_{i=1}^{n} \EE \left[ W Z_i \right] \\
			&= -2p^2 + \frac{2}{n} \sum_{i=1}^{n} \EE \left[ \frac{1}{n} \sum_{j=1}^{n} Z_j Z_i \right] \\
			&= -2p^2 + \frac{2}{n^2} \sum_{i=1}^{n} \sum_{j=1}^{n} \EE \left[ Z_j Z_i \right] \\
			&= -2p^2 + \frac{2}{n^2} \sum_{i=1}^{n} \left( \sum_{j \not= i} \EE \left[ Z_i \right] \EE \left[ Z_j \right] + \EE \left[ Z_i^2 \right] \right) \\
			&= -2p^2 + \frac{2}{n} \left( (n-1)p^2 + p \right) \\
			&= \frac{2p(1-p)}{n}.
		\end{aligned}
	\end{equation}
	
	\subsection{MI Case} \label{subapp:cal_MI}
	We first evaluate the MI bound in \eqref{eq:MI}, where the hypothesis is $W = \frac{1}{n} \sum_{i=1}^{n} Z_i$. Let $w$ and $z_{[n]} = (z_1, \ldots, z_n)$ be the realizations of $W$ and $Z_{[n]}$. Let $x = \sum_{i=1}^{n} z_i$, where $x$ in an integer from $[0, n]$, then we can obtain the joint, marginal, conditional probabilities as follows:
	\begin{align}
		&P_{W | Z_{[n]}}(\frac{x}{n} | z_{[n]}) = 1, \\
		&P_{W, Z_{[n]}}(\frac{x}{n}, z_{[n]}) = p^{x} (1 - p)^{n - x}, \\
		&P_{W}(\frac{x}{n}) = B_{n,x}(p).
	\end{align}
	By definition of mutual information, $I( W; Z_{[n]} )$ can be calculated as
	\begin{equation}
		\begin{aligned}
			I\left( W; Z_{[n]} \right) =& \sum_{w, z_{[n]}} P_{W, Z_{[n]}}(w, z_{[n]}) \log \frac{P_{W | Z_{[n]}}(w | z_{[n]})}{P_{W}(w)} \\
			=& - \sum_{x=0}^{n} B_{n,x}(p) \log B_{n,x}(p).
		\end{aligned}
	\end{equation}
	Let $x$ be the realization of a Binomial random variable $X \sim \Binomial(n, p)$ and let $\bar{X} = n - X \sim \Binomial(n, 1 - p)$, we can rewritten $I( W; Z_{[n]} )$ as
	\begin{equation}
		\begin{aligned}
			I\left( W; Z_{[n]} \right) = - \EE \left[ \log \begin{psmallmatrix} n \\ X \end{psmallmatrix} + X \log p + \bar{X} \log (1 - p) \right],
		\end{aligned}\label{eq:Ber_MI}
	\end{equation}
	plugging which back to the standard MI bound yields
	\begin{equation}
		\begin{aligned}
			\sqrt{\frac{1}{2n} I\left( W; Z_{[n]} \right)} &= \sqrt{ - \frac{1}{2n} \EE \left[ \log \begin{psmallmatrix} n \\ X \end{psmallmatrix} + X \log p + \bar{X} \log (1 - p) \right]} \\
			&\overset{(a)}{=} \sqrt{\frac{1}{2n} \left( - n H(p) + \frac{1}{2} \log n + O \left( 1 \right) - np \log p - n(1 - p) \log (1 - p) \right)} \\
			&= \sqrt{- \frac{1}{2} H(p) + \frac{1}{4} \frac{\log n}{n} + O \left( \frac{1}{n} \right) + \frac{1}{2} H(p)} \\
			&= \sqrt{\frac{1}{4} \frac{\log n}{n} + O \left( \frac{1}{n} \right)},
		\end{aligned}
	\end{equation}
	where $(a)$ is by Lemma~\ref{lem:bin_coe_approx}. Thus, the MI bound scales as $O\left( \sqrt{\frac{\log n}{n}} \right)$.
	
	\subsection{IMI Case} \label{subapp:cal_IMI}
	This subsection evaluates the IMI bound in \eqref{eq:IMI}. For the realizations $w$ and $z_i$, let $x = \sum_{j=1}^{n} z_j - z_i$, where $x$ in an integer from $[0, n - 1]$, then we have
	\begin{align}
		&P_{W | Z_i}(\frac{x + z_i}{n} | z_i) = B_{n-1,x}(p), \\
		&P_{W, Z_i}(\frac{x + z_i}{n}, z_i) = B_{n-1,x}(p) p^{z_i} (1-p)^{1-z_i}, \\
		&P_{W}(\frac{x + z_i}{n}) = B_{n, x + z_i}(p).
	\end{align}
	With the above elements obtained, we can further calculate the mutual information according to its definition:
	\begin{equation}
		\begin{aligned}
			I\left( W; Z_i \right) =& \sum_{w, z_i} P_{W, Z_i}(w, z_i) \log \frac{P_{W | Z_i}(w | z_i)}{P_{W}(w)} \\
			=& \sum_{w} P_{W, Z_i}(w, 0) \log \frac{P_{W | Z_i}(w | 0)}{P_{W}(w)} + \sum_{w} P_{W, Z_i}(w, 1) \log \frac{P_{W | Z_i}(w | 1)}{P_{W}(w)} \\
			=& \sum_{x=0}^{n-1} (1-p) B_{n-1,x}(p) \log \frac{B_{n-1,x}(p)}{B_{n,x}(p)} + \sum_{x=0}^{n-1} p B_{n-1,x}(p) \log \frac{B_{n-1,x}(p)}{B_{n,x+1}(p)} \\
			=& (1 - p) \sum_{x=0}^{n-1} B_{n-1,x}(p) \log \frac{n - x}{n (1 - p)} + p \sum_{x=0}^{n-1} B_{n-1,x}(p) \log \frac{x + 1}{n p}.
		\end{aligned}
	\end{equation}
	Let $X \sim \Binomial(n - 1, p)$ and $\bar{X} = n - 1 - X \sim \Binomial(n - 1, 1 - p)$, after arrangement and simplification we arrive at
	\begin{equation}
		\begin{aligned}
			I\left( W; Z_i \right) = H(p) - \log n + p \EE \left[ \log (X + 1) \right] + (1 - p) \EE \left[ \log (\bar{X} + 1) \right].
		\end{aligned}\label{eq:Ber_IMI}
	\end{equation}
	Next, applying the approximation introduced in Lemma~\ref{lem:exp_bin_logx} yields
	\begin{equation}
		\begin{aligned}
			I\left( W; Z_i \right) =& H(p) - \log n + p \log ((n - 1) p + 1) - \frac{(n - 1) (1 - p) p}{2 \left( (n - 1) p + 1 \right)^2} \\
			&+ (1 - p) \log ((n - 1) (1 - p) + 1) - \frac{(n - 1) (1 - p) p}{2 \left( (n - 1) (1 - p) + 1 \right)^2} + O\left(\frac{1}{n^2}\right) \\
			=& H(p) + p \log \frac{np + 1 - p}{n} + (1 - p) \log \frac{(1 - p) n + p}{n} + O\left( \frac{1}{n} \right).
		\end{aligned}
	\end{equation}
	Then, it follows by Lemma~\ref{lem:log_frac} that
	\begin{equation}
		\begin{aligned}
			I\left( W; Z_i \right) =& H(p) + p \log p + (1 - p) \log (1 - p) + O\left( \frac{1}{n} \right) \\
			=& O\left( \frac{1}{n} \right),
		\end{aligned}
	\end{equation}
	which implies that the IMI bound \eqref{eq:IMI} is of order $O\left( \frac{1}{\sqrt{n}} \right)$.
	
	\subsection{IPMI Case} \label{subapp:cal_IPMI}
	We continue to evaluate the IPMI bound in \eqref{eq:IPMI}, which is a general case of the MI and IMI bound. Under the partitioned setting, the hypothesis $W = \frac{1}{n} \sum_{i=1}^{n} Z_i$ can be rewritten as $W = \frac{1}{n} \sum_{i=1}^{k} \sum_{j=1}^{\frac{n}{k}} Z^{(i)}_j$. For all $i \in [k]$ and $j \in [\frac{n}{k}]$, let $w$ and $z^{(i)}_j $ be the realization of $W$ and $Z^{(i)}_j$, respectively. Furthermore, let $x = \sum_{j=1}^{\frac{n}{k}} z^{(i)}_j$ and $y = n w - x$, with $x \in [0, \frac{n}{k}]$ and $y \in [0, n - \frac{n}{k}]$. Then, one can calculate the following elemental probabilities as
	\begin{align}
		&P_{W | Z_{[\frac{n}{k}]}^{(i)}}(\frac{x + y}{n} | z_{[\frac{n}{k}]}^{(i)}) = B_{n - \frac{n}{k}, y} (p), \\
		&P_{W, Z_{[\frac{n}{k}]}^{(i)}}(\frac{x + y}{n}, z_{[\frac{n}{k}]}^{(i)}) = B_{n - \frac{n}{k}, y} (p) p^{x} (1-p)^{\frac{n}{k}-x}, \\
		&P_{W}(\frac{x + y}{n}) = B_{n, x+y} (p).
	\end{align}
	with which $I( W; Z_{[\frac{n}{k}]}^{(i)} )$ can be further obtained as
	\begin{equation}
		\begin{aligned}
			I\left( W; Z_{[\frac{n}{k}]}^{(i)} \right) =& \sum_{y=0}^{n - \frac{n}{k}} \sum_{z_{[\frac{n}{k}]}^{(i)}} P_{W, Z_{[\frac{n}{k}]}^{(i)}}(\frac{x + y}{n}, z_{[\frac{n}{k}]}^{(i)}) \log \frac{P_{W | Z_{[\frac{n}{k}]}^{(i)}}(\frac{x + y}{n} | z_{[\frac{n}{k}]}^{(i)})}{P_{W}(\frac{x + y}{n})}  \\
			=& \sum_{y=0}^{n - \frac{n}{k}} B_{n - \frac{n}{k}, y} (p) \sum_{x=0}^{\frac{n}{k}} B_{\frac{n}{k}, x} (p) \log \frac{\begin{psmallmatrix} n - \frac{n}{k} \\ y \end{psmallmatrix}}{\begin{psmallmatrix} n \\ x + y \end{psmallmatrix} p^x (1 - p)^{\frac{n}{k} - x}} \\
			=& \EE \left[ \log \begin{psmallmatrix} n - \frac{n}{k} \\ Y \end{psmallmatrix} - \log \begin{psmallmatrix} n \\ X + Y \end{psmallmatrix} - X \log p - \bar{X} \log (1 - p) \right],
		\end{aligned}
	\end{equation}
	where $X \sim \Binomial(\frac{n}{k}, p)$, $Y \sim \Binomial(n - \frac{n}{k}, p)$ and $\bar{X} = \frac{n}{k} - X \sim \Binomial(\frac{n}{k}, 1 - p)$.
	
	Using Lemmas~\ref{lem:bin_coe_pro} and \ref{lem:bin_coe_approx}, the above mutual information can be approximated as
	\begin{equation}
		\begin{aligned}
			I\left( W; Z_{[\frac{n}{k}]}^{(i)} \right) =& (n - \frac{n}{k}) H(p) - \frac{1}{2} \log (n - \frac{n}{k}) - n H(p) + \frac{1}{2} \log n + O\left( 1 \right) - \frac{n}{k} p \log p - \frac{n}{k} (1 - p) \log (1 - p) \\
			=& - \frac{n}{k} H(p) + \frac{1}{2} \log \frac{k}{k - 1} + O\left( 1 \right) + \frac{n}{k} H(p) \\
			=& \frac{1}{2} \log \frac{k}{k - 1} + O\left( 1 \right).
		\end{aligned}
	\end{equation}
	Since $k$ is a constant, $I( W; Z_{[\frac{n}{k}]}^{(i)} )$ is of a constant order $O\left( 1 \right)$, and it is simple to further verify that the IPMI bound in \eqref{eq:IPMI} is of order $O\left( \frac{1}{\sqrt{n}} \right)$. 
	%If $k \propto n$, by Taylor expansion we have $I( W; Z_{[\frac{n}{k}]}^{(i)} ) = \frac{1}{2(k - 1)} + o(\frac{1}{n}) + O\left( C \right)$, which converges to a constant. Since the factor $\frac{k}{2n}$ under the square root of \eqref{eq:IPMI} is also a constant, the resultant IPMI bound is of a constant order. 
	
	\subsection{LOO-CMI Case} \label{subapp:cal_LOO}
	\label{subsubsec: LOO-case}
	Under this case, the learner is simplified as $W = \frac{1}{n} (\sum_{i=1}^{n+1} \hat{Z}_i - \hat{Z}_U)$, and by Bayes rule we get
	\begin{equation}
		\begin{aligned}
			P_{U | W, \hat{Z}_{[n+1]}}(u | w, \hat{z}_{[n+1]}) = \frac{P_{W| \hat{Z}_{[n+1]}, U}(w| \hat{z}_{[n+1]}, u) P_{\hat{Z}_{[n+1]}}(\hat{z}_{[n+1]}) P_{U}(u)}{\sum_{u} P_{W, \hat{Z}_{[n+1]}, U}(w, \hat{z}_{[n+1]}, u)}.
		\end{aligned}
	\end{equation}
	For $w = \frac{x}{n}$, it is easy to verify that the number of $1$'s in $\hat{z}_{[n+1]}$ is either $x$ when $\hat{z}_u = 0$ or $x+1$ when $\hat{z}_u = 1$ since otherwise $P_{W| \hat{Z}_{[n+1]}, U}$ is zero. When $\hat{z}_u = 0$ and $\sum_{i = 1}^{n+1} \hat{z}_i = x$, we have
	\begin{equation} 
		\begin{aligned}
			P_{U | W, \hat{Z}_{[n+1]}}(u | \frac{x}{n}, \hat{z}_{[n+1]}) &= \frac{\frac{1}{n+1} p^x (1-p)^{n+1-x}}{\frac{n+1-x}{n+1} p^x (1-p)^{n+1-x}} \\
			&= \frac{1}{n+1-x}.
			\label{eq:pdf-zu-0}
		\end{aligned}
	\end{equation}
	Likewise, when $\hat{z}_u = 1$ and $\sum_{i = 1}^{n+1} \hat{z}_i = x+1$, it holds that:
	\begin{equation}
		\begin{aligned}
			P_{U | W, \hat{Z}_{[n+1]}}(u | \frac{x}{n}, \hat{z}_{[n+1]}) = \frac{1}{x+1}.
			\label{eq:pdf-zu-1}
		\end{aligned}
	\end{equation}
	For $w = \frac{x}{n}$, combining \eqref{eq:pdf-zu-0} and \eqref{eq:pdf-zu-1} gives
	\begin{equation}
		\begin{aligned}
			P_{U | W, \hat{Z}_{[n+1]}}(u | \frac{x}{n}, \hat{z}_{[n+1]}) = \begin{cases}
				\frac{1}{n+1-x}, & \text{if $\hat{z}_u = 0$ and $\sum_{i = 1}^{n+1} \hat{z}_i = x$} \\
				\frac{1}{x+1}, & \text{if $\hat{z}_u = 1$ and $\sum_{i = 1}^{n+1} \hat{z}_i = x+1$} \\
				0, & \text{otherwise}
			\end{cases}.
		\end{aligned}
	\end{equation}
	By definition, we compute the mutual information as
	\begin{equation}
		\begin{aligned}
			I\left( W; U | \hat{Z}_{[n+1]} \right) &= \EE_{W, \hat{Z}_{[n+1]}, U} \left[ \log \frac{P_{U | W, \hat{Z}_{[n+1]}}}{P_{U}} \right] \\
			&= \sum_{x=0}^{n} \frac{n+1-x}{n+1} B_{n+1,x}(p) \log \frac{n+1}{n+1-x} + \sum_{x=0}^{n} \frac{x+1}{n+1} B_{n+1,x+1}(p) \log \frac{n+1}{x+1} \\
			&= \log (n+1) - \sum_{x=0}^{n} B_{n,x}(p) \left( (1-p)\log (n+1-x) + p \log (x+1) \right).
		\end{aligned}
	\end{equation}
	Let $X$ be a binomial random variable $\Binomial(n, p)$, then the above mutual information can be rewritten as
	\begin{equation}
		\begin{aligned}
			I\left( W; U | \hat{Z}_{[n+1]} \right) = \log (n+1) - \EE \left[ (1-p)\log (n+1-X) + p \log (X+1) \right].
		\end{aligned}
		\label{eq:Ber_LOO}
	\end{equation}
	Using the approximations in Lemma~\ref{lem:exp_bin_logx}, we get
	\begin{equation}
		\begin{aligned}
			I\left( W; U | \hat{Z}_{[n+1]} \right) =& \log (n + 1) - p \left( \log (np + 1) - \frac{np (1 - p)}{2 (np + 1)^2} \right) \\
			&- (1 - p) \left( \log (n + 1 - np) - \frac{np (1 - p)}{2 (n + 1 - np)^2} \right) + O\left( \frac{1}{n^2} \right) \\
			=& p \log \frac{n + 1}{np + 1} + (1 - p) \log \frac{n + 1}{(1 - p) n + 1} + O\left( \frac{1}{n} \right)\\
			\overset{(a)}{=}& - p \log p - (1-p) \log (1 - p) + O\left( \frac{1}{n} \right) \\
			=& H(p) + O\left( \frac{1}{n} \right),
		\end{aligned}
	\end{equation}
	where in $(a)$ we apply Lemma~\ref{lem:log_frac}. Finally, it is shown that the LOO bound converges to a constant $\sqrt{\frac{1}{2} H(p)}$.
	
	\subsection{ICIMI Case} \label{subapp:cal_ICIMI}
	\label{subsubsec: ICIMI-case}
	In this subsection, we evaluate the ICIMI bound in \eqref{eq:ICIMI}, where the hypothesis is $W = \frac{1}{n} \sum_{i=1}^{n} \hat{Z}_{i + R_i n}$. Let $w = \frac{x}{n}$, $r_{[n]} = (r_1, \ldots, r_n)$ and $\hat{z}_{[2n]} = (\hat{z}_1, \ldots, \hat{z}_{2n})$ be the realizations, where $x \in \{0, \ldots, n\}$. First, we calculate the conditional probability as
	\begin{equation}
		\begin{aligned}
			P_{W | R_i, \hat{Z}_{i}, \hat{Z}_{i+n}}(\frac{x}{n} | r_i, \hat{z}_{i}, \hat{z}_{i+n}) =& B_{n-1, x - \hat{z}_{i + r_i n}}(p).
		\end{aligned}
	\end{equation}
	By marginalizing out $R_i$, we get
	\begin{equation}
		\begin{aligned}
			P_{W | \hat{Z}_{i}, \hat{Z}_{i+n}}(\frac{x}{n} | \hat{z}_{i}, \hat{z}_{i+n}) =& P_{R_i}(0) P_{W | R_i, \hat{Z}_{i}, \hat{Z}_{i+n}}(\frac{x}{n} | 0, \hat{z}_{i}, \hat{z}_{i+n}) + P_{R_i}(1) P_{W | R_i, \hat{Z}_{i}, \hat{Z}_{i+n}}(\frac{x}{n} | 1, \hat{z}_{i}, \hat{z}_{i+n}) \\
			=& \frac{1}{2} B_{n-1, x - \hat{z}_{i}}(p) + \frac{1}{2} B_{n-1, x - \hat{z}_{i + n}}(p).
		\end{aligned}
	\end{equation}
	Then, the joint probability is given by
	\begin{equation} \label{eq:ICIMI_P_W_Ri_Zi_Zi+n}
		\begin{aligned}
			P_{W, R_i, \hat{Z}_{i}, \hat{Z}_{i+n}}(\frac{x}{n}, r_i, \hat{z}_{i}, \hat{z}_{i+n}) =& P_{W | R_i, \hat{Z}_{i}, \hat{Z}_{i+n}}(\frac{x}{n} | r_i, \hat{z}_{i}, \hat{z}_{i+n}) P_{R_i}(r_i) P_{\hat{Z}_{i}}(\hat{z}_{i}) P_{\hat{Z}_{i+n}} (\hat{z}_{i+n})\\
			=& \frac{1}{2} B_{n-1, x - \hat{z}_{i + r_i n}}(p) p^{\hat{z}_{i}} (1-p)^{1-\hat{z}_{i}} p^{\hat{z}_{i+n}} (1-p)^{1-\hat{z}_{i+n}}.
		\end{aligned}
	\end{equation}
	With the above elemental probabilities obtained, we continue to calculate the mutual information:
	\begin{equation}
		\begin{aligned}
			I\left( W; R_i | \hat{Z}_{i}, \hat{Z}_{i+n} \right) =& \sum_{w, r_i, \hat{z}_{i}, \hat{z}_{i+n}} P_{W, R_i, \hat{Z}_{i}, \hat{Z}_{i+n}}(w, r_i, \hat{z}_{i}, \hat{z}_{i+n}) \log \frac{P_{W | R_i, \hat{Z}_{i}, \hat{Z}_{i+n}}(w | r_i, \hat{z}_{i}, \hat{z}_{i+n})}{P_{W | \hat{Z}_{i}, \hat{Z}_{i+n}}(w | \hat{z}_{i}, \hat{z}_{i+n})} \\
			=& \sum_{w} P_{W, R_i, \hat{Z}_{i}, \hat{Z}_{i+n}}(w, 0, 0, 0) \log \frac{P_{W | R_i, \hat{Z}_{i}, \hat{Z}_{i+n}}(w | 0, 0, 0)}{P_{W | \hat{Z}_{i}, \hat{Z}_{i+n}}(w | 0, 0)} + \ldots \\
			& + \sum_{w} P_{W, R_i, \hat{Z}_{i}, \hat{Z}_{i+n}}(w, 0, 1, 1) \log \frac{P_{W | R_i, \hat{Z}_{i}, \hat{Z}_{i+n}}(w | 0, 1, 1)}{P_{W | \hat{Z}_{i}, \hat{Z}_{i+n}}(w | 1, 1)} \\
			& + \sum_{w} P_{W, R_i, \hat{Z}_{i}, \hat{Z}_{i+n}}(w, 1, 0, 0) \log \frac{P_{W | R_i, \hat{Z}_{i}, \hat{Z}_{i+n}}(w | 1, 0, 0)}{P_{W | \hat{Z}_{i}, \hat{Z}_{i+n}}(w | 0, 0)} + \ldots \\
			& + \sum_{w} P_{W, R_i, \hat{Z}_{i}, \hat{Z}_{i+n}}(w, 1, 1, 1) \log \frac{P_{W | R_i, \hat{Z}_{i}, \hat{Z}_{i+n}}(w | 1, 1, 1)}{P_{W | \hat{Z}_{i}, \hat{Z}_{i+n}}(w | 1, 1)}.
		\end{aligned}
		\label{eq:ICIMI_cal}
	\end{equation}
	Notice the following facts: $(a)$ whenever $\hat{z}_{i} = \hat{z}_{i+n}$,
	\begin{equation*}
		\log \frac{P_{W | R_i, \hat{Z}_{i}, \hat{Z}_{i+n}}(w | r_i, \hat{z}_{i}, \hat{z}_{i+n})}{P_{W | \hat{Z}_{i}, \hat{Z}_{i+n}}(w | \hat{z}_{i}, \hat{z}_{i+n})} = 0,
	\end{equation*}
	$(b)$ for $(R_i, \hat{Z}_{i}, \hat{Z}_{i+n}) = (0,0,1)$ and $(R_i, \hat{Z}_{i}, \hat{Z}_{i+n}) = (1,1,0)$,
	\begin{equation}
		\begin{aligned}
			\sum_{w} P_{W, R_i, \hat{Z}_{i}, \hat{Z}_{i+n}}(w, 0, 0, 1) \log \frac{P_{W | R_i, \hat{Z}_{i}, \hat{Z}_{i+n}}(w | 0, 0, 1)}{P_{W | \hat{Z}_{i}, \hat{Z}_{i+n}}(w | 0, 1)} =& \frac{1}{2} \sum_{x = 0}^{n - 1} p(1-p) B_{n-1,x}(p) \log \frac{2}{\frac{x}{n - x} \cdot \frac{1 - p}{p} + 1} \\
			=& \sum_{w} P_{W, R_i, \hat{Z}_{i}, \hat{Z}_{i+n}}(w, 1, 1, 0) \log \frac{P_{W | R_i, \hat{Z}_{i}, \hat{Z}_{i+n}}(w | 1, 1, 0)}{P_{W | \hat{Z}_{i}, \hat{Z}_{i+n}}(w | 1, 0)},
		\end{aligned}
	\end{equation}
	and similarly, $(c)$ for $(R_i, \hat{Z}_{i}, \hat{Z}_{i+n}) = (0,1,0)$ and $(R_i, \hat{Z}_{i}, \hat{Z}_{i+n}) = (1,0,1)$,
	\begin{equation}
		\begin{aligned}
			\sum_{w} P_{W, R_i, \hat{Z}_{i}, \hat{Z}_{i+n}}(w, 0, 1, 0) \log \frac{P_{W | R_i, \hat{Z}_{i}, \hat{Z}_{i+n}}(w | 0, 1, 0)}{P_{W | \hat{Z}_{i}, \hat{Z}_{i+n}}(w | 1, 0)} =& \frac{1}{2} \sum_{x = 0}^{n - 1} p(1-p) B_{n-1,x}(p) \log \frac{2}{\frac{n - 1 - x}{x + 1} \cdot \frac{p}{1 - p} + 1} \\
			=& \sum_{w} P_{W, R_i, \hat{Z}_{i}, \hat{Z}_{i+n}}(w, 1, 0, 1) \log \frac{P_{W | R_i, \hat{Z}_{i}, \hat{Z}_{i+n}}(w | 1, 0, 1)}{P_{W | \hat{Z}_{i}, \hat{Z}_{i+n}}(w | 0, 1)}.
		\end{aligned}
	\end{equation}
	Plugging all the summand terms back to \eqref{eq:ICIMI_cal}, we get
	\begin{equation}
		\begin{aligned}
			I\left( W; R_i | \hat{Z}_{i}, \hat{Z}_{i+n} \right) =& \sum_{x = 0}^{n - 1} p(1-p) B_{n-1,x}(p) \log \frac{2}{\frac{x}{n - x} \cdot \frac{1 - p}{p} + 1} + \sum_{x = 0}^{n - 1} p(1-p) B_{n-1,x}(p) \log \frac{2}{\frac{n - 1 - x}{x + 1} \cdot \frac{p}{1 - p} + 1} \\
			=& 2 p(1 - p) \log 2 - \sum_{x = 0}^{n - 1} p(1-p) B_{n-1,x}(p) \log \left[ \left( \frac{x}{n - x} \cdot \frac{1 - p}{p} + 1 \right) \left( \frac{n - 1 - x}{x + 1} \cdot \frac{p}{1 - p} + 1 \right) \right] \\
			=& p(1 - p) \left( 2 \log 2 - \EE \left[ \log \left( \frac{X}{n - X} \cdot \frac{1 - p}{p} + 1 \right) \left( \frac{n - 1 - X}{X + 1} \cdot \frac{p}{1 - p} + 1 \right) \right] \right) \\
			=& p(1 - p) \left( 2 \log 2 - \EE \left[ \log \left( \frac{\alpha X}{\bar{X} + 1} + 1 \right) \right] - \EE \left[ \log \left( \frac{\beta \bar{X}}{X + 1} + 1 \right) \right] \right),
		\end{aligned}
		\label{eq:Ber_ICIMI}
	\end{equation}
	where we let $X \sim \Binomial(n - 1, p)$, $\bar{X} = n - 1 - X \sim \Binomial(n - 1, 1 - p)$, $\alpha = \frac{1-p}{p}$ and $\beta = \frac{p}{1-p}$. Without loss of generality, we assume $p < \frac{1}{2}$ and $\alpha > 1$, $\beta < 1$. The above expression can be approximated as
	\begin{equation}
		\begin{aligned}
			I\left( W; R_i | \hat{Z}_{i}, \hat{Z}_{i+n} \right) =& p(1 - p) \left( 2 \log 2 - \EE \left[ \log \frac{(\alpha - 1) X + n}{\bar{X} + 1} + \log \frac{n - (1 - \beta) \bar{X}}{X + 1} \right] \right) \\
			=& p(1 - p) \left( 2 \log 2 - \EE \left[ \log (\alpha - 1) + \log \left( X + \frac{n}{\alpha - 1} \right) - \log \left( \bar{X} + 1 \right)  \right.\right. \\
			&\left.\left.  + \log (1 - \beta) + \log \left( \frac{n}{1 - \beta} - \bar{X} \right) - \log \left( X + 1 \right) \right] \right) \\
			\overset{(a)}{=}& p(1 - p) \left( 2 \log 2 - \log \left( (\alpha - 1) (1 - \beta) \right) - \log \left( (n - 1) p + \frac{n}{\alpha - 1} \right) + \frac{(n - 1) (1 - p) p}{2 \left( (n - 1) p + \frac{n}{\alpha - 1} \right)^2} \right. \\
			&\left. + \log \left( (n - 1) (1 - p) + 1 \right) - \frac{(n - 1) (1 - p) p}{2 \left( (n - 1) (1 - p) + 1 \right)^2} \right. \\
			&\left. - \log \left( \frac{n}{1 - \beta} -  (n - 1) (1 - p) \right) + \frac{(n - 1) (1 - p) p}{2 \left( \frac{n}{1 - \beta} -  (n - 1) (1 - p) \right)^2} \right. \\
			&\left. + \log \left( (n - 1) p + 1 \right) - \frac{(n - 1) (1 - p) p}{2 \left( (n - 1) p + 1 \right)^2} + O\left( \frac{1}{n^2} \right) \right) \\
			=& p(1 - p) \left( 2 \log 2 - \log \left( (\alpha - 1) (1 - \beta) \right) + \log \frac{(n - 1) (1 - p) + 1}{(n - 1) p + \frac{n}{\alpha - 1}} + \log \frac{(n - 1) p + 1}{\frac{n}{1 - \beta} -  (n - 1) (1 - p)} + O\left( \frac{1}{n} \right) \right) \\
			\overset{(b)}{=}& p(1 - p) \left( 2 \log 2 - \log \left( (\alpha - 1) (1 - \beta) \right) + \log \frac{1 - p}{p + \frac{1}{\alpha - 1}} + \log \frac{p}{\frac{1}{1 - \beta} - (1 - p)} + O\left( \frac{1}{n} \right) \right) \\
			=& p(1 - p) \left( 2 \log 2 + \log \frac{1 - p}{p (\alpha - 1) + 1} + \frac{p}{1 - (1 - p) (1 - \beta)} + O\left( \frac{1}{n} \right) \right) \\
			=& p(1 - p) \left( 2 \log 2 + \log \frac{1 - p}{2 - 2 p} + \log \frac{p}{2 p} + O\left( \frac{1}{n} \right) \right) \\
			=& O\left( \frac{1}{n} \right).
		\end{aligned}
	\end{equation}
	where $(a)$ is obtained by using Lemma~\ref{lem:exp_bin_logx} and Lemma~\ref{lem:exp_bin_log-x} to each summand, $(b)$ is due to Lemma~\ref{lem:log_frac}. After taking the square root of $I( W; R_i | \hat{Z}_{i}, \hat{Z}_{i+n} )$, one can find that the ICIMI bound is of order $O\left( \frac{1}{\sqrt{n}} \right)$.
	
	\subsection{L$m$O-CMI Case} \label{subapp:cal_LmO}
	\label{subsubsec: LmO-case}
	Generalized from the LOO case, the hypothesis in the L$m$O case is $W = \frac{1}{n} \sum_{i=1}^{n}\hat{Z}_{U_i}$. Consider realizations $w = \frac{1}{n} \sum_{i=1}^{n}\hat{z}_{u_i} = \frac{x}{n}$ and $\sum_{i=1}^{n+m} \hat{z}_i - \sum_{i=1}^{n}\hat{z}_{u_i} = y$, then the joint probability is
	\begin{equation}
		\begin{aligned}
			P_{W, \hat{Z}_{[n+m]}, U_{[n]}}(\frac{x}{n}, \hat{z}_{[n+m]}, u_{[n]}) &= P_{W | \hat{Z}_{[n+m]}, U_{[n]}}(\frac{x}{n} | \hat{z}_{[n+m]}, u_{[n]}) P_{U_{[n]}}(u_{[n]}) \prod_{i=1}^{n+m} P_{\hat{Z}_i}(\hat{z}_i)\\
			&= \frac{1}{\begin{psmallmatrix} n+m \\ n \end{psmallmatrix}} p^{x + y} (1 - p)^{n + m - x - y}.
		\end{aligned}
	\end{equation}
	For the conditional probability $P_{W | \hat{Z}_{[n+m]}}(\frac{x}{n} | \hat{z}_{[n+m]})$, we have
	\begin{equation}
		\begin{aligned}
			P_{W | \hat{Z}_{[n+m]}}(\frac{x}{n} | \hat{z}_{[n+m]}) &= \sum_{u'_{[n]}} P_{W | \hat{Z}_{[n+m]}, U_{[n]}}(\frac{x}{n} | \hat{z}_{[n+m]}, u'_{[n]}) P_{U_{[n]}}(u'_{[n]})\\
			&= \frac{1}{\begin{psmallmatrix} n+m \\ n \end{psmallmatrix}} \sum_{u'_{[n]}} P_{W | \hat{Z}_{[n+m]}, U_{[n]}}(\frac{x}{n} | \hat{z}_{[n+m]}, u'_{[n]}).
		\end{aligned}
	\end{equation}
	We now consider the number of feasible $u_{[n]}$'s such that $P_{W | \hat{Z}_{[n+m]}, U_{[n]}}(\frac{x}{n} | \hat{z}_{[n+m]}, u_{[n]}) = 1$, which means $\sum_{i=1}^{n}\hat{z}_{u_i} = x$ is still satisfied. For clarity, we say that $\hat{z}_i = 1$ is a ``successful trial'' while the opposite case, i.e., $\hat{z}_i = 0$, is referred to as a ``failed trial''. Consider that there are $x+y$ successful trials and $n+m-x-y$ failed trials among all $n+m$ trials, then a feasible $u_{[n]}$ should choose $x$ successful trials, which contributes to $\big(\begin{smallmatrix} x+y \\ x \end{smallmatrix}\big)$ free choices, along with $n - x$ failed trials, which provides $\big(\begin{smallmatrix} n+m-x-y \\ n-x \end{smallmatrix}\big)$ choices. Thus,
	\begin{equation}
		\begin{aligned}
			P_{W | \hat{Z}_{[n+m]}}(\frac{x}{n} | \hat{z}_{[n+m]}) &= \frac{\begin{psmallmatrix} x+y \\ x \end{psmallmatrix} \begin{psmallmatrix} n+m-x-y \\ n-x \end{psmallmatrix}}{\begin{psmallmatrix} n+m \\ n \end{psmallmatrix}}.
		\end{aligned}
	\end{equation}
	Then, the mutual information is given by
	\begin{equation}
		\begin{aligned}
			I\left( W; U_{[n]} | \hat{Z}_{[n+m]} \right) &= \sum_{w, \hat{z}_{[n+m]}, u_{[n]}} P_{W, \hat{Z}_{[n+m]}, U_{[n]}}(w, \hat{z}_{[n+m]}, u_{[n]}) \log \frac{P_{W | \hat{Z}_{[n+m]}, U_{[n]}}(w | \hat{z}_{[n+m]}, u_{[n]})}{P_{W | \hat{Z}_{[n+m]}}(w | \hat{z}_{[n+m]})} \\
			&= \sum_{x=0}^{n} \sum_{y=0}^{m} \frac{\begin{psmallmatrix} x+y \\ x \end{psmallmatrix} \begin{psmallmatrix} n+m-x-y \\ n-x \end{psmallmatrix}}{\begin{psmallmatrix} n+m \\ n \end{psmallmatrix}} B_{n+m, x+y}(p) \log \frac{\begin{psmallmatrix} n+m \\ n \end{psmallmatrix}}{\begin{psmallmatrix} x+y \\ x \end{psmallmatrix} \begin{psmallmatrix} n+m-x-y \\ n-x \end{psmallmatrix}} \\
			&= \log \begin{psmallmatrix} n + m \\ n \end{psmallmatrix} - \sum_{x=0}^{n} \sum_{y=0}^{m} B_{n, x}(p) B_{m, y}(p) \log \begin{psmallmatrix} x+y \\ x \end{psmallmatrix} \begin{psmallmatrix} n+m-x-y \\ n-x \end{psmallmatrix}.
		\end{aligned}
	\end{equation}
	Let $X$ and $Y$ be two binomial random variables following $X \sim \Binomial(n, p)$ and $Y \sim \Binomial(m, p)$, respectively, then the above mutual information can be rewritten as
	\begin{equation}
		\begin{aligned}
			I\left( W; U_{[n]} | \hat{Z}_{[n+m]} \right) = \log \begin{psmallmatrix} n + m \\ n \end{psmallmatrix} - \EE \left[ \log \begin{psmallmatrix} X+Y \\ X \end{psmallmatrix} \begin{psmallmatrix} \bar{X} + \bar{Y} \\ \bar{X} \end{psmallmatrix} \right],
		\end{aligned}
		\label{eq:Ber_lmo}
	\end{equation}
	where $\bar{X} = n - X$ and $\bar{Y} = m - Y$. Note that one can verify that \eqref{eq:Ber_lmo} will reduce to \eqref{eq:Ber_LOO} with $m = 1$, while the detailed calculations are omitted.

	\subsection{LOFO-CMI Cases} \label{subapp:cal_LOFO}
	We now evaluate the $(m, m)$-IPCIMI and LOFO-CMI bounds. The hypothesis is $W = \frac{1}{n} \sum_{i=1}^{m} \left( \sum_{j=1}^{\frac{n}{m} + 1} \hat{Z}^{(i)}_j - \hat{Z}^{(i)}_{U^{(i)}} \right)$. Let $w = \frac{x + y}{n}$, $\hat{z}^{(i)}_{u^{(i)}} = t$ and $\sum_{j=1}^{\frac{n}{m} + 1} \hat{z}^{(i)}_j - t = y$, where $x \in [0, n - \frac{n}{m}]$, $y \in [0, \frac{n}{m}]$ and $t \in \{0, 1\}$. Then,
	\begin{equation}
		\begin{aligned}
			P_{W | U^{(i)}, \hat{Z}_{[\frac{n}{m}+1]}^{(i)}}(\frac{x + y}{n} | u^{(i)}, \hat{z}_{[\frac{n}{m}+1]}^{(i)}) = B_{n - \frac{n}{m}, x}(p).
		\end{aligned}
	\end{equation}
	\begin{equation}
		\begin{aligned}
			P_{W, U^{(i)}, \hat{Z}_{[\frac{n}{m}+1]}^{(i)}}(\frac{x + y}{n}, u^{(i)}, \hat{z}_{[\frac{n}{m}+1]}^{(i)}) = \frac{1}{\frac{n}{m} + 1} \begin{psmallmatrix} n - \frac{n}{m} \\ x \end{psmallmatrix} p^{x+y} (1 - p)^{n - x - y} p^t (1 - p)^t.
		\end{aligned}
	\end{equation}
	\begin{equation}
		\begin{aligned}
			P_{W | \hat{Z}_{[\frac{n}{m}+1]}^{(i)}}(\frac{x + y}{n} | \hat{z}_{[\frac{n}{m}+1]}^{(i)}) =& \sum_{u'} P_{W | U^{(i)}, \hat{Z}_{[\frac{n}{m}+1]}^{(i)}}(\frac{x + y}{n} | u', \hat{z}_{[\frac{n}{m}+1]}^{(i)}) P_{U^{(i)}}(u') \\
			=& \frac{y + t}{\frac{n}{m} + 1} B_{n - \frac{n}{m}, x + 1 - t}(p) + \frac{\frac{n}{m} + 1 - y - t}{\frac{n}{m} + 1} B_{n - \frac{n}{m}, x - t}(p).
		\end{aligned}
	\end{equation}
	Using the above expressions, we calculate the mutual information as
	\begin{equation}
		\begin{aligned}
			I&\left( W; U^{(i)} | \hat{Z}_{[\frac{n}{m}+1]}^{(i)} \right)\\
			=& \sum_{x=0}^{n - \frac{n}{m}} \sum_{y=0}^{\frac{n}{m}} \frac{y + 1}{\frac{n}{m} + 1} \begin{psmallmatrix} \frac{n}{m} + 1 \\ y + 1 \end{psmallmatrix} \begin{psmallmatrix} n - \frac{n}{m} \\ x \end{psmallmatrix} p^{x+y} (1 - p)^{n - x - y} p \log \frac{\frac{n}{m} + 1}{y + 1 + \frac{1 - p}{p} \cdot \frac{(\frac{n}{m} - y) x}{n - \frac{n}{m} + 1 - x}} \\
			& + \sum_{x=0}^{n - \frac{n}{m}} \sum_{y=0}^{\frac{n}{m}} \frac{\frac{n}{m} + 1 - y}{\frac{n}{m} + 1} \begin{psmallmatrix} \frac{n}{m} + 1 \\ y \end{psmallmatrix} \begin{psmallmatrix} n - \frac{n}{m} \\ x \end{psmallmatrix} p^{x+y} (1 - p)^{n - x - y} (1 - p) \log \frac{\frac{n}{m} + 1}{\frac{p}{1 - p} \cdot \frac{(n - \frac{n}{m} - x) y}{x + 1} + \frac{n}{m} + 1 - y} \\
			=& \sum_{x=0}^{n - \frac{n}{m}} \sum_{y=0}^{\frac{n}{m}} B_{\frac{n}{m}, y}(p) B_{n - \frac{n}{m}, x}(p) \left( p \log \frac{\frac{n}{m} + 1}{y + 1 + \frac{1 - p}{p} \cdot \frac{(\frac{n}{m} - y) x}{n - \frac{n}{m} + 1 - x}} + (1 - p) \log \frac{\frac{n}{m} + 1}{\frac{p}{1 - p} \cdot \frac{(n - \frac{n}{m} - x) y}{x + 1} + \frac{n}{m} + 1 - y} \right),
		\end{aligned}
	\end{equation}
	which can be rewritten as
	\begin{equation}
		\begin{aligned}
			I\left( W; U^{(i)} | \hat{Z}_{[\frac{n}{m}+1]}^{(i)} \right) =& \log \left( \frac{n}{m} + 1 \right) - p \EE \left[ \log \left( Y + 1 + \frac{\alpha X \bar{Y}}{\bar{X} + 1} \right) \right]- (1 - p) \EE \left[ \log \left( \bar{Y} + 1 + \frac{\beta \bar{X} Y}{X + 1} \right) \right],
		\end{aligned}
		\label{eq:Ber_LOFO}
	\end{equation}
	where $X \sim \Binomial\left( n - \frac{n}{m}, p \right)$, $Y \sim \Binomial\left( \frac{n}{m}, p \right)$, $\bar{X} = n - \frac{n}{m} - X \sim \Binomial\left( n - \frac{n}{m}, 1 - p \right)$, $\bar{Y} = \frac{n}{m} - Y \sim \Binomial\left( \frac{n}{m}, 1 - p \right)$, $\alpha = \frac{1 - p}{p}$ and $\beta = \frac{p}{1 - p}$. Through the above expression, one can validate the argument that the $(m, m)$-IPCIMI bound is a general case of the LOO-CMI and ICIMI bounds. With detailed calculations omitted, we note that letting $m = 1$ reduces \eqref{eq:Ber_LOFO} to the LOO-CMI bound in \eqref{eq:Ber_LOO}, whereas setting $m = n$ leads to the ICIMI bound shown in \eqref{eq:Ber_ICIMI}. Furthermore, it can be checked that \eqref{eq:Ber_LOFO} also holds for the LOFO-CMI bound.

	\subsection{$(m, n)$-IPCIMI Case} \label{subapp:cal_mn_IPCIMI}
	Following a similar procedure in Appendix~\ref{subapp:cal_LOFO}, we calculate $I( W; U^{(i)} | \hat{Z}_{[\frac{m}{n}+1]}^{(i)} )$ as
	\begin{equation}
		\begin{aligned}
			I\left( W; U^{(i)} | \hat{Z}_{[\frac{m}{n}+1]}^{(i)} \right) =& \log \left( \frac{m}{n} + 1 \right) - p \EE \left[ \log \left( Y + 1 + \frac{\beta \bar{X} \bar{Y}}{X + 1} \right) \right] - (1 - p) \EE \left[ \log \left( \bar{Y} + 1 + \frac{\alpha X Y}{\bar{X} + 1} \right) \right],
		\end{aligned}
	\end{equation}
	where $X \sim \Binomial\left( n - 1, p \right)$, $Y \sim \Binomial\left( \frac{m}{n}, p \right)$, $\bar{X} = n - 1 - X \sim \Binomial\left( n - 1, 1 - p \right)$, $\bar{Y} = \frac{m}{n} - Y \sim \Binomial\left( \frac{m}{n}, 1 - p \right)$, $\alpha = \frac{1 - p}{p}$ and $\beta = \frac{p}{1 - p}$.
	
	\subsection{SICIMI Case}  \label{subapp:cal_SICIMI}
	Let $w = \frac{x}{n}$, $r_{[n]} = (r_1, \ldots, r_n)$ and $\hat{z}_{[2n]} = (\hat{z}_1, \ldots, \hat{z}_{2n})$ be the realizations, where $x \in [0, n]$. We first consider the case that $r_i = 0$, i.e., $\hat{Z}_i$ is selected as the training sample. Thus, the joint probability $P_{W, R_i, \hat{Z}_{i}}(\frac{x}{n}, r_i, \hat{z}_{i})$ can be calculated by marginalizing $\hat{Z}_{i+n}$ out of \eqref{eq:ICIMI_P_W_Ri_Zi_Zi+n}:
	\begin{equation}
		\begin{aligned}
			P_{W, R_i, \hat{Z}_{i}}(\frac{x}{n}, 0, \hat{z}_{i}) =& \sum_{\hat{z}_{i+n} \in \{0,1\}} P_{W, R_i, \hat{Z}_{i}, \hat{Z}_{i+n}}(\frac{x}{n}, 0, \hat{z}_{i}, \hat{z}_{i+n}) \\
			=& \frac{1}{2} B_{n-1, x-\hat{z}_{i}}(p) p^{\hat{z}_{i}} (1-p)^{1-\hat{z}_{i}}.
		\end{aligned}\label{eq:P_W_R_Z_0}
	\end{equation}
	For the opposite case $r_i = 1$, by marginalization we obtain
	\begin{equation}
		\begin{aligned}
			P_{W, R_i, \hat{Z}_{i}}(\frac{x}{n}, 1, \hat{z}_{i}) =& \frac{1}{2} B_{n,x}(p) p^{\hat{z}_{i}} (1-p)^{1-\hat{z}_{i}}.
		\end{aligned}\label{eq:P_W_R_Z_1}
	\end{equation}
	The above two expressions further yield
	\begin{equation}
		\begin{aligned}
			P_{W, \hat{Z}_{i}}(\frac{x}{n}, \hat{z}_{i}) =& P_{W, R_i, \hat{Z}_{i}}(\frac{x}{n}, 0, \hat{z}_{i}) + P_{W, R_i, \hat{Z}_{i}}(\frac{x}{n}, 1, \hat{z}_{i}) \\ 
			=& \frac{1}{2} B_{n-1, x-\hat{z}_{i}}(p) p^{\hat{z}_{i}} (1-p)^{1-\hat{z}_{i}} + \frac{1}{2} B_{n,x}(p) p^{\hat{z}_{i}} (1-p)^{1-\hat{z}_{i}}.
		\end{aligned}\label{eq:P_W_Z}
	\end{equation}
	Then, we expand the mutual information by its definition:
	\begin{equation}
		\begin{aligned}
			I\left( W; R_i | \hat{Z}_{i} \right) =& \sum_{w, r_i, \hat{z}_{i}} P_{W, R_i, \hat{Z}_{i}}(w, r_i, \hat{z}_{i}) \log \frac{P_{W, R_i, \hat{Z}_{i}}(w, r_i, \hat{z}_{i})}{P_{W, \hat{Z}_{i}}(w, \hat{z}_{i}) P_{R_i}(r_i)} \\
			=& \sum_{w} P_{W, R_i, \hat{Z}_{i}}(w, 0, 0) \log \frac{2 P_{W, R_i, \hat{Z}_{i}}(w, 0, 0)}{P_{W, \hat{Z}_{i}}(w, 0)} + \sum_{w} P_{W, R_i, \hat{Z}_{i}}(w, 0, 1) \log \frac{2 P_{W, R_i, \hat{Z}_{i}}(w, 0, 1)}{P_{W, \hat{Z}_{i}}(w, 1)} \\
			& + \sum_{w} P_{W, R_i, \hat{Z}_{i}}(w, 1, 0) \log \frac{2 P_{W, R_i, \hat{Z}_{i}}(w, 1, 0)}{P_{W, \hat{Z}_{i}}(w, 0)} + \sum_{w} P_{W, R_i, \hat{Z}_{i}}(w, 1, 1) \log \frac{2 P_{W, R_i, \hat{Z}_{i}}(w, 1, 1)}{P_{W, \hat{Z}_{i}}(w, 1)}.
		\end{aligned}
	\end{equation}
	By combining the elemental probabilities obtained in \eqref{eq:P_W_R_Z_0}, \eqref{eq:P_W_R_Z_1} and \eqref{eq:P_W_Z}, we have
	\begin{equation}
		\begin{aligned}
			I\left( W; R_i | \hat{Z}_{i} \right) =& \sum_{x=0}^{n-1} \frac{1-p}{2} B_{n-1,x}(p) \log \frac{2 \begin{psmallmatrix} n - 1 \\ x \end{psmallmatrix}}{\begin{psmallmatrix} n - 1 \\ x \end{psmallmatrix} + \begin{psmallmatrix} n \\ x \end{psmallmatrix} (1-p)} + \sum_{x=1}^{n} \frac{p}{2} B_{n-1,x-1}(p) \log \frac{2 \begin{psmallmatrix} n - 1 \\ x - 1 \end{psmallmatrix}}{\begin{psmallmatrix} n - 1 \\ x - 1 \end{psmallmatrix} + \begin{psmallmatrix} n \\ x \end{psmallmatrix} p} \\
			& + \sum_{x=0}^{n} \frac{1-p}{2} B_{n,x}(p) \log \frac{2 \begin{psmallmatrix} n \\ x \end{psmallmatrix} (1-p)}{\begin{psmallmatrix} n - 1 \\ x \end{psmallmatrix} + \begin{psmallmatrix} n \\ x \end{psmallmatrix} (1-p)} + \sum_{x=0}^{n} \frac{p}{2} B_{n,x}(p) \log \frac{2 \begin{psmallmatrix} n \\ x \end{psmallmatrix} p}{\begin{psmallmatrix} n - 1 \\ x - 1 \end{psmallmatrix} + \begin{psmallmatrix} n \\ x \end{psmallmatrix} p}.
		\end{aligned}
	\end{equation}
	Let $X \sim \Binomial(n, p)$, $\bar{X} = n - X \sim \Binomial(n, 1 - p)$, $Y \sim \Binomial(n - 1, p)$, $\bar{Y} = n - 1 - Y \sim \Binomial(n - 1, 1 - p)$, we can finally get the simplified expression as
	\begin{equation}
		\begin{aligned}
			I\left( W; R_i | \hat{Z}_{i} \right) =& \log 2 - \frac{p}{2} \left(\EE \left[ \log \left( 1 + \frac{X}{n p} \right) + \log \left( 1 + \frac{n p}{Y + 1} \right) \right] \right) \\
			&- \frac{1 - p}{2} \left( \EE \left[ \log \left( 1 + \frac{\bar{X}}{n (1 - p)} \right) + \log \left( 1 + \frac{n (1 - p)}{\bar{Y} + 1} \right) \right] \right).
		\end{aligned}
		\label{eq:Ber_SICIMI}
	\end{equation}
	
	\subsection{LOO-SCMI Case} \label{subapp:cal_LOO_SCMI}
	Let $w = \frac{x}{n}$, $U = u$ and $\hat{z}_{[n+1]} = (\hat{z}_1, \ldots, \hat{z}_{n+1})$ be the realizations, where $x \in [0, n]$. Analogous to the calculations in the case of SICIMI, for any $i \in [n+1]$ we calculate $P_{W, U, \hat{Z}_{i}}$ as
	\begin{equation}
		\begin{aligned}
			P_{W, U, \hat{Z}_{i}}(\frac{x}{n}, u, \hat{z}_{i}) =& \begin{cases}
				\frac{1}{n+1} B_{n-1,x-\hat{z}_{i}}(p) p^{\hat{z}_{i}} (1-p)^{1-\hat{z}_{i}}, & \textnormal{if $u \not= i$} \\
				\frac{1}{n+1} B_{n,x}(p) p^{\hat{z}_{i}} (1-p)^{1-\hat{z}_{i}}, & \textnormal{if $u = i$}
			\end{cases},
		\end{aligned}
	\end{equation}
	which further yields
	\begin{equation}
		\begin{aligned}
			P_{W, \hat{Z}_{i}}(\frac{x}{n}, \hat{z}_{i}) =& \sum_{u=1}^{n+1} P_{W, U, \hat{Z}_{i}}(\frac{x}{n}, u, \hat{z}_{i}) \\
			=& \left( \frac{n}{n+1} B_{n-1,x-\hat{z}_{i}}(p) + \frac{1}{n+1} B_{n,x}(p) \right) p^{\hat{z}_{i}} (1-p)^{1-\hat{z}_{i}}.
		\end{aligned}
	\end{equation}
	Then, we expand the mutual information by its definition:
	\begin{equation}
		\begin{aligned}
			I\left( W; U | \hat{Z}_{i} \right) =& \sum_{w, u, \hat{z}_{i}} P_{W, U, \hat{Z}_{i}}(w, u, \hat{z}_{i}) \log \frac{P_{W, U, \hat{Z}_{i}}(w, u, \hat{z}_{i})}{P_{W, \hat{Z}_{i}}(w, \hat{z}_{i}) P_{U}(u)} \\
			=& \sum_{w} \sum_{u \in [n+1], u \not= i} P_{W, U, \hat{Z}_{i}}(w, u, 0) \log \frac{P_{W, U, \hat{Z}_{i}}(w, u, 0)}{P_{W, \hat{Z}_{i}}(w, 0)} + \sum_{w} P_{W, U, \hat{Z}_{i}}(w, i, 0) \log \frac{P_{W, U, \hat{Z}_{i}}(w, i, 0)}{P_{W, \hat{Z}_{i}}(w, 0)} \\
			& + \sum_{w} \sum_{u \in [n+1], u \not= i} P_{W, U, \hat{Z}_{i}}(w, u, 1) \log \frac{P_{W, U, \hat{Z}_{i}}(w, u, 1)}{P_{W, \hat{Z}_{i}}(w, 1)} + \sum_{w} P_{W, U, \hat{Z}_{i}}(w, i, 1) \log \frac{P_{W, U, \hat{Z}_{i}}(w, i, 1)}{P_{W, \hat{Z}_{i}}(w, 1)},
		\end{aligned}
	\end{equation}
	Plugging $P_{W, U, \hat{Z}_{i}}$ and $P_{W, \hat{Z}_{i}}$ into the above expression, with some additional rearrangements and simplifications, we finally get
	\begin{equation}
		\begin{aligned}
			I\left( W; U | \hat{Z}_{1} \right) =& \log (n + 1) - \frac{p}{n + 1} \EE \left[ \log \left( 1 + \frac{X}{p} \right) + n \log \left( n + \frac{n p}{Y + 1} \right) \right] \\
			&- \frac{1 - p}{n + 1} \EE \left[ \log \left( 1 + \frac{\bar{X}}{1 - p} \right) + n \log \left( n + \frac{n (1 - p)}{\bar{Y} + 1} \right) \right],
		\end{aligned}
		\label{eq:Ber_LOO_SCMI}
	\end{equation}
	where $X \sim \Binomial(n, p)$, $\bar{X} = n - X \sim \Binomial(n, 1 - p)$, $Y \sim \Binomial(n - 1, p)$, $\bar{Y} = n - 1 - Y \sim \Binomial(n - 1, 1 - p)$.
	
	\subsection{L$m$O-SCMI Case} \label{subapp:cal_LmO_SCMI}
	Let $w = \frac{x}{n}$, $U_{[n]} = u$ and $\hat{z}_{[n+m]} = (\hat{z}_1, \ldots, \hat{z}_{n+m})$ be the realizations, where $x \in [0, n]$. Analogous to the calculations in the case of SICIMI, for any $i \in [n+m]$ we calculate $P_{W, U_{[n]}, \hat{Z}_{i}}$ as
	\begin{equation}
		\begin{aligned}
			P_{W, U_{[n]}, \hat{Z}_{i}}(\frac{x}{n}, u, \hat{z}_{i}) =& \begin{cases}
				\frac{1}{\begin{psmallmatrix} n + m \\ n \end{psmallmatrix}} B_{n-1,x-\hat{z}_{i}}(p) p^{\hat{z}_{i}} (1-p)^{1-\hat{z}_{i}}, & \textnormal{if $i \in u$} \\
				\frac{1}{\begin{psmallmatrix} n + m \\ n \end{psmallmatrix}} B_{n,x}(p) p^{\hat{z}_{i}} (1-p)^{1-\hat{z}_{i}}, & \textnormal{if $i \not\in u$}
			\end{cases},
		\end{aligned}
	\end{equation}
	which further yields
	\begin{equation}
		\begin{aligned}
			P_{W, \hat{Z}_{i}}(\frac{x}{n}, \hat{z}_{i}) =& \sum_{u: i \in u} P_{W, U, \hat{Z}_{i}}(\frac{x}{n}, u, \hat{z}_{i}) + \sum_{u: i \not\in u} P_{W, U, \hat{Z}_{i}}(\frac{x}{n}, u, \hat{z}_{i}) \\
			=& \frac{\begin{psmallmatrix} n + m - 1 \\ n - 1 \end{psmallmatrix}}{\begin{psmallmatrix} n + m \\ n \end{psmallmatrix}} B_{n-1,x-\hat{z}_{i}}(p) p^{\hat{z}_{i}} (1-p)^{1-\hat{z}_{i}} + \frac{\begin{psmallmatrix} n + m - 1 \\ n \end{psmallmatrix}}{\begin{psmallmatrix} n + m \\ n \end{psmallmatrix}} B_{n,x}(p) p^{\hat{z}_{i}} (1-p)^{1-\hat{z}_{i}} \\
			=& \left( \frac{n}{n+m} B_{n-1,x-\hat{z}_{i}}(p) + \frac{m}{n+m} B_{n,x}(p) \right) p^{\hat{z}_{i}} (1-p)^{1-\hat{z}_{i}}.
		\end{aligned}
	\end{equation}
	Plugging these probabilities into the CMI given by
	\begin{equation}
		\begin{aligned}
			I\left( W; U_{[n]} | \hat{Z}_{i} \right) =& \sum_{w, u, \hat{z}_{i}} P_{W, U_{[n]}, \hat{Z}_{i}}(w, u, \hat{z}_{i}) \log \frac{P_{W, U_{[n]}, \hat{Z}_{i}}(w, u, \hat{z}_{i})}{P_{W, \hat{Z}_{i}}(w, \hat{z}_{i}) P_{U_{[n]}}(u)} \\
			=& \sum_{w} \sum_{u: i \in u} P_{W, U_{[n]}, \hat{Z}_{i}}(w, u, 0) \log \frac{P_{W, U_{[n]}, \hat{Z}_{i}}(w, u, 0)}{P_{W, \hat{Z}_{i}}(w, 0)} \\
			& + \sum_{w} \sum_{u: i \not\in u} P_{W, U_{[n]}, \hat{Z}_{i}}(w, u, 0) \log \frac{P_{W, U_{[n]}, \hat{Z}_{i}}(w, i, 0)}{P_{W, \hat{Z}_{i}}(w, 0)} \\
			& + \sum_{w} \sum_{u: i \in u} P_{W, U_{[n]}, \hat{Z}_{i}}(w, u, 1) \log \frac{P_{W, U_{[n]}, \hat{Z}_{i}}(w, u, 1)}{P_{W, \hat{Z}_{i}}(w, 1)} \\
			& + \sum_{w} \sum_{u: i \not\in u} P_{W, U_{[n]}, \hat{Z}_{i}}(w, u, 1) \log \frac{P_{W, U_{[n]}, \hat{Z}_{i}}(w, i, 1)}{P_{W, \hat{Z}_{i}}(w, 1)}.
		\end{aligned}
	\end{equation}
	With some additional simplifications, we have
	\begin{equation}
		\begin{aligned}
			I\left( W; U_{[n]} | \hat{Z}_{i} \right) =& \log (n + m) - p \EE \left[ \frac{m}{n + m} \log \left( m + \frac{X}{p} \right) + \frac{n}{n + m} \log \left( n + \frac{n m p}{Y + 1} \right) \right] \\
			&- (1 - p) \EE \left[ \frac{m}{n + m} \log \left( m + \frac{\bar{X}}{1 - p} \right) + \frac{n}{n + m} \log \left( n + \frac{n m (1 - p)}{\bar{Y} + 1} \right) \right],
		\end{aligned}
		\label{eq:Ber_LmO_SCMI}
	\end{equation}
	where we let $X \sim \Binomial(n, p)$, $\bar{X} = n - X \sim \Binomial(n, 1 - p)$, $Y \sim \Binomial(n - 1, p)$, $\bar{Y} = n - 1 - Y \sim \Binomial(n - 1, 1 - p)$.
	
	\subsection{LOFO-SCMI Case} \label{subapp:cal_LOFO_SCMI}
	Since the ERM algorithm adopted in Example~\ref{eg:Bernoulli} satisfies Assumption~\ref{ass:invariant-alg}, then Proposition~\ref{pro:invariant-SCMI} implies that L$m$O-SCMI and LOFO-SCMI bounds are identical. Therefore, the calculations for the LOFO-SCMI bound are omitted.
\end{document}